\documentclass[11pt,english]{article}
\pdfoutput=1

\usepackage{mathptmx}

\usepackage[T1]{fontenc}
\usepackage[utf8]{inputenc}
\usepackage[letterpaper]{geometry}
\usepackage{babel}
\usepackage{verbatim}
\usepackage{float}
\usepackage{calc}
\usepackage{amsmath}
\usepackage{amsthm}
\usepackage{amssymb}
\usepackage{graphicx}
\usepackage{setspace}
\usepackage{pbox}
\usepackage{footmisc}
\renewcommand{\thefootnote}{\fnsymbol{footnote}}
\renewcommand{\thetable}{\Roman{table}} 
\renewcommand{\thefigure}{\Roman{figure}} 
\newcommand{\tabitem}{~~\hspace*{5pt}\llap{\textbullet}~~} 

\usepackage{titletoc} 
\usepackage[hidelinks]{hyperref} 

\newcounter{daggerfootnote}
\newcommand*{\daggerfootnote}[1]{%
    \setcounter{daggerfootnote}{\value{footnote}}%
    \renewcommand*{\thefootnote}{\fnsymbol{footnote}}%
    \footnote[2]{#1}%
    \setcounter{footnote}{\value{daggerfootnote}}%
    \renewcommand*{\thefootnote}{\arabic{footnote}}%
    }

\usepackage[longnamesfirst]{natbib}
\shortcites{Leeetal2019,Cardetal2015,buja2009statistical,osf,aearct,Bondetal2014,Garciaetal2017,Boletal2016,Cattaneoetal2019,Chettyetal2011,Chettyetal2014,Hahnetal1999,Chapmanetal2018,Gusaketal2010,li2020essential,montiel2019competing,banerjee2020theory,Calonicoetal2019regression,Peietal2018local,Peietal2018local,brown2021meta,kortingetal2020wp,Cardetal2015ECMA}

\makeatletter

\providecommand{\tabularnewline}{\\}

\theoremstyle{plain}
\newtheorem{assumption}{\protect\assumptionname}
\theoremstyle{plain}
\newtheorem{prop}{\protect\propositionname}

\extrafloats{150}

\usepackage{multicol}
\usepackage{setspace}

\usepackage{amsbsy}
\usepackage{amsthm}

\usepackage{pgfplots}
\pgfplotsset{compat=1.16}
\usepackage{adjustbox}

\usepackage{pdflscape}

\usepackage{makecell}

\usepackage{booktabs}
\usepackage[labelsep=newline, justification=centering]{caption} 

\@ifundefined{showcaptionsetup}{}{%
 \PassOptionsToPackage{caption=false}{subfig}}
\usepackage{subfig}
\makeatother

\providecommand{\assumptionname}{Assumption}
\providecommand{\propositionname}{Proposition}

\begin{document}
\pagestyle{plain}

\thispagestyle{empty}

\setcounter{footnote}{0}

\begin{titlepage}
\noindent \begin{center}
~\\
{\Large{}VISUAL INFERENCE AND GRAPHICAL REPRESENTATION IN }\\
{\Large{}REGRESSION DISCONTINUITY DESIGNS\footnote[1]{We are grateful for the insightful and constructive comments from two co-editors, Larry Katz and Andrei Shleifer, and four anonymous reviewers. We have also benefited from discussions with Alberto Abadie, Sahara Byrne, Colin Camerer, Matias Cattaneo, Damon Clark, Geoff Fisher, Paul Goldsmith-Pinkham, Nathan Grawe, Jessica Hullman, David Lee, Lars Lefgren, Thomas Lemieux, Pauline Leung, Jia Li, Adam Loy, Alex Mas, Doug Miller, Ted O'Donoghue, Bitsy Perlman, Steve Pischke, Jonathan Roth, Jesse Rothstein, Rocio Titiunik, Cindy Xiong, Stephanie Wang, Andrea Weber, and Xiaoyang Ye, as well as participants of various seminars and conferences. Lexin Cai, Matt Comey, Michael Daly, Rebecca Jackson, Motasem Kalaji, Xingyue Li, Fiona Qiu, and Tatiana Velasco provided research assistance, and we thank Brad Turner, Mary Ross, and Patti Tracey for providing logistical support. We are indebted to our friends for testing the experiment and to colleagues for participating in the study. We gratefully acknowledge financial support from the Cornell Institute of Social Sciences and the Princeton Industrial Relations Section. The authors have no relevant or material financial interests that relate to the research of this paper. This study is registered in the Open Science Framework and the AEA RCT Registry with ID AEARCTR-0004331. Any opinions and conclusions expressed herein are those of the authors and do not reflect the views of the U.S.\ Census Bureau.}}\\
\par\end{center}

\begin{singlespace}
\begin{center}
{\large{}Christina Korting, University of Delaware}\\
{\large{}Carl Lieberman, U.S. Census Bureau}\\
{\large{}Jordan Matsudaira, Columbia University}\\
{\large{}Zhuan Pei\daggerfootnote{Corresponding author. Associate Professor, Department of Economics and Jeb E. Brooks School of Public Policy, Cornell University, Ithaca, NY 14853, USA; zhuan.pei@cornell.edu}, Cornell University and IZA}\\
{\large{}Yi Shen, University of Waterloo}\\
\end{center}
\end{singlespace}

\begin{center}
{\large{}January 2023}{\large\par}
\par\end{center}

\begin{abstract}
\begin{singlespace}

Despite the widespread use of graphs in empirical research, little is known about readers' ability to process the statistical information they are meant to convey (``visual inference''). We study visual inference within the context of regression discontinuity (RD) designs by measuring how accurately readers identify discontinuities in graphs produced from data generating processes calibrated on 11 published papers from leading economics journals. First, we assess the effects of different graphical representation methods on visual inference using randomized experiments. We find that bin widths and fit lines have the largest impacts on whether participants correctly perceive the presence or absence of a discontinuity. Our experimental results allow us to make evidence-based recommendations to practitioners, and we suggest using small bins with no fit lines as a starting point to construct RD graphs. Second, we compare visual inference on graphs constructed using our preferred method with widely used econometric inference procedures. We find that visual inference achieves similar or lower type I error (false positive) rates and complements econometric inference.
\bigskip{}

Key Words: Graphical Methods; Visual Inference; Regression Discontinuity Design; Expert Prediction; Scientific Communication\bigskip{}

JEL Code: A11, C10, C40

\end{singlespace}
\end{abstract}

\end{titlepage}
\renewcommand{\thefootnote}{\arabic{footnote}}
\newpage{}

\setcounter{footnote}{0}

\setlength{\abovedisplayskip}{3pt}

\setlength{\belowdisplayskip}{3pt}

\setlength{\abovedisplayshortskip}{3pt}

\setlength{\belowdisplayshortskip}{3pt}

\section{Introduction\label{sec:Introduction}}
\begin{center}
``Few would deny that the most powerful statistical tool is graph
paper.''\\
--- Geoffrey S. \citet{watson1964smooth}
\par\end{center}

Graphical analysis is increasingly prevalent in empirical research, a phenomenon \citet{currie2020technology} call the ``graphical revolution.'' Effective use of graphs conveys a large set of statistical information at once and improves research transparency (\citealp{andrews2020transparency}). However, there are different ways to construct a graph with the same data, and the particular construction an analyst chooses has the potential to mislead readers (\citealp{schwartzstein2021using}). To understand the best use of graphical evidence, it is important to study readers' ability to process information from graphs---which we term ``visual statistical inference'' or \emph{visual inference} per \citet{Majumderetal2013}---as well as the sensitivity of visual inference to choices in graph construction. To date, little is known about visual inference for commonly presented graphs in empirical research designs.

We begin to fill this knowledge gap and study visual inference in the regression discontinuity design (RDD or RD design). The popularity of RDD in the modern causal inference toolkit, which began in economics with \citet{AngristandLavy1999}, makes it an important setting in which to study visual inference. Standard practices in applying RDD today perhaps best embody the spirit of Watson's quote above, with graphs playing a central role in the presentation of findings. The key RD graph plots the bivariate relationship between outcome variable $Y$ and running variable $X$ and is meant to display a discontinuity (or lack thereof) in the underlying conditional expectation function (CEF) as $X$ crosses a policy threshold. Influential practitioner guides by \citet{ImbensandLemieux2008},  \citet{LeeandLemieux2010}, and \citet{cattaneo2019practical} recommend creating this graph by dividing $X$ into bins, computing the average of $Y$ within each bin, and generating a scatter plot of these $Y$-averages against the midpoints of the bins.

We assess the performance of visual inference by studying whether people presented with this graph can accurately extract the embedded statistical information, where our main criterion is the correct identification of the existence or absence of a discontinuity at the policy threshold. Our project has two major components. In the first, we build on work pioneered by \citet{eells1926relative} and refined by \citet{ClevelandandMcGill1984b} by conducting a series of randomized experiments to examine how different graphical parameters affect visual inference in RD. We present participants recruited through the Cornell University Johnson College's Business Simulation Lab with RD graphs produced from data generating processes (DGPs) based on microdata from 11 published papers that we randomly select from a list of 110 empirical studies from top economics journals. For each graph, we ask participants to identify the existence or absence of a discontinuity. We randomize respondents into different treatment arms and show participants within each arm graphs produced with particular graphical parameters such as small bin widths and evenly spaced bins.

There is limited research on how to choose these parameters in practice. For example, \citet{Calonicoetal2015} propose two popular data-driven bin width selectors: one that minimizes the integrated mean squared error (IMSE) of the bin averages, resulting in fewer, larger bins, and another that mimics the variability of the underlying data (mimicking variance or MV), which leads to more, smaller bins. While both proposals set a graphical parameter to satisfy an econometric criterion, practitioners are left with little basis to choose between them. Moreover, a host of other choices over graphical parameters remains with minimal guidance from the literature, such as including smoothed regression lines in the binned scatter plot, adding a vertical line to indicate the policy threshold, and choosing the axis scales.

By comparing the rates at which respondents correctly classify discontinuities across treatment arms, we can assess the advantages and disadvantages of different graphical parameters. We find that certain graphical parameters such as bin width and smoothed regression lines create important tradeoffs between type I errors (identifying a discontinuity when there is none, i.e., a false positive) and type II errors (identifying the absence of a discontinuity when there is one, i.e., a false negative). Relative to MV (small) bins, using IMSE (large) bins tends to increase type I error rates but decrease type II error rates. Similarly, imposing fit lines may also increase type I error rates, echoing the concerns of \citet{ct2021nber,cattaneo2021regression}.

To translate our findings to recommendations on graphical practices in RDDs, we empirically implement a decision theoretic framework that incorporates classification accuracy as a metric to compare graphical methods. The method that uses MV bins with no fit lines consistently performs well relative to IMSE bins and/or imposing fit lines. Bin spacing (equally spaced versus quantile-spaced), axis scaling, and the presence of a vertical line indicating the policy threshold do not appear to matter, implying that researchers can adhere to reasonable personal preferences. Our recommendation is robust to alternative decision theoretic framework formulations. 

Because only non-experts participate in our randomized experiments on the effects of graphical methods, one may be concerned that our results are less relevant to academic audiences. To assess whether our findings generalize across experience levels, we also recruit experts from a pool of seminar attendees and affiliates of the National Bureau of Economic Research (NBER) and the Institute of Labor Economics (IZA) to participate in our study. Although our expert sample is not large enough to conduct the same randomized experiments, we can compare the non-expert and expert results by using the subset of non-experts who saw the same graphs as the experts. We find that the two groups perform comparably. 

In the spirit of \citet{DellaVignaandPope2018}, we also test whether experts are able to predict the graphical parameters that result in the highest rate of visual inference success by non-experts. We find that experts only partly anticipate the aforementioned effects of bin widths and fit lines.

As a second major component of the project, we compare the performances of visual inference and econometric inference. For visual inference, we use results from the sample of experts who viewed graphs constructed with the best-performing technique from our experiments. For econometric inference, we apply three influential methods by \citet{ImbensKalyanaraman2012}, \citet{Calonicoetal2014}, and \citet{ArmstrongandKolesar2017} (henceforth IK, CCT, and AK, respectively) and conduct hypothesis testing at the 5\% (asymptotic) level. We find that visual inference achieves a type I error rate that, at just below 8\%, is lower than the IK and CCT procedures (the CCT type I error rate is not significantly higher), but the two econometric procedures enjoy considerably lower type II error rates. Visual inference performs very similarly to the AK procedure, a remarkable result given the minimax optimality property of AK.    

We also find that visual and econometric inferences appear to be complementary. First, we examine the joint distribution of visual and econometric tests: while they commit similar type II errors, there does not appear to be a strong association in their type I errors. Second, we assess the performance of a combined visual and econometric inference. One simple way of combining the two inferences mirrors an approach in which a researcher believes a discontinuity exists if and only if a formal test rejects the null hypothesis of no effect \emph{and} she sees a discontinuity in the RD graph. We find that the combined IK and visual inference performs similarly to the AK procedure, which may help explain the enduring credibility of the RD design despite formal inference issues in earlier RD papers.

Finally, we ask experts to estimate the discontinuity magnitude when they classify a discontinuity, and we compare the accuracy of their estimates to that of econometric methods. On this front, econometric methods tend to do better. For example, the simple local linear IK estimator yields lower mean squared errors than experts across all 11 DGPs, shedding light on the limits of visual inference.

This paper connects a diverse set of literatures and makes the following contributions. First, we begin to fill an important gap in our understanding of graphical evidence by evaluating visual inference and graphical representation practices in a widely used quasi-experimental research design. Our endeavor draws from three strands of the statistics literature that study the choice of graphical parameters (e.g., \citealp{Calonicoetal2015}, \citealp{li2020essential}), their effects on visual inference (e.g., \citealp{ClevelandandMcGill1984b}), and the evaluation of visual inference through comparison with econometric inference (e.g., \citealp{Majumderetal2013}). Our paradigm can be applied to other important areas, as discussed in Section \ref{sec:Conclusion}.

Second, to guide our study design and to help interpret our empirical results, we propose a general conceptual framework, which may extend to future studies of visual inference in other contexts. In particular, we can interpret the average type I or II error rate we use as an estimate of the probability that a randomly sampled reader commits such an error when viewing a graph generated from a randomly chosen DGP. We show that these error probabilities are key inputs in a standard decision theoretic framework, which helps inform best graphical practices. We also discuss alternative decision theoretic formulations and make connections to the recent literature on scientific communication (e.g., \citealt{andrews2020model}).

Third, we add to the literature on expert judgments (e.g., \citealp{CamererandJohnson1997}) and expert forecasts of research results (e.g., \citealp{sanders2015just}). Our finding that experts only partly anticipate our experimental results underscores the value of empirically evaluating visual inference and providing evidence-based guidance on graphical methods.

We introduce the conceptual framework in Section \ref{sec:conceptual-framework}, describe the design of our experiments and studies in Section \ref{sec:overall-description}, present results in Section \ref{sec:Results}, and conclude in Section \ref{sec:Conclusion}. For readers in a hurry, the key takeaway results are in Figures \ref{fig:RDD-Power} and \ref{fig:RDD-Power-Experts-vs-All} with corresponding discussions in Sections \ref{subsec:non-expert-results} and \ref{sec:Visual-vs-Econometric}.

\section{Conceptual Framework\label{sec:conceptual-framework}}

In this section, we propose a conceptual framework for evaluating visual inference to guide our study design and aid in the interpretation of our empirical results. More specifically, we show how to aggregate visual inference performances across subjects, who may reach different conclusions even when viewing the same graph, and how to meaningfully interpret the parameter to which our aggregate measure corresponds. This performance measure helps to inform best graphical practices. 

Although RD graphs may serve other purposes, we view their most important function as accurately conveying discontinuity existence and magnitude at the policy threshold. According to \citet{LeeandLemieux2010}, other purposes of RD graphs include i) helping to assess regression specifications and ii) allowing for the inspection of discontinuities away from the policy cutoff. But ultimately, these other functions are also motivated by inference on the discontinuity at the policy threshold: i) can be viewed as reconciling visual and econometric inferences thereof and ii) informs the reader, under implicit global homogeneity assumptions, whether to believe the existence of a discontinuity at the policy threshold. 

We focus on binary classifications of a graph and treat type I and type II errors as the main performance measures for visual inference. The conceptual framework easily generalizes to assessing visual estimates of the discontinuity magnitude, which we elicit from experts. A person commits a type I error in RD visual inference if she classifies a continuous graph as having a discontinuity and a type II error if she classifies a discontinuous graph as continuous.

To define our aggregate measures of visual inference performance, we introduce the following notations. First, the vector $\gamma$ denotes a combination of graphical parameters (see \citealp{wilkinson2013grammar} for an extensive list). We study five parameters in this paper (bin width, bin spacing, axis scaling, polynomial fit lines, and a vertical line at the policy threshold), and each of the five entries of $\gamma$ represents the value of a particular parameter.\footnote{In our experiments, we randomly assign each participant to view only graphs generated with a certain fixed value of $\gamma$. Ideally, one could run a large experiment with a full factorial design to test all combinations of the graphical values outlined above, but resource constraints force us to test a subset of the graphical parameter space via a sequence of studies as described in Section \ref{sec:primary-experiment-general}.}

The combination $(g,d)$ denotes the probability model underlying an RD dataset. $g$ encompasses four elements: i) the distribution of the running variable $X$; ii) the conditional expectation function $E[\tilde{Y}|X=x]$ which is continuous at the policy threshold $x=0$; iii) the distribution of the error term $u$ where $\tilde{Y}\equiv E[\tilde{Y}|X=x]+u$; and iv) the sample size $N$. Intuitively, $g$ specifies everything in the probability model except for the discontinuity, including the shape of the conditional expectation function. The discontinuity then results from shifting the right arm of the smooth function $E[\tilde{Y}|X=x]$ by some discontinuity level $d$, that is, $Y=\tilde{Y}+d\cdot1_{[X\geqslant0]}$ (which implies that $E[Y|X=x]=E[\tilde{Y}|X=x]+d\cdot1_{[X\geqslant0]}$; we provide a graphical illustration in Section \ref{sec:creation-datasets-graphs}). We note that the variable $Y$ can represent the outcome, baseline covariates, or treatment take-up, so this framework applies to all graphs typically included in RD studies, including those from a fuzzy design. Typically, the $(g,d)$ combination is jointly referred to as the ``data generating process,'' but we separate the discontinuity level $d$ and call $g$ the DGP for ease of exposition below. 

We think of each (bivariate) RD dataset as a realization from the probability model $(g,d)$, which we denote by $W$, or $W(g,d)$ if we want to emphasize the underlying probability model. Implementing a graphical procedure with parameters $\gamma$ on dataset $W$ results in an RD graph $(\gamma,W)$, which we denote by $T$ or $T(\gamma,g,d)$. Alternatively, we can think of $T$ as a realization from $(\gamma,g,d)$ and refer to $(\gamma,g,d)$ as the graph generating process (GGP).

When presented with the same RD graph, readers may draw different visual inferences. For example, some readers may be more skilled than others at classifying a discontinuity because they have received more training in statistics, have more experience with RD graphs, or otherwise have superior ability. We use $\phi$ to capture these human characteristics that affect graph perception.

The probability that a reader with characteristics $\phi$ reports that a discontinuity exists in RD graph $T(\gamma,g,d)$ is denoted by $\tilde{p}(T(\gamma,g,d),\phi)$. From casual observation, we know that the same reader may be influenced by idiosyncratic elements not encapsulated in $\phi$ and classify the same graph differently on different days. The probability formulation $\tilde{p}$ allows these factors to affect visual inference.

We now define the type I and type II error probabilities we use to gauge reader performance. First, averaging $\tilde{p}$ over both data realizations $W$ and reader characteristics $\phi$ leads to the quantity 
\[
p(\gamma,g,d)\equiv E_{W,\phi}[\tilde{p}(T(\gamma,g,d),\phi)].
\]
This is the probability that a randomly chosen reader reports a discontinuity in a graph randomly generated from the GGP $(\gamma,g,d)$. A high value of $p$ indicates a high classification error probability when the true discontinuity $d$ is zero (type I error), but a low classification error probability when $d$ is nonzero (type II error). Formally, the DGP-specific or $g$-specific type I and type II error probabilities for graphical parameters $\gamma$ are defined as:
\begin{align*}
g\text{-specific }\text{type I error probability: } & p(\gamma,g,0)\\
g\text{-specific }\text{type II error probability: } & 1-p(\gamma,g,d)\text{ for }d\neq0.
\end{align*}
Conceptually, we can further average $p(\gamma,g,d)$ over the space of DGPs, $\mathcal{G}$ (we discuss $\mathcal{G}$ after Assumption \ref{assu:Assum_iid} below) to arrive at the \emph{overall} discontinuity classification probability for $\gamma$:
\[
\bar{p}(\gamma,d)\equiv E_{g\in\mathcal{G}}[p(\gamma,g,d)].
\]
Correspondingly, the overall type I and type II error probabilities are defined as
\begin{align*}
\text{overall type I error probability: } & \bar{p}(\gamma,0)\\
\text{overall type II error probability: } & 1-\bar{p}(\gamma,d)\text{ for }d\neq0.
\end{align*}
Consistent with the definitions in \citet[p.~382]{casella2021statistical}, we call $p(\gamma,g,d)$ and $\bar{p}(\gamma,d)$ \emph{power functions} as functions of $d$.

In this paper, we design experiments to estimate the type I and type II error probabilities as defined above. For each GGP $(\gamma,g,d)$, we generate $M$ different realized graphs and present each to a random participant. That is, participant $i$ is shown one RD graph denoted by $T_i(\gamma,g,d)$, where $i$ takes on values in the set $\{1,...,M\}$, and is asked to assess the presence of a discontinuity. Let the binary variable $R_{i}(T_i(\gamma,g,d))$ denote participant $i$'s discontinuity classification, which equals one if the participant reports a discontinuity at the policy threshold. Under random sampling, the following assumption holds:
\begin{assumption}
\label{assu:Assum_iid}For a given GGP $(\gamma,g,d)$, the $R_{i}(T_{i}(\gamma,g,d))$'s
are i.i.d.\ with $E[R_{i}(T_{i}(\gamma,g,d))]=p(\gamma,g,d).$
\end{assumption}
A natural estimator for $p(\gamma,g,d)$ is the sample average of discontinuity classifications:
\[
\hat{p}(\gamma,g,d)=\frac{1}{M}\sum_{i}R_{i}(T_{i}(\gamma,g,d)).
\]
Proposition \ref{prop:g} in Online Appendix \ref{sec:Proofs} states the distribution of $\hat{p}(\gamma,g,d)$ and shows the estimator to be unbiased and consistent as $M\to\infty$ for $p(\gamma,g,d)$ under Assumption \ref{assu:Assum_iid}.

To estimate the overall probability $\bar{p}(\gamma,d)$, we need to sample from the DGP space $\mathcal{G}$, which we formally define in Online Appendix \ref{subsec:DGP-space}. While the infinite dimensionality of $\mathcal{G}$ makes it difficult to characterize the distribution of DGPs, we think of the data used in empirical RD research as realizations when sampling from $\mathcal{G}$ according to this distribution. To that end, we can specify $J$ DGPs that approximate data from existing research and present graphs generated with discontinuity \textit{$d$} for each DGP $g_{j}$ ($j=1,...,J$) to a distinct group of $M$ participants for a total of $M\cdot J$ participants and visual discontinuity classifications.
\begin{assumption}
\label{assu:Assum_g}The DGP $g_{j}$'s are randomly sampled from
$\mathcal{G}$.
\end{assumption}
A natural estimator for $\bar{p}(\gamma,d)$ is
\[
\hat{\bar{p}}(\gamma,d)\equiv\frac{1}{J}\sum_{j}\hat{p}(\gamma,g_{j},d)=\frac{1}{M\cdot J}\sum_{i,j}R_{i}(T_{i}(\gamma,g_{j},d)),
\]
the average of discontinuity classifications across the $M\cdot J$ classifications. Proposition \ref{prop:overall} in Online Appendix \ref{sec:Proofs} states the distribution of $\hat{\bar{p}}(\gamma,d)$ and shows the estimator to be unbiased and consistent (as $J\to\infty$) for $\bar{p}(\gamma,d)$ under Assumptions \ref{assu:Assum_iid} and \ref{assu:Assum_g} (given that $J=11$ in our experiments, consistency here is a conceptual statement implying that were we to incorporate DGPs from more RD studies, our estimators would be closer in probability to the population parameters of interest). We henceforth refer to $\hat{p}(\gamma,g,0)$ as the DGP-specific or $g$-specific type I error rate, and $\hat{\bar{p}}(\gamma,0)$ as the average type I error rate (or simply the type I error rate). For a particular $d\neq0$, we refer to $1-\hat{p}(\gamma,g,d)$ as the DGP-specific or $g$-specific type II error rate, and $1-\hat{\bar{p}}(\gamma,d)$ as the average type II error rate (or simply the type II error rate) at discontinuity $d$.

We can also define the type I and type II error probabilities of an econometric inference procedure based on a discontinuity estimator, $\hat{\theta}$, but with two adjustments. First, $\gamma$ is no longer an argument in these expressions because we directly implement $\hat{\theta}$ on microdata $W$. Second, we need to specify the level of the testing procedure, which we set to 5\%, the prevailing standard in empirical studies. Because the definitions of these probabilities and their estimators are similar to the quantities defined above, we omit them here.

In subsequent sections, we empirically trace out $\hat{\bar{p}}(\gamma,d)$ as functions of $d$, which concisely summarize the type I and type II error probabilities of a graphical method. For brevity, we also use the term ``power functions,'' as opposed to ``estimated power functions,'' to refer to their empirical estimates. We study visual inference by comparing its power functions $\hat{\bar{p}}(\gamma,d)$ across $\gamma$ and against the corresponding power functions of various econometric inference procedures. We discuss the calculation of the standard errors on the differences between the visual and econometric power functions when we present our empirical results in Section \ref{sec:Results}, as its details depend on the design of our experiments.

The estimated power functions help shed light on best graphical practices. While the ``optimal'' graphical method depends in part on the optimality criterion we specify, the power functions provide key inputs for some common criteria that are well suited to this context. A simple criterion, as suggested by a referee, is how close the type I error rate is to 5 percent, the conventional threshold in econometric analysis. A second criterion quantifies the costs of type I and type II errors and uses the power function to consider explicitly the tradeoff between the two types of errors. We also consider a third criterion, which is adapted from \citet{andrews2020model}; instead of using the power function, this criterion relies on readers' confidence in their discontinuity classification, which we can proxy using participants' bonus scheme choice as discussed in Section \ref{sec:primary-experiment-general} and Online Appendix \ref{sec:primary-experiment-bonus-scheme}. A full discussion of the second and third criteria require formal decision-theoretic frameworks, and we leave it to Online Appendix \ref{subsec:Decision-theory}. 

In summary, we have defined the type I and type II error rates for visual inference. The type I error rate is the fraction of continuous graphs participants incorrectly classify as having a discontinuity, and the type II error rate is the fraction of discontinuous graphs incorrectly classified as being continuous. Our framework allows us to interpret these rates as unbiased and consistent estimates of the probabilities of type I and type II errors a randomly chosen person commits when classifying a graph generated from a representative DGP. These probabilities also help inform best graphical practices.

\section{Description of Experiments and Studies\label{sec:overall-description}}

\subsection{Graphical Parameters Tested\label{subsec:Graphical-Parameters-Tested}}

We test the effects of bin width, bin spacing, parametric fit lines, vertical lines at the policy threshold, and $y$-axis scaling. We discuss each of these treatments in detail below and provide graphical illustrations in Figure \ref{fig:rdd_treatments}.

The most studied graphical parameter in RD is the width of each bin in the binned scatter plot. The first class of bin width selection algorithms comes from \citet{LeeandLemieux2010}: start with some number of bins, double that number, test whether the additional bins fit the data significantly better, and repeat until the test fails to reject the null hypothesis. \citet{Calonicoetal2015} propose two bin width selection algorithms based on different econometric criteria. The first, which is more in line with the convention of the nonparametric regression literature, is the bin selector that minimizes the IMSE of the bin-average estimator of the CEF, where the resulting number of bins increases with the sample size $N$ at the rate $N^{1/3}$. For their second bin selector, the MV selector, \citet{Calonicoetal2015} state that they ``choose the number of bins so that the binned sample means have an asymptotic (integrated) variability approximately equal to the amount of variability of the raw data.'' The resulting number of bins increases with the sample size more quickly, at the rate $N/\log(N)^{2}$. The IMSE-optimal bin width therefore selects fewer bins, and hence has larger bin widths, than the MV algorithm (we describe these algorithms further in Online Appendix \ref{sec:Bin-algorithms}). In addition to these two algorithms, \citet{Calonicoetal2015} provide an interpretation for any given number of bins as the output of a weighted IMSE-optimal algorithm. Applying the \citet{LeeandLemieux2010} algorithms to our datasets leads to bin numbers that tend to be between the IMSE and MV selectors and closer to those of the IMSE. Thus, we restrict our analysis to the visual inference properties of the IMSE and MV bin selectors.

Although the prevailing approach is to adopt evenly spaced bins, this method has drawbacks in that the resulting bins may contain vastly different numbers of observations, or even none at all.\footnote{In a literature review we conduct for current practices, 98\% of the more than 100 studies we compile use evenly spaced bins. Our review includes RDD studies as well as studies that apply the regression kink design---RK design or RKD.} This can happen when the distribution of the running variable is far from uniform. As a remedy, \citet{Calonicoetal2015} also propose quantile-spaced bins where each bin contains (approximately) the same number of observations. Both spacings support IMSE and MV bin selectors, and we test each of these combinations.

Following suggestions by \citet{ImbensandLemieux2008} and \citet{LeeandLemieux2010}, RD graphs frequently feature parametric fit lines and a vertical line at the treatment threshold. Both papers suggest fit lines improve ``visual clarity'' by approximating the conditional expectation functions, and the default in the popular \texttt{rdplot} command by \citet{Calonicoetal2015} uses piecewise global quartic regressions on each side of the policy threshold. Of the 11 RDD papers on which we calibrate our DGPs, 10 include fit lines. For the six of these papers that generate fit lines using polynomial regressions, we use the same polynomial order as in the source graph, and in the remaining cases, our team unanimously decides on the fit that best matches the original fit line or the data. We could also use a formal data-driven approach to select the polynomial order, but different criteria from \citet{LeeandLemieux2010} lead to conflicting recommendations (for example, for our first DGP, the Akaike information criterion selects a fifth-order polynomial, while the Bayesian information criteria and $F$-test they describe choose a zeroth-order polynomial). 

The presence of a vertical line is designed to visually separate observations above and below the cutoff. We test all combinations of the fit line and vertical line treatments except for including the fit lines and excluding the vertical line, which is used infrequently in practice (our literature review shows that fewer than 10\% of papers use this combination).

The motivation for rescaling graph axes comes from \citet{Clevelandetal1982}, who note that correlations on scatter plots seem stronger when scales are increased. We use two axis scaling options in our experiments. First is the default output returned by Stata 14. Second, we double that range by recording the range of the $y$-variable from the default graph and increasing the bounds by 50\% of the original range in each direction, resulting in a graph where the data are condensed along the vertical axis. We do not manipulate the scale of the $x$-axis because our survey of the literature suggests that this is not a common adjustment. A related decision researchers encounter is the range of the running variable to use in producing graphs: should they use the entire dataset or only a subsample close to the policy threshold? We do not test this margin of adjustment in our experiments due to the difficulty in generalizing the findings from such an exercise. Suppose we find that selecting 50\% of the observations closest to the threshold improves visual inference, should researchers ``chop'' the sample they are already planning to use? And after doing so, will it be beneficial to chop again? One could argue for testing the effect of using the full sample versus the subsample falling within the IK or CCT bandwidth, but these bandwidths themselves depend on the full sample---the first step in bandwidth calculation is a (semi-)global regression---and it is not even clear how we should define the ``full'' sample: some of the replication data used for our DGP calibration are already subsets of a larger sample.

There are other graphical parameters we do not test in our experiments. One is plotting confidence bands around the binned averages or fit lines. However, confidence intervals are too complex to explain to the non-experts in our short tutorial without potentially affecting the way participants think about the classification task itself, and therefore we do not experimentally test their effects on visual inference. Because the moderate size of our expert pool makes it unsuited to randomized experiments, we refrain from testing other graphical parameters.

\subsection{Creation of Simulated Datasets and Graphs\label{sec:creation-datasets-graphs}}

We specify data generating processes based on
the actual data used in published research. We randomly sample 11 from a total of 110 empirical RD papers published in the \emph{American Economic Review}, \emph{American Economic Journals}, \emph{Econometrica}, \emph{Journal of Business and Economic Statistics, Journal of Political Economy, Quarterly Journal of Economics, Review of Economic Studies,} and \emph{Review of Economics and Statistics} between 1999 and 2017 that have replication data available to create our DGPs. We refer to these DGPs as DGP1-DGP11.

The calibration of each DGP $g$ entails the specification of its four components: the distribution of the running variable $X$, the continuous conditional expectation function $E[\tilde{Y}|X=x]$, the distribution of the error term $u$, and the sample size $N$.\footnote{Out of the 11 studies, two plot residualized outcomes instead of the original $Y$ to adjust for covariates, conceptually consistent with the covariate-adjusted RD estimation by \citet{Calonicoetal2019regression} and the ideas of \citet{AngristandRokkanen2015wanna}.} We use the empirical distribution of the running variable from each of the 11 papers but normalize it to lie in $[-1,1]$ by dividing $X$ of each data point by the maximum $|X|$ on its side of the cutoff, with zero representing the location of that policy cutoff. We also remove the most extreme observations where $|X|>0.99$, following \citet{ImbensKalyanaraman2012} and \citet{Calonicoetal2014}. For two papers which feature semi-discrete running variables, we add small amounts of normal noise to the running variable to match the regularity conditions from \citet{Calonicoetal2015}. To create a continuous CEF, we fit global piecewise quintics (still following \citealp{ImbensKalyanaraman2012} and \citealp{Calonicoetal2014}) and vertically shift the right arm. We specify the distribution of the error term $u$ as i.i.d.\ normal with mean zero and standard deviation $\sigma$, which we set as the root mean squared error (RMSE) of the piecewise quintic regression. We use the same number of observations as the original paper minus any observations removed while trimming the data. Plots of the resulting CEFs before we vertically shift their right arm to make them continuous are in Figure \ref{fig:RDD-CEF}, and Figure \ref{fig:dgp_construction} illustrates the construction process. We describe the DGP creation process in full detail in Online Appendix \ref{sec:creation-datasets-graphs-appendix}.

Because the outcomes from the 11 papers are measured in different units, we need to standardize the discontinuity levels and choose to specify discontinuity levels $d$ as multiples of $\sigma$. Alternatively, we could specify $d$ as multiples of the overall standard deviation of the outcome variable net of the discontinuity, a measure that also captures the variation due to the conditional expectation function besides the error term. Using this alternative measure turns out not to make a difference: the variance of the error term, $\sigma^2$, dominates the variance of $E[Y|X]$ in all of our DGPs, with the ratio of the two ranging from 8 to 690.

As a multiple of $\sigma$, $d$ takes on 11 values: $0,\pm0.1944\sigma,\pm0.324\sigma,\pm0.54\sigma,\pm0.9\sigma,\pm1.5\sigma$. We choose the upper bound $|d|=1.5\sigma$ based on our own visual judgment: it represents the point at which we expect every reasonable person to say a graph from any of our 11 DGPs features a discontinuity. The nonzero magnitudes of $d$ are equally spaced on the log scale. We use this scale rather than a linear scale to generate more graphs with smaller discontinuities, which are harder to detect, to better capture the shape of the power functions.

Our discontinuity magnitudes are similarly distributed to those observed in the main outcome graphs in our literature review. The average absolute value of the discontinuity $t$-statistics in our datasets from piecewise quintic regressions is 5.0 with a standard deviation of 5.3, compared to the observed mean of 3.9 with a standard deviation of 6.4. If we instead compare the distributions of the absolute value of the discontinuity divided by the control magnitude (the left intercept of the CEF), the means are similar, 1.3 in our datasets and 1.9 in the field, while our standard deviations are somewhat smaller at 2.2 compared to 7.7.

As argued in Section \ref{sec:conceptual-framework}, we want our DGPs to be representative. Although we select the papers randomly, we also need to evaluate how well our DGPs approximate the actual data from the respective studies. To do this, we adapt the lineup protocol from \citet{buja2009statistical} and \citet{Majumderetal2013}, which uses visual inference to conduct hypothesis testing. In our case, we test the null hypothesis that the original datasets come from the calibrated DGPs. Specifically, we present one graph of the original data randomly placed among 19 graphs from datasets drawn from the corresponding DGP. The goal is to identify the true dataset by choosing the graph that least resembles the others. If the viewer does not select the original graph, then we cannot reject the null hypothesis. Under the null hypothesis, the probability of identifying the graph produced from the original data (or the type I error probability) among the 19 simulated datasets is 5\% ($1/20$) for a single reader. For our lineup protocol, each graph is a binned scatter plot using the MV bin selector. We present two examples in Figure \ref{fig:Lineup-Protocol-Example-1}.

Based on visual testing among the authors, we cannot identify the graph from the true data for eight out of our 11 DGPs, which supports the idea that our DGPs approximate the original datasets well. For three DGPs, however, there is an obvious difference, as exemplified by DGP3 in the right panel of Figure \ref{fig:Lineup-Protocol-Example-1}. All three ``fail'' seemingly because of the misspecification of the variance structure of the error term $u$. Recall that we specify $u$ as being i.i.d.\ across observations and homoskedastic. But in the right panel of Figure \ref{fig:Lineup-Protocol-Example-1}, for example, the running variable is time, and there is positive serial correlation in the outcome. As a consequence, the outcome variability in the binned scatter plot is understated when $u$ is assumed to be i.i.d. Nevertheless, we adhere to the i.i.d.\ specification because it is standard in Monte Carlo exercises to evaluate RD estimators and inference procedures.

Another caveat of our DGP specification is that using global quintic regressions can lead to overfitting, the same issue that has brought forth the warning by \citet{gelman2019high} against using high-order global polynomial regressions to estimate RD treatment effects (see \citealp{Peietal2018local} for related discussions on the order of local polynomial regressions). We acknowledge this potential drawback of using quintics as some of our graphs indeed feature high variation in the tails. That said, the lineup protocol we adapt offers a novel and transparent method to evaluate our DGP specifications, and we find our inability to distinguish the real data from those drawn from one of our DGPs in a majority of cases reassuring. To further assuage the concerns regarding our DGP specification, we carry out a supplemental phase of experiments to gauge the sensitivity of visual inference to alternative DGP specifications. In Online Appendix \ref{sec:dgp-robustness}, we demonstrate the remarkable robustness of our results to using local linear estimates as an alternative to model the CEF (and allowing for heteroskedasticity), which is much less likely to overfit. 

\subsection{Non-Expert Experiments\label{sec:primary-experiment}\label{sec:primary-experiment-general}}

In our randomized experiments, we present non-expert participants with binned scatter plots made from our DGPs and ask them to classify the graphs as having a discontinuity or not. We conduct five phases of computer-based experiments online through the Cornell University Johnson College's Business Simulation Lab. Our subject pool consists of current and former Cornell students, Cornell staff, and non-student local residents with an expressed interest in focus groups or surveys. Although these educated laypeople are not the primary audience for academic research, RD graphs are sufficiently transparent that they are featured in popular media articles in publications such as \textit{The New York Times}, \textit{The Washington Post}, and \textit{The Atlantic} \citep{Dynarski2014,Sides2015,Rosen2012}, suggesting the participants in our sample should be capable of interpreting the graphs.

Before the experiment, participants watch a video tutorial explaining how the graphs are constructed.\footnote{The video tutorial is available at \url{https://storage.googleapis.com/rd-video-turial/rd_video_tutorial.mp4}.} We do not instruct participants on how to make their decisions, e.g., whether only to look at points near the cutoff or mentally to trace out the CEF. The video contains an attention check with a corresponding question later in the experiment to ensure that subjects are attentive to the instructions. After the video, participants complete a series of interactive example tasks and receive feedback on their answers. As part of the instructions, we explicitly tell participants that all, some, or none of the 11 graphs they classify may feature a discontinuity.

In each phase of the experiment, we present participants with a series of RD graphs using data generated as described in Section \ref{sec:creation-datasets-graphs}. Participants see two graphs with zero discontinuities, one each of $\pm0.1944\sigma,\pm0.324\sigma,\pm0.54\sigma,\pm0.9\sigma$, and one of either $1.5\sigma$ or $-1.5\sigma$. Participants see one graph from all 11 DGPs in a randomized order. We have up to 88 participants per treatment arm, and every graph we generate is seen by only one participant. For each graph, we ask participants whether they believe there is a discontinuity at $x=0$. 

Because running an experiment with $2^{5}=32$ treatment arms is infeasible with our resources, we conduct our experiment in phases, testing only a few treatments in each phase. Table \ref{tab:rdd-details-main} details the timeline of the experiments and lists the graphical parameters we test and hold fixed for various experimental phases. In phase 1, we test both bin width and axis scaling options. In phase 2, we test bin widths and bin spacings. In phase 3, we test imposing fit lines and a vertical line at the treatment threshold. Based on the results from these three phases, in which only bin widths and fit lines have major impacts, phase 4 tests all four combinations of those two treatments together.

Participants receive a base pay of \$3 for being in the experiment. To stimulate participant engagement and elicit participants' confidence in their response, participants can choose for each graph they classify a bonus that is either based on a monetary wager which pays 40 cents if their judgment is correct but nothing otherwise or a fixed payment of 20 cents irrespective of their classification. In Online Appendix \ref{sec:primary-experiment-bonus-scheme}, we explore the implication of a participant's bonus choice and discuss how we use it---in addition to the type I and type II error rates---to evaluate graphical methods as mentioned in Section \ref{sec:conceptual-framework}.

We do not give participants real-time feedback on the accuracy of their responses. Instead, we report total earnings and the final tally of correct classifications at the very end of the experiment after a short exit survey soliciting demographic information and comments.

The experiments are programmed in oTree \citep{chen2016otree} and pre-registered at the AEA RCT registry \citep{aearct} and the Center for Open Science's OSF platform \citep{osf}. The study takes participants approximately 15 minutes to complete.

\subsection{Expert Study\label{subsec:expert-study}}

In addition to our non-expert experiment, we conduct a study with researchers in economics and related fields who work on topics that often employ RDDs. We collect data at three technical social science seminars and online by contacting randomly selected members of the NBER in applied microeconomic fields (aging, children, development, education, health, health care, industrial organization, labor, and public) and IZA fellows and affiliates. After removing six responses from participants who completed the survey more than once, did not provide a valid email address for payment, or were not part of our recruited sample, we are left with 143 expert responses.

This expert study allows us to answer two questions. First, how do classification accuracy and the impacts of graphical techniques differ between experts and non-experts? And second, can experts correctly predict which graphing options perform best for our non-expert sample? This second question speaks to experts' ability to predict which visualization choices are best suited for interpretation by a lay audience. Because the success of a graphical technique ultimately lies in the reader's correct perception of graphs using it, it is important to understand whether experts' intuition regarding the relative advantages and drawbacks of alternative representation choices aligns with the evidence we find in practice. In related work on experts' ability to predict non-expert performance, \citet{DellaVignaandPope2018} find that economic experts are better than non-experts at estimating the effect of alternative incentive schemes on performance in a real effort task, but perform similarly to non-experts in terms of a simple ranking of incentive schemes.

Our expert study consists of two parts. The first is similar in structure to the non-expert experiment. Participants see a series of RD graphs and are asked to classify them by whether they have a discontinuity. To assess the accuracy of point estimates in addition to binary classifications of discontinuities, we also ask participants for an estimate of the discontinuity magnitude whenever they report a discontinuity. Due to sample size limitations, we do not randomize graphical treatments in the expert study, and all participants see graphs with equally spaced bins, no fit lines, default axis scaling, and a vertical line at the treatment threshold. All expert graphs use small bins, except for one seminar where participants see large bins. Four randomly selected participants receive a base payment of \$450 plus a bonus payment of \$50 per correct discontinuity classification. The bonus payment does not depend on the accuracy of the magnitude estimate.

The second part of the expert study asks about experts' preferences and their beliefs regarding non-expert performance across alternative graphical parameters. We present experts with three discontinuity magnitudes: $0$, $0.54\sigma$, and $1.5\sigma$. At each magnitude, we present four graphs, one for each combination of bin width and fit lines, in a random order using the same underlying data from the DGP where visual inference performs most closely to the average across all 11 DGPs. At each magnitude, we ask the experts to indicate which of the four treatment options they prefer and which they believe perform best and worst in our non-expert sample. We evaluate the experts' predictions about non-expert performances using phase 4 of the non-expert experiment, which tests these four treatment permutations simultaneously.

\section{Results\label{sec:Results}}

\subsection{Non-Expert Experiment Results and Graphical Method Recommendation\label{subsec:non-expert-results}}

For each combination of graphical parameters in all phases of the experiment, we compute power functions based on participants' classifications of whether graphs feature a discontinuity. Using notation from Section \ref{sec:conceptual-framework}, a DGP-specific power function represents the estimates $\hat{p}(\gamma,g,d)$ for graphical parameters $\gamma$ and DGP $g$ across different levels of discontinuity $d$. An overall power function represents the estimates $\hat{\bar{p}}(\gamma,d)$.

The intercept of a power function indicates the type I error rate as defined in Section \ref{sec:conceptual-framework}. At all other discontinuity magnitudes, the power function represents the proportion of graphs with discontinuities that participants classify correctly, which can be interpreted as one minus the type II error rate when the DGP is chosen uniformly randomly from the 11 possibilities. A desirable inference method has a small intercept before quickly rising to achieve high power. Plots of the overall power functions for each phase are shown in Figure \ref{fig:RDD-Power}, and Figure \ref{fig:RDD-Power-Phase-By-DGP} plots the corresponding DGP-specific power functions. The $x$-axis in these graphs is the magnitude of the discontinuity divided by the DGP-specific $\sigma$. This normalization facilitates aggregation and comparisons across DGPs, which then have six identical discontinuity magnitudes: 0, 0.1944, 0.324, 0.54, 0.9, and 1.5. Estimated effects on visual inference are in Tables \ref{tab:TE-RD-I}-\ref{tab:TE-RD-VII}. Because we effectively adopt stratified randomization in the design of our experiments as described in Section \ref{sec:primary-experiment-general}, where each of the 11 strata is determined by the DGPs seen for every discontinuity magnitude, we obtain these estimates by regressing the participants' responses on treatment indicators and stratum fixed effects.

Phase 1 of the experiment tests the four combinations of the bin width treatments (large IMSE-optimal bins and small MV bins) with the $y$-axis scaling treatments (the default in Stata 14 and double that range). Comparing power functions, large bins have a significantly higher type I error rates relative to small bins. While both small bin treatments lead to a type I error rate of approximately 5\%, the large bins have a type I error rate of around 20\% to 25\%. The large bins have a type II error advantage over the small bins, which is related to their higher type I error rate, but the power functions converge as the discontinuity magnitude increases. In contrast to bin widths, axis scaling has little effect on participant perception. Based on this result, we use Stata's default scaling for all subsequent phases.

In phase 2, we again test the two bin width treatments, this time interacted with the two bin spacing treatments: even spacing and quantile spacing. Note that large bins and small bins with even spacing appear in both phases 1 and 2. With this design, we can gauge the stability of visual inference across different samples from the non-expert population, and it is encouraging to see the results for these two treatments being virtually identical across phases. Comparing the power functions for these repeated treatments with their new quantile-spaced versions, we see that evenly-spaced and quantile-spaced bins perform very similarly. We conclude that bin spacing has a small or null effect on visual inference for the DGPs we test.

Phase 3 tests three treatments: the inclusion of a vertical line at the treatment threshold with and without polynomial fit lines and the omission of both the vertical line and fit lines. We find that the vertical line at the cutoff makes little difference in perception. Fit lines, on the other hand, appear to increase type I error rates in this phase, in line with a common concern that they may be overly suggestive of discontinuities.

Jointly, these three phases of experiments suggest that the presence of fit lines and the bin width choice have the largest impact on visual perceptions of discontinuities. We therefore base our analysis of expert preferences and expert predictions about non-expert performance on the interaction of these two treatments and run a final phase of experiments, phase 4, directly comparing the four possible treatment combinations. Interestingly, while the effects of bin width choice are once again robust across phases, the effects of fit lines are more muted in this phase. In particular, the treatment with small bins and fit lines has a type I error rate of only 0.052 in phase 4 but 0.175 in phase 3. This finding suggests that we cannot conclude that fit lines unequivocally result in an increase in type I errors, but that they do add uncertainty to visual inference.

We provide additional results in Online Appendices \ref{sec:Appendix-Additional-RDD-Analysis} and \ref{sec:t-stat}. Online Appendix \ref{sec:Appendix-Additional-RDD-Analysis} includes evidence for the balance of covariates across treatment arms and the measurement of predictive power of demographic and DGP characteristic variables on visual inference performance. Online Appendix \ref{sec:t-stat} examines how large the $t$-statistics need to be for readers to detect a discontinuity visually.

Based on our experimental results, we recommend the graphical method that uses small bins, no fit lines, even spacing, default $y$-axis scaling, and a vertical line at the policy threshold as a sensible default for generating RD graphs. It has an associated type I error rate close to 5 percent in all four phases of the experiment. As we document in Online Appendix \ref{sec:Risk}, it also performs well under the two additional criteria mentioned in Section \ref{sec:conceptual-framework}. Of the five graphical parameters, bin spacing, $y$-axis scaling, and the presence of the vertical line do not appear to matter much, allowing researchers to use reasonable discretion. The other two parameters are much more important, and the use of small bins and no fit lines appears to be key for good visual inference performance.

Finally, we emphasize that our recommendation is not intended as a doctrine that practitioners must abide by. We are limited to the 11 DGPs we test. In addition, as described in Online Appendix \ref{sec:Bin-algorithms}, there is an ad hoc element to the construction of small bins. In fact, we use quantile-spaced large bins in the regression kink design experiments in a previous working paper \citet{kortingetal2020wp}, where the sample sizes are much larger than for our RD DGPs. In this regard, we view our recommended graphical method as a reasonable starting point based on the best evidence we have. An equally important takeaway is the value in documenting the robustness of graphical evidence given our finding of divergent visual inferences under commonly used methods.

\subsection{Expert Study Results}

We show most expert participants (95 out of 143) graphs generated with our preferred method as discussed above. Figure \ref{fig:Experts-Nonexperts-ROC} plots the expert power functions against those of the non-experts who saw the same graphs. When comparing expert and non-expert performances, we use solid (hollow) markers to indicate that the point is (not) statistically significantly different  from the reference curve, and present plots of the corresponding 95\% confidence intervals (created here with the large sample approximation described at the end of Online Appendix \ref{sec:Proofs} and by assuming independence between the experts and non-experts) in Figure \ref{fig:RDD-Experts-Nonexperts-Diff}. The two groups perform similarly, with experts having a slightly higher type I error rate (approximately 8\% to the non-expert 5\%) and a slightly lower type II error rate. The only statistically significant differences are for the experts' marginally lower type II error rates at the 0.1944 and 0.324 discontinuities.

In addition to the aforementioned treatment, we show experts in one seminar pool (48 out of 143) graphs using the large bins and no fit lines treatment. The two groups again perform similarly, and the expert and non-expert power functions are not statistically significantly different anywhere. Both groups have type I error rates well above their corresponding small bin rates. Like non-experts, experts do worse when viewing graphs constructed with large bins. 

\subsubsection{Expert Preferences and Predicting Non-Expert Performance}

We present experts' preferences and their beliefs about non-expert performance across the four considered treatments in Figure \ref{fig:RDD-Experts-Pred-Pref}. When asked about graphing options for the main graph of a paper that conveys the treatment effect, most experts report preferring small bins, usually with fit lines. These results hold at all three discontinuity magnitudes considered, including zero. Experts' predictions about the most effective treatments for non-experts tend to mirror their preferences. By a large margin, experts believe small bins with fit lines to be the most efficacious treatment for non-experts at all discontinuity magnitudes. Conversely, most experts view large bins without fit lines least favorably in the context of non-expert performance.

Comparing the expert predictions to our experimental data from phase 4, we find substantial discordance for the effects of bin width choice on non-expert classification accuracy. The best- and worst-performing treatments at each discontinuity magnitude have + and - signs, respectively, in Figure \ref{fig:RDD-Experts-Pred-Pref}. The actual power functions are shown in Figure \ref{fig:RDD-Power} (Figure \ref{fig:RDD-Power-Phase-VI-DGP9} shows the power functions based only on DGP9, the DGP used in the example graphs shown to experts in the second part of the expert study). While a majority of experts correctly identifies the bin width treatment with lowest type I error rates (i.e., most experts prefer small bins at the zero discontinuity level, either with or without fit lines), there is also significant expert support for the large bin with fit lines treatment, even when there is no discontinuity, which exhibits the greatest type I error rate in our sample. In addition, experts fail to predict the type I vs type II error tradeoff presented by the bin width choice: most experts expect large bins to perform worst even at large discontinuities, while we find this treatment arm has the lowest type II error rates in those cases. Although in the actual power functions, the effects of bin width are much more pronounced than the effect of fit lines, we find more expert disagreement regarding non-expert performance on bin width than on fit lines. We also find expert predictions to be similar whether their own visual inference performance is above or below the median.

\subsection{Visual versus Econometric Inference\label{sec:Visual-vs-Econometric}}

Here we compare the performances of visual inference from our small bin expert sample with various econometric RD procedures. We present both the overall power functions and the difference between visual and econometric inferences for each econometric procedure. For a fair comparison, we base the estimators' power calculations on their rejection decisions over the same set of datasets underlying the graphs seen by the experts. That is, the estimators ``see'' the same data as the experts, preventing differences driven by variation in sampling from the same DGP.

As a benchmark, our first estimator comes from a correctly specified model: a global piecewise quintic regression with homoskedastic standard errors. The power function for the corresponding 5\% test  compared with human performance is presented in the left panel of Figure \ref{fig:RDD-Power-Experts-vs-All}. We again use solid (hollow) markers to indicate that the difference to the comparison power function is (not) statistically significant, and provide plots of the differences in Figure \ref{fig:RDD-Experts-Vs-All-Diff}. We additionally include in Table \ref{tab:avg_error_rates_experts} type I and II error rates (the latter at both \(\lvert d \rvert = 0.324\sigma\) and averaged across nonzero discontinuities) for visual and econometric inferences.

Next, we implement the IK, CCT, and AK inference procedures, again plotting the corresponding power functions in the left panel of Figure \ref{fig:RDD-Power-Experts-vs-All}. All three procedures build upon local linear regressions but take different approaches to conducting inference. We introduce them here briefly before discussing the performance of each. We provide a more detailed review of the methods in Online Appendix \ref{sec:visual_v_metrics_details}. 

Strictly speaking, \citet{ImbensKalyanaraman2012} do not study inference but propose an MSE-optimal bandwidth selector, the IK bandwidth. The IK inference procedure we refer to is the ``conventional'' (terminology from \citealp{Calonicoetal2014}) inference procedure practitioners typically implement in conjunction with the IK bandwidth in which the asymptotic bias is ignored. As seen in the left panel of Figure \ref{fig:RDD-Power-Experts-vs-All}, the IK procedure achieves even lower type II error rates than the piecewise quintic estimator, and is significantly better than visual inference at detecting discontinuities up to $1.5\sigma$. But with this advantage in type II error rate comes a significant disadvantage in type I error rate. When there is truly no discontinuity, the estimator still rejects the null hypothesis in 22.6\% of datasets.

Unlike IK, the CCT procedure by \citet{Calonicoetal2014} directly estimates the asymptotic bias, centers the confidence interval at the bias-corrected estimate, and adjusts the width of the confidence interval to account for the uncertainty in the bias estimate. CCT also generalize IK and propose a new class of MSE-optimal bandwidth selectors, which are implemented in the Stata package \texttt{rdrobust} along with the inference procedure.\footnote{While the CCT bandwidth remains the default and modal choice of \texttt{rdrobust}, new work by \citet{Calonicoetal2020optimal} proposes inference-optimal RD bandwidth selectors. As seen in Figure \ref{fig:RDD-Experts-VS-CCT-CER}, using these bandwidths reduces the excess type I error rate relative to visual inference.} Note that although CCT's type I error rate is approximately 12.5\%, as seen in the left panel of Figure \ref{fig:RDD-Power-Experts-vs-All}, this could be a small sample problem. As mentioned in Section \ref{subsec:expert-study}, we effectively have 88 datasets modulo the discontinuity level for the expert study. In a separate Monte Carlo simulation with 1,000 draws, the type I error rate is in line with that of the experts at 7.0\%. However, it is possible that the type I error rate of visual inference also decreases over graphs based on these alternative datasets---the realization of the disturbance term can impact both econometric and visual inference---and therefore we keep the comparison based on the datasets the experts saw. The CCT inference procedure achieves lower type II error rates, enjoying a significant 15 to 20 percentage point advantage over expert visual inference at intermediate discontinuity levels.\footnote{One may be concerned that visual inference's lower type I error rate is a consequence of our experimental design where the majority (nine out of 11) of graphs feature a discontinuity. If subjects speculate that only about half the graphs feature a discontinuity, our type-I-error-rate result is biased in favor of visual inference. As mentioned in Section \ref{sec:primary-experiment}, we explicitly tell participants that ``all, some or none of the 11 graphs you see in this survey may feature a discontinuity,'' and we present evidence against this bias in Online Appendix \ref{subsec:Dynamic-Visual-Inference}, where we test dynamic visual inference.}

Finally, the AK inference procedure by \citet{ArmstrongandKolesar2017} adapts \citet{donoho1994statistical} and produces asymptotically valid and minimax (near-)optimal confidence intervals over a class of conditional expectation functions with a bound---loosely speaking---on the second derivative magnitudes just above and just below the cutoff. This bound informs the worst-case bias of a local linear estimator over the class of functions, and AK corrects for this bias in their procedure (unlike CCT, AK's bias correction is nonrandom). Using the default rule-of-thumb method in the \texttt{RDHonest} R package to estimate the tuning parameter, the AK inference has a type I error rate of approximately 6\%, and the power function is very close to that of the experts.\footnote{\citet{ArmstrongKolesar2017} propose analogous confidence intervals that maintain coverage and enjoy minimax optimality over a H\"{o}lder class of functions, which is determined by a global, as opposed to local, bound on the second derivative of the CEF. Though not presented here, we find the corresponding power functions to be similar. Like \citet{ArmstrongandKolesar2017,ArmstrongKolesar2017}, \citet{imbens2019optimized} also adapt the idea of \citet{donoho1994statistical}. They propose an RD estimator through numerical optimization that is minimax mean-squared-error optimal over CEFs with a global second derivative bound. Because the corresponding inference procedure performs similarly to \citet{ArmstrongandKolesar2017} in simulations by \citet{Peietal2018local} in their 2018 working paper version and can be computationally demanding, we do not implement the \citet{imbens2019optimized} procedure here.} Given the minimax optimality of AK, the comparable performance of visual inference is remarkable. However, it is worth emphasizing that although they have approximately the same average type I error rate, AK offers a theoretical guarantee to control the (asymptotic) type I error rate for all DGPs in the Taylor class while visual inference does not. 

The IK and CCT inference procedures at the 5\% level exhibit higher type I and lower type II error rates than visual inference. We can circumvent this tradeoff by adjusting their type I error rate to the level of visual inference: we search for alternative critical $t$-values such that the resulting type I error rate of the econometric inference procedure is equal to that of visual inference and then use those critical values to conduct inference. For IK, this critical $t$-value is 2.46, and it is 2.28 for CCT. We present the results for these type-I-error-rate-adjusted inference procedures in the right panel of Figure \ref{fig:RDD-Power-Experts-vs-All}, with the differences between these procedures and expert visual inference in Figure \ref{fig:RDD-Experts-Vs-All-Size-Adjusted-Diff}. Despite the extent of their differences in type I error rates relative to visual inference, both econometric procedures' type II error rates only increase by around 5-10 percentage points from this adjustment and are still significantly lower than those of visual inference at moderate discontinuities. However, we again caution that this result is specific to the sample of the 11 DGPs we consider; the type-I-error-rate adjusted IK and CCT procedures are not guaranteed to have correct coverage over the broader Taylor class as AK does. 

We provide two sets of additional results in Online Appendix \ref{sec:visual_v_metrics_details}. First, to investigate the mechanisms that underlie the econometric methods' performances, we impose our knowledge of the DGP when implementing them. We find that the driving force of IK's high type I error rate appears to be the noisy estimates that feed into the optimal bandwidth formula. When we use the theoretical optimal bandwidth, the corresponding type I error rate is lower and statistically indistinguishable from visual inference's. 

Second, we venture beyond binary classifications of a discontinuity and compare the accuracy of discontinuity point estimates across visual and econometric approaches. We find that a simple method like IK attains a lower RMSE than visual inference for each of the 11 DGPs. We conclude that visual inference is less competitive at estimating discontinuity magnitudes than at identifying their existence.

\subsubsection{The Complementarity of Visual and Econometric Inferences \label{subsubsec:Combine-inferences}}

Thus far, we have studied the ``marginal'' power functions for visual and econometric inferences. In this section, we use the joint distribution of visual and econometric discontinuity tests to explore their complementarity and evaluate the performance of a simple combined visual-econometric inference procedure.

First, we examine the joint distribution of visual and econometric inferences to see whether they tend to agree on the same data. For each discontinuity magnitude, we characterize the joint distribution of visual and econometric classifications in the form of a two-by-two contingency table. We conduct Fisher's exact test for independence and present one-sided $p$-values in Table \ref{tab:fisher} by discontinuity magnitude and for each econometric method. We report one-sided $p$-values because two-sided $p$-values are method-dependent due to the ambiguity in classifying contingency tables as extreme in the opposite direction (see, e.g., \citealp{Agresti1992}).\footnote{In principle, we could also report correlations between inferences, but they may be hard to interpret. Because our classification variables are binary, the maximum of the correlation measure depends on the marginal distributions of the classifications. For example, if the probabilities of rejecting no discontinuity are different between two inference procedures, then their correlation is strictly less than 1. Because the (marginal) classification probabilities vary across discontinuity levels and by methods, it is difficult to compare the correlation measures.}

Two patterns emerge from Table \ref{tab:fisher}. The first row shows no strong support for an association between visual and econometric classifications when the true discontinuity is zero. But when the true discontinuity is nonzero (rows two through six), there appears to be strong evidence in support of association (though not reported in Table \ref{tab:fisher}, all correlations are positive in these cases). In other words, type II errors by experts are predictive of type II errors by various econometric inference methods, but this is not true for type I errors, which highlights the complementarity of visual and econometric inferences.

Second, we illustrate this complementarity more concretely by studying the performance of a particular combined visual-econometric inference procedure. It infers a discontinuity if and only if both the visual and econometric procedures reject the null hypothesis of no discontinuity. As a referee points out, many researchers may already use this procedure informally when reading or writing RD papers.

We plot the resulting power functions in the top panel of Figure \ref{fig:RDD-Visual-Econometric-Inference-Combined-vs-AK} along with the power function of the AK procedure for comparison. Because of the lack of dependence between visual and econometric classifications when $d=0$, the combined inferences achieve lower type I error rates than either of the individual inference types. On the other hand, the same mechanism pushes type II error rates higher, but the positive associations between the classifications when $d\neq0$ help limit their increase. 

In fact, the power function of the combined visual-IK inference procedure is fairly close to that of AK. The bottom panel of Figure \ref{fig:RDD-Visual-Econometric-Inference-Combined-vs-AK} presents the difference between the two. Despite the limitations of the IK procedure, with a type I error rate above 20\% as shown before, the IK-expert hybrid has a type I error rate of 2.6\% while not performing statistically significantly differently from AK at any of the nonzero discontinuity levels. This finding helps to explain the enduring credibility of RDDs despite potential issues with the econometric inference method used prior to CCT and AK. It also suggests that the \emph{de facto} type I error probability may be lower than the nominal level if researchers informally combine different statistical evidence instead of relying on a single econometric inference result, a point that deserves attention in future research.

\section{Conclusion\label{sec:Conclusion}}

This paper studies visual inference and graphical representation in RD designs via crowdsourcing. Through a series of experiments and studies that recruit both non-expert and expert participants, we provide answers to two sets of questions. First, how do graphical representation techniques affect visual inference and which technique should practitioners use? And second, when presented with well constructed graphs, how does visual inference perform compared to common econometric inference procedures in RDDs?

To answer the first set of questions, we experimentally assess how five graphical parameters impact visual inference accuracy. We find that generating graphs with the \citet{Calonicoetal2015} IMSE (large) bin selector leads to higher type I error rates but lower type II error rates relative to their MV (small) bin selector. Imposing fit lines can have a similar effect as using large bins, confirming the worries by \citet{ct2021nber,cattaneo2021regression}. We recommend the graphical method of using small bins and no fit lines as a sensible starting point in practice, which we show performs well under several criteria. Bin spacing, a vertical line at the policy threshold, and $y$-axis scaling have little effect, implying that researchers can adhere to reasonable preferences. 

For the second set of questions, we find that visual inference performs competitively on graphs constructed with the recommended method. It achieves a lower type I error rate than econometric inference at the 5\% level based on the \citet{ImbensKalyanaraman2012} and \citet{Calonicoetal2014} methods (the difference between the visual and CCT type I error rates is not statistically significant), though the two econometric inference procedures offer considerable type-II-error advantages. The performance of visual inference is very similar to that based on the procedure suggested by \citet{ArmstrongandKolesar2017}. Furthermore, visual and econometric inferences appear to be complementary. Through the analysis of the joint distribution of visual and econometric tests we find that, while they commit similar type II errors, there does not appear to be a strong association in their type I errors.

Our study is subject to several important limitations. The first is the restricted set of parameters we are able to test experimentally. As mentioned earlier, we do not impose fit lines with confidence intervals in our graphs, which researchers sometimes do, due to the difficulty in explaining it to non-expert participants. We also do not vary the size and color of the dots. However, these choices may impact inference due to their effect on visual attention and visual complexity as suggested by the literature on the psychological and neurological mechanisms underlying the processing of (visual) information \citep{hegarty2010thinking,kriz2007top,rosenholtz2007measuring,wolfe2004attributes}. We leave these investigations to future work.

Second, our results are based on a specific set of DGPs. For example, while the bin spacing choice---equally spaced or quantile spaced---appears immaterial in our experiments, it could be important when the distribution of the running variable is farther from uniform than in our DGPs. On the other hand, the number of DGPs used in Monte Carlo simulations that lead to methodological recommendations is often far lower than our 11, and those DGPs sometimes bear no resemblance to real-world data. In addition, we test the validity of our simulated datasets by adapting the lineup protocol from \citet{Majumderetal2013} to assess the degree to which our DGPs approximate the original data, and we document the robustness of our experimental results to alternative DGP specifications.

And third, the mechanism of RD visual inference remains elusive. In a previous working paper \citep{kortingetal2020wp}, we reported the results from an eyetracking study, in which we sought to identify eyegaze patterns (e.g., ``visual bandwidths'') that robustly predict visual inference success. Had predictive patterns emerged from the eyetracking study, we would have followed up with additional experiments, in which we instruct a random subset of the participants to focus their visual attention according to our finding. But we were not able to identify predictive ocular patterns and could only conclude that the processing of visual signals, as opposed to where in a graph participants looked, drove visual inference success. A next step toward better understanding the mechanism is to systematically study the types of DGPs for which visual inference performs well and poorly.

These limitations notwithstanding, our study answers the call by \citet{leek2015statistics} to provide empirical evidence on best practices in data analysis, and our approach can find applications in other important areas. We have conducted analogous experiments to study visual inference and graphical representations in RKDs using DGPs based on  \citet{CardLeePei2009}, \citet{Cardetal2015}, and \citet{Cardetal2015ECMA}  (\citealp{GanongJager2018} also discuss RK visual inference, albeit informally). Interested readers can consult our previous working paper \citep{kortingetal2020wp} for our nuanced findings. Another related topic to study follows the recent work by \citet{Cattaneoetal2019}, who, among other contributions, propose econometric tests of linearity and monotonicity based on binned scatter plots, which are motivated by studies such as \citet{Chettyetal2011} and \citet{Chettyetal2014}. One could assess the impact of the graphical parameters on reader perception and compare visual and econometric linearity/monotonicity tests. A third related topic is structural breaks in time series econometrics, in which graphs serve practically the same purpose as those in RDD. Finally, within time series econometrics, studying visual inference for unit root/stationarity analysis may also be promising.\footnote{In recent work, \citet{ShenWirjanto2019} propose a new framework for stationarity tests, which formalizes the intuition that a visual characteristic of stationary time series is the infinite recurrence of ``simple events'' asymptotically.} In an influential textbook, \citet{StockWatson3E} conduct an augmented Dickey-Fuller (ADF) test for the presence of a unit root in U.S. inflation. Upon finding that the test rejects a unit root at the 10\% level but not at the 5\% level, \citet{StockWatson3E} write ``The ADF statistics paint a rather ambiguous picture\ldots Clearly, inflation in {[}the figure{]} exhibits long-run swings, consistent with the stochastic trend model.\textquotedblright{} In this case, \citet{StockWatson3E} apply visual unit root inference when the test statistic is marginal, which raises the question: can we leverage our eyes to begin with?

\newpage{}

\begin{singlespace}


\bibliography{RDMainBib_arxiv}
\end{singlespace}

\newpage{}

\section*{Tables}

\begin{table}[htbp]
    \caption{Timeline of Experiments and Graphical Parameters Tested}
    \label{tab:rdd-details-main}
    \centering{\footnotesize{}}%
    \begin{tabular}{lcllcc}
        Phase &  & Holding Fixed & Treatments & Date & \makecell[c]{\# Recruited\\ {[} \# Completions (Rate) {]}}
        \\[0.5em]\hline
        &  &  &  &  &   \\[-0.5em]
        \multicolumn{6}{c}{Main Phases}\\[0.5em]
        1 &  & bin spacing: ES & bin width: large vs small & Nov. 13-16, & 4{*}88=352 \\
        &  & fit lines: no & \# & 2018 & [330 (94\%)] \\
        &  & vertical line: yes & axis scaling: normal vs large &  &  \\[0.5em]
        2 &  & axis scaling: default & bin width: large vs small & Feb. 11-12, & 4{*}88=352 \\
        &  & fit lines: no & \# & 2019 &  [325 (92\%)] \\
        &  & vertical line: yes & bin spacing: ES vs QS &  &   \\[0.5em]
        3 &  & bin width: small & fit lines: no; vertical line: yes & Feb. 27, & 3{*}88=264 \\
        &  & bin spacing: ES & fit lines: no; vertical line: no & 2019 &  [248 (94\%)] \\
        &  & axis scaling: default & fit lines: yes; vertical line: yes &  \\[0.5em]
        4 &  & bin spacing: ES & bin width: large vs small & Oct. 28-29, & 4{*}88=352 \\
        &  & axis scaling: default & \# & 2019 &  [340 (97\%)] \\
        &  & vertical line: yes & fit lines: yes vs no &  &   \\[0.5em]\hline
        &  &  &  &  &   \\[-0.5em]
        \multicolumn{6}{c}{Supplemental Phase} \\[0.5em]
        5 &  & bin width: small & global quintic vs  & Mar. 10-11, & 4{*}88=352 \\
        &  & fit lines: no & local linear specification & 2021 &  [339 (96\%)] \\
        &  & bin spacing: ES & \#  &   \\
        &  & axis scaling: default & homoskedastic vs  & \\
        &  & vertical line: yes &  heteroskedastic error & \\[0.5em]\hline
    \end{tabular}{\footnotesize\par}
    \begin{singlespace}
        \noindent %
        \begin{minipage}[t]{6.3in}%
            \begin{singlespace}
            {\scriptsize{}Notes: In our four main experimental phases, we test the effects of: \\
            \tabitem the bin width selector (we choose two bin width algorithms from \citealp{Calonicoetal2015}: the first, called \emph{large} above, minimizes the integrated mean squared error of the bin-average estimators of the conditional expectation function and results in fewer, larger bins; the second, called \emph{small} above, aims to approximate the variability of the underlying data and results in more, smaller bins);\\
            \tabitem bin spacing (evenly spaced, called \textit{ES} above, and quantile spaced, called \textit{QS} above);\\
            \tabitem parametric fit lines;\\
            \tabitem a vertical line at the policy threshold; and\\
            \tabitem $y$-axis scaling (the default output from Stata 14, called \emph{normal} above, and an increased scale created by recording the range of the $y$-variable from the default graph and increasing the bounds by 50\% of the original range in each direction, called \emph{large} above). \\
            In the supplemental phase, we test the sensitivity of visual inference to alternative specifications of the data generating processes.}
        \end{singlespace}
        \end{minipage}
    \end{singlespace}   
\end{table}

\newpage

\begin{table}[htbp]
    \caption{Type I (False Positive) and Type II (False Negative) Error Rates for Expert and Econometric Inferences}
    \label{tab:avg_error_rates_experts}
    \centering{\footnotesize{}}%
    \begin{tabular}{lccc}
        {\footnotesize{}Inference Type} & {\footnotesize{}Type I Error Rate} & {\footnotesize{}Type II Error Rate: $\lvert d \rvert = 0.324\sigma$} & {\footnotesize{}Average Type II Error Rate: $d \neq 0$} \\
        \hline 
        {\footnotesize{}Experts} & {\footnotesize{}0.079} & {\footnotesize{}0.537} & {\footnotesize{}0.336} \\
        {\footnotesize{}Piecewise Quintic} & {\footnotesize{}0.068 [0.785]} & {\footnotesize{}0.395 [0.003]} & {\footnotesize{}0.260 [0.000]} \\
        {\footnotesize{}IK} & {\footnotesize{}0.226 [0.005]} & {\footnotesize{}0.216 [0.000]} & {\footnotesize{}0.145 [0.000]} \\
        {\footnotesize{}CCT} & {\footnotesize{}0.132 [0.252]} & {\footnotesize{}0.342 [0.000]} & {\footnotesize{}0.219 [0.000]} \\ 
        {\footnotesize{}AK} & {\footnotesize{}0.058 [0.585]} & {\footnotesize{}0.500 [0.429]} & {\footnotesize{}0.321 [0.400]} \\
        \hline 
    \end{tabular}{\footnotesize\par}
    \medskip
    \begin{singlespace}
    \noindent %
    \begin{minipage}[t]{6.3in}%
        \begin{singlespace}
            {\scriptsize{}Notes: The second column shows the type I (false positive) error rate. The third column shows the type II (false negative) error rate for a discontinuity of $\vert d\vert =0.324\sigma$. The fourth column shows the average type II error rate when $d \neq 0$, weighting all discontinuity magnitudes equally. We present $p$-values for the difference between experts and each estimator in brackets. They are based on two-way cluster-robust standard errors computed via a stacked regression where we account for the potential correlation between visual and econometric inferences at the dataset level (there are 88 datasets in total) and in visual inferences for the same individual across graphs--see Online Appendix \ref{subsec:two-way-clustering} for details. IK inference is based on a local linear estimator using the IK bandwidth \citep{ImbensKalyanaraman2012}. CCT is the default RDD inference procedure from CCT's }\texttt{\scriptsize{}rdrobust}{\scriptsize{} \citep{Calonicoetal2014}. AK uses the }\texttt{\scriptsize{}RDHonest}{\scriptsize{} procedure with the rule-of-thumb bound on each DGP's second derivative \citep{ArmstrongandKolesar2017}.}
        \end{singlespace}
    \end{minipage}
    \end{singlespace}
\end{table}

\newpage

\begin{table}[htbp]
    \caption{Fisher's Exact Test of Association: Expert Visual vs Econometric Inferences $p$-values (One-Sided)}
    \label{tab:fisher}
    \centering{\footnotesize{}}%
    \begin{tabular}{cccccc}
        & \multicolumn{4}{c}{{\footnotesize{}Estimator for Econometric Inference}} & \tabularnewline
        \hline 
        {\footnotesize{}Discontinuity $|d|$} & {\footnotesize{}PQ} & {\footnotesize{}IK} & {\footnotesize{}CCT} & {\footnotesize{}AK} & \tabularnewline\hline 
        {\footnotesize{}$0$} & {\footnotesize{}0.727} & {\footnotesize{}0.231} & {\footnotesize{}0.616} & {\footnotesize{}0.394} & \tabularnewline
        {\footnotesize{}$0.1944\sigma$} & {\footnotesize{}0.000} & {\footnotesize{}0.000} & {\footnotesize{}0.000} & {\footnotesize{}0.000} & \tabularnewline
        {\footnotesize{}$0.324\sigma$} & {\footnotesize{}0.000} & {\footnotesize{}0.000} & {\footnotesize{}0.000} & {\footnotesize{}0.000} & \tabularnewline
        {\footnotesize{}$0.54\sigma$} & {\footnotesize{}0.000} & {\footnotesize{}0.000} & {\footnotesize{}0.001} & {\footnotesize{}0.000} & \tabularnewline
        {\footnotesize{}$0.9\sigma$} & {\footnotesize{}0.073} & {\footnotesize{}.} & {\footnotesize{}0.042} & {\footnotesize{}0.134} & \tabularnewline
        {\footnotesize{}$1.5\sigma$} & {\footnotesize{}.} & {\footnotesize{}.} & {\footnotesize{}.} & {\footnotesize{}.} & \tabularnewline
        \hline 
    \end{tabular}{\footnotesize\par}
    \medskip
    \begin{singlespace}
    \noindent %
    \begin{minipage}[t]{3.5in}%
        \begin{singlespace}
            {\scriptsize{}Notes: Missing values indicate that a test always rejects the null hypothesis, in which case the $p$-value cannot be computed. PQ uses a correctly specified regression model with global piecewise quintics above and below the treatment threshold and assuming homoskedasticity. IK inference is based on a local linear estimator using the IK bandwidth \citep{ImbensKalyanaraman2012}. CCT is the default RDD inference procedure from CCT's }\texttt{\scriptsize{}rdrobust}{\scriptsize{} \citep{Calonicoetal2014}. AK uses the }\texttt{\scriptsize{}RDHonest}{\scriptsize{} procedure with the rule-of-thumb bound on each DGP's second derivative \citep{ArmstrongandKolesar2017}.}
        \end{singlespace}
    \end{minipage}
\end{singlespace}
\end{table}

\newpage
\section*{Figures}

\begin{figure}[H]
    \centering
    \includegraphics[width=\textwidth/3]{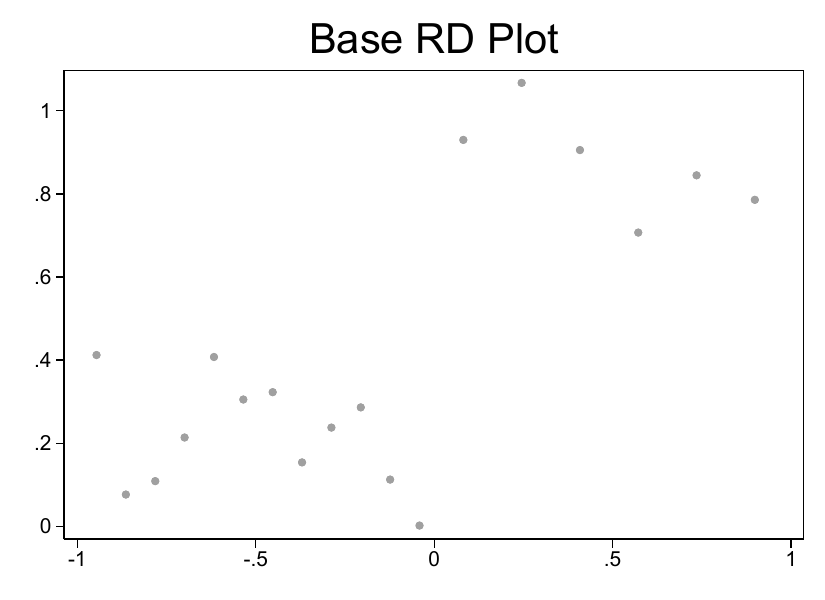}\includegraphics[width=\textwidth/3]{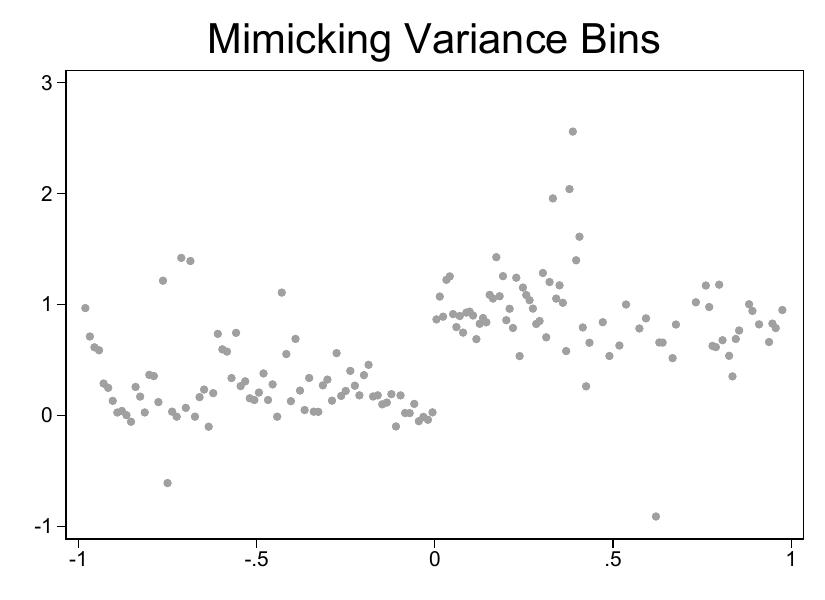}\includegraphics[width=\textwidth/3]{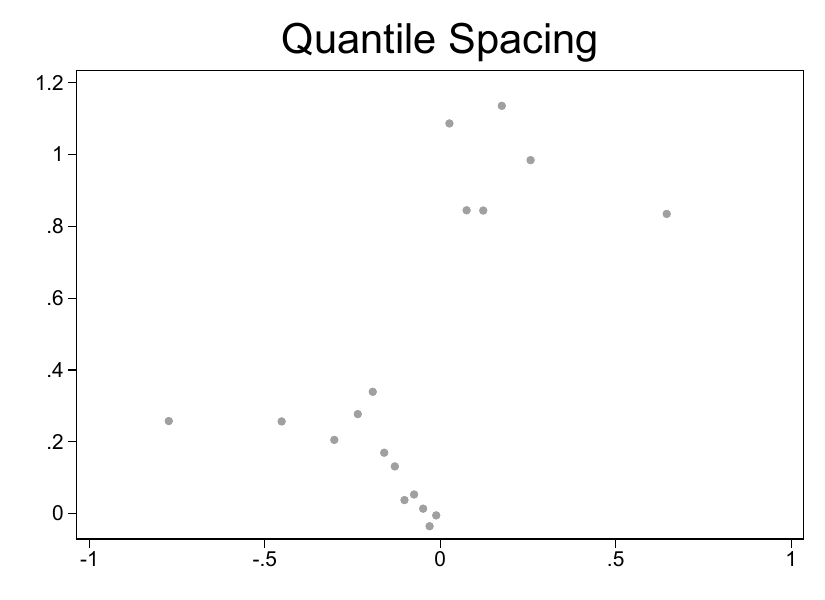}
    \includegraphics[width=\textwidth/3]{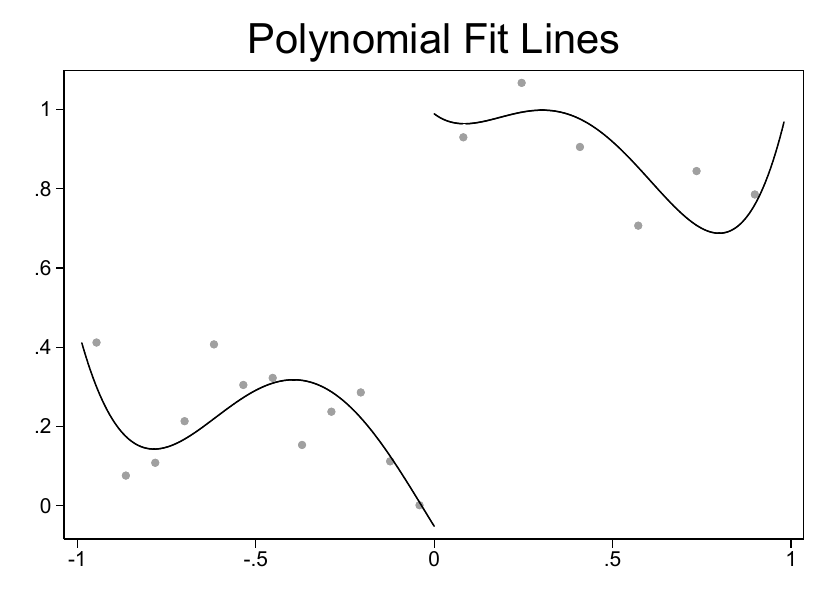}\includegraphics[width=\textwidth/3]{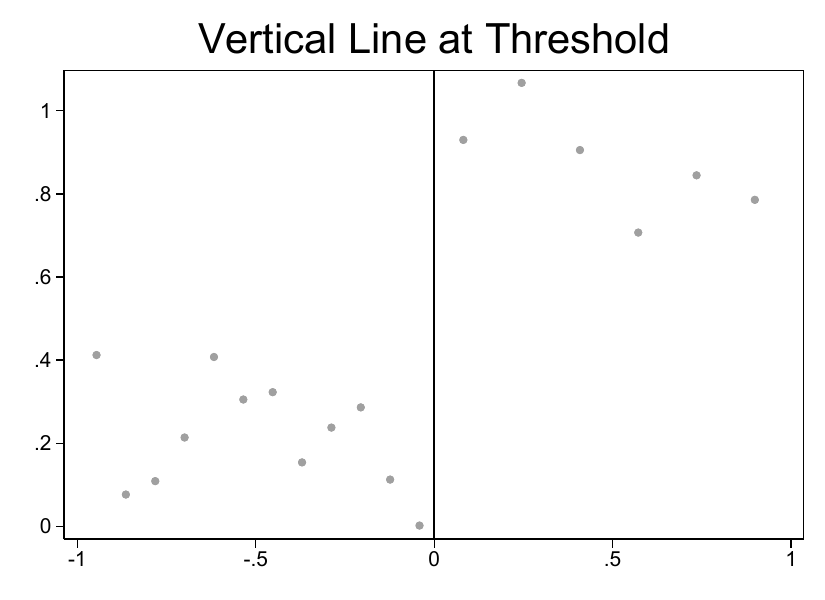}\includegraphics[width=\textwidth/3]{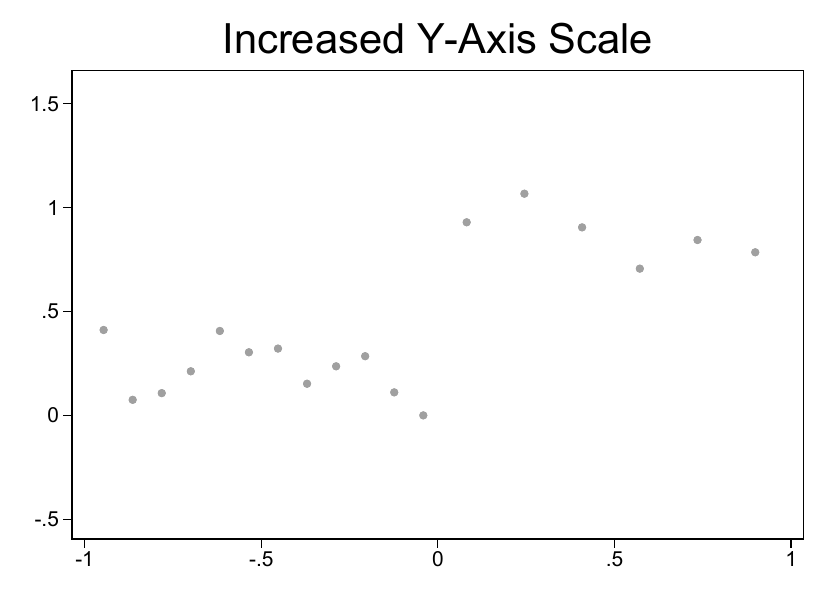}
    \caption{\label{fig:rdd_treatments}Illustration of Graphical Parameters Tested}
    \vspace{3mm}
    \centering{}%
    \begin{minipage}[t]{\textwidth}%
        {\scriptsize{}Notes: Plots are based on the original data from DGP9. The ``Base RD Plot'' uses evenly spaced IMSE-optimal (or ``large'') bins and Stata 14's default axis scaling. The graph labeled ``Mimicking Variance Bins'' uses mimicking-variance (or ``small'') bins and maintains equal bin spacing and default axis scaling. The ``Quantile Spacing'' graph uses quantile-spaced large bins and default axis scaling. The three graphs in the second row present the same binned points as the base plot but with different graphical options: the first imposes parametric polynomial fit lines, the second adds a vertical line indicating the policy threshold at zero, and the third uses a $y$-axis scaling that is twice the Stata 14 default. The bin selectors and spacings---large, small, equally spaced, and quantile spaced---come from the study by \citet{Calonicoetal2015}.}%
    \end{minipage}
\end{figure}

\begin{figure}[H]
    \centering\includegraphics[width=5.5in]{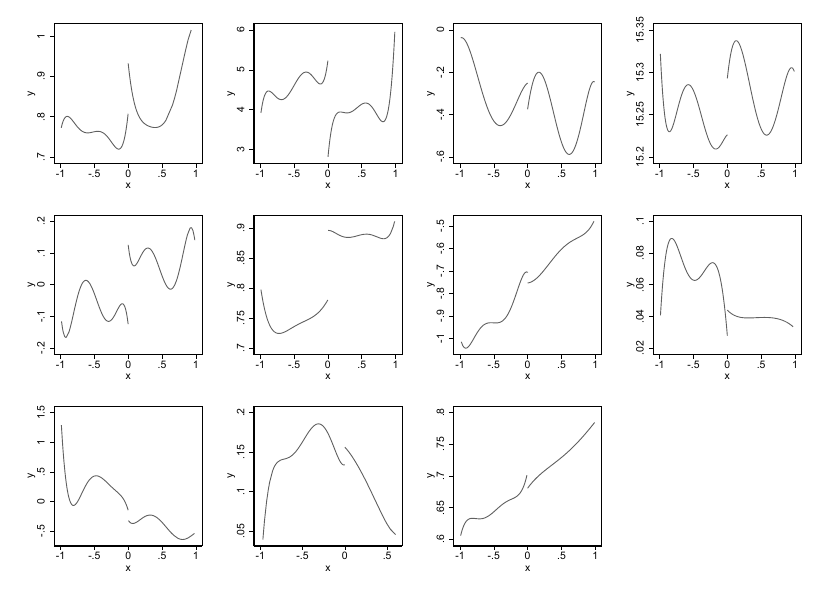}
    \caption{\label{fig:RDD-CEF}Conditional Expectation Functions DGP1-DGP11}
    \vspace{3mm}
    \begin{minipage}[t]{\textwidth}%
        {\scriptsize{}Notes: Each panel represents the plot of the conditional expectation function for DGP1 through DGP11. These functions are obtained by fitting piecewise global quintics to the original microdata after normalizing the support of the running variable and trimming the tails, as described in Section \ref{sec:creation-datasets-graphs}.}%
    \end{minipage}
\end{figure}

\newpage

\begin{figure}[H]
    \centering\includegraphics[width=\textwidth/3]{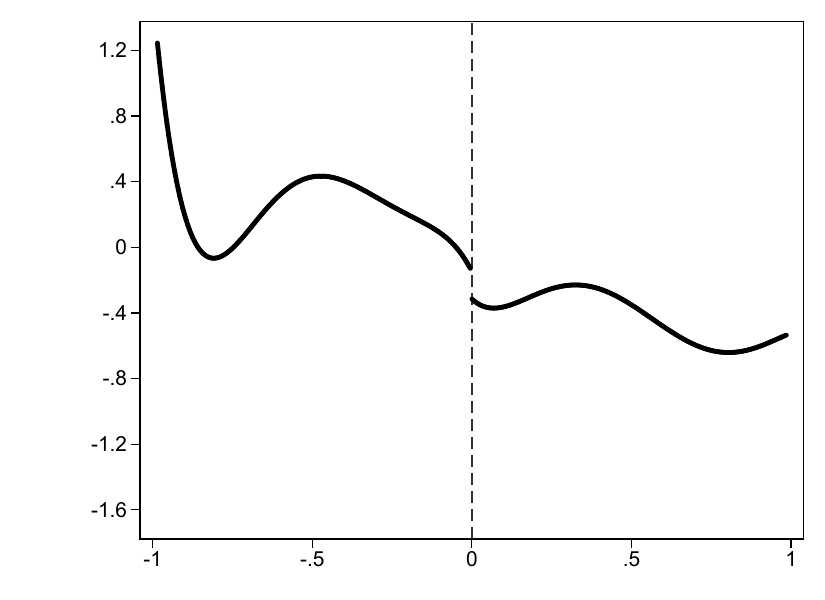}\includegraphics[width=\textwidth/3]{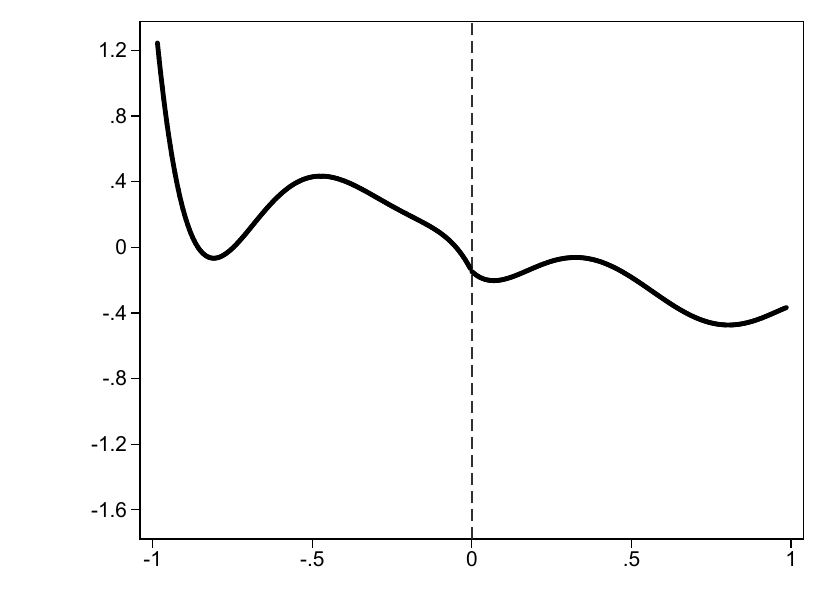}\includegraphics[width=\textwidth/3]{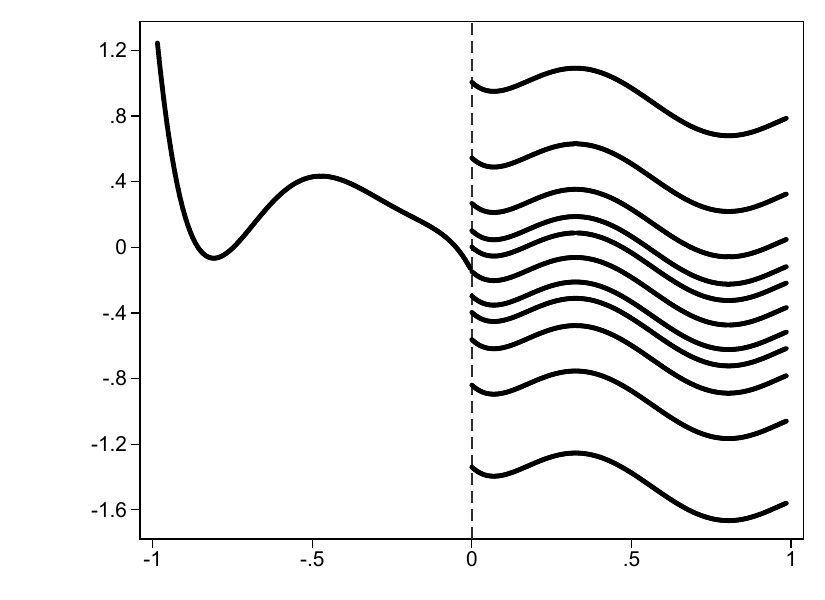}
    \caption{\label{fig:dgp_construction}Creation of Conditional Expectation Functions, DGP9}
    \vspace{3mm}
    \centering{}%
    \begin{minipage}[t]{\textwidth}%
        {\scriptsize{}Notes: The leftmost figure plots the piecewise quintic CEF fitted to the original microdata underlying DGP9 after normalizing the support of the running variable and trimming the tails, as described in Section \ref{sec:creation-datasets-graphs}. The central figure removes the discontinuity by setting the right intercept to equal the left intercept. The rightmost figure plots the final 11 CEFs for DGP9 corresponding to different levels of discontinuity by further changing the right intercept.}%
    \end{minipage}
\end{figure}

\begin{landscape}
    \begin{figure}[H]
        \centering\includegraphics[width=4.5in]{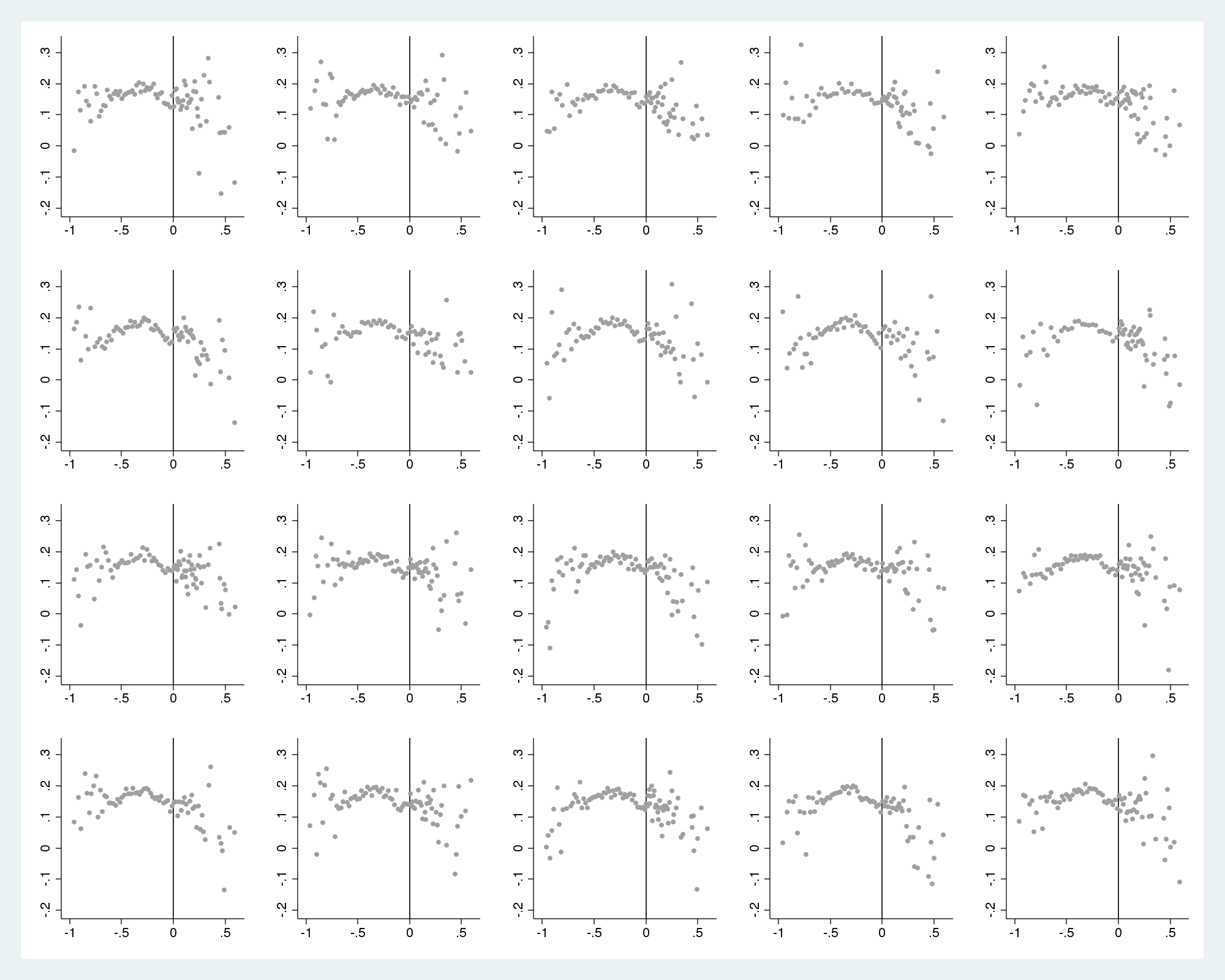}\includegraphics[width=4.5in]{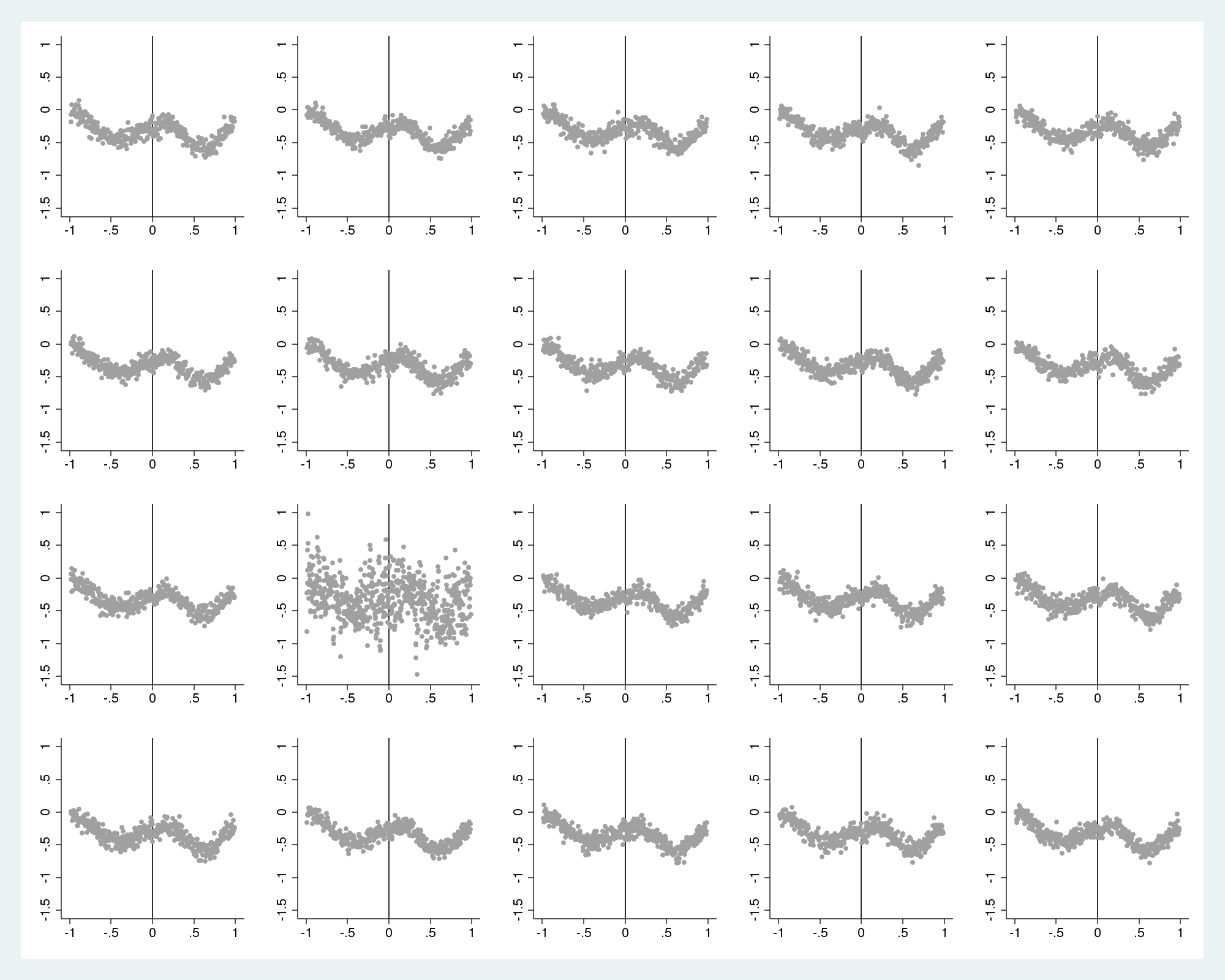}
        \caption{\label{fig:Lineup-Protocol-Example-1}Lineup Protocol
        Graph Examples: DGP10 and DGP3}
        \vspace{3mm}
        \centering{}%
        \begin{minipage}[t]{0.8\columnwidth}%
            {\scriptsize{}Notes: One of the 20 graphs for each lineup protocol is produced from the real data. The other 19 are produced from simulated data drawn from the DGP calibrated to the real data. For DGP10 on the left, the graph produced from data used in the original paper is in row $-3\cdot2+7$ and column $\sqrt{4+5}-2$ (using conventional matrix index notation), while the remaining graphs are generated from our specified DGP (we follow \citealp{Majumderetal2013} and use simple arithmetic to indicate the graph location, so that readers do not accidentally see the answer before reaching their own conclusion). For DGP3 on the right, the graph made with the original data is in row 3 and column 2.}%
        \end{minipage}
    \end{figure}

    \begin{figure}[H]
        \centering\includegraphics[width=3.2in]{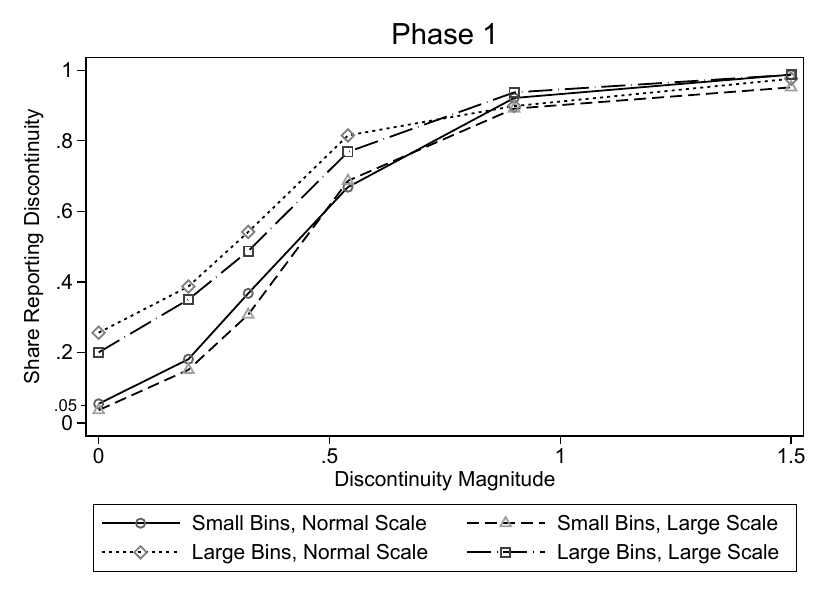}\includegraphics[width=3.2in]{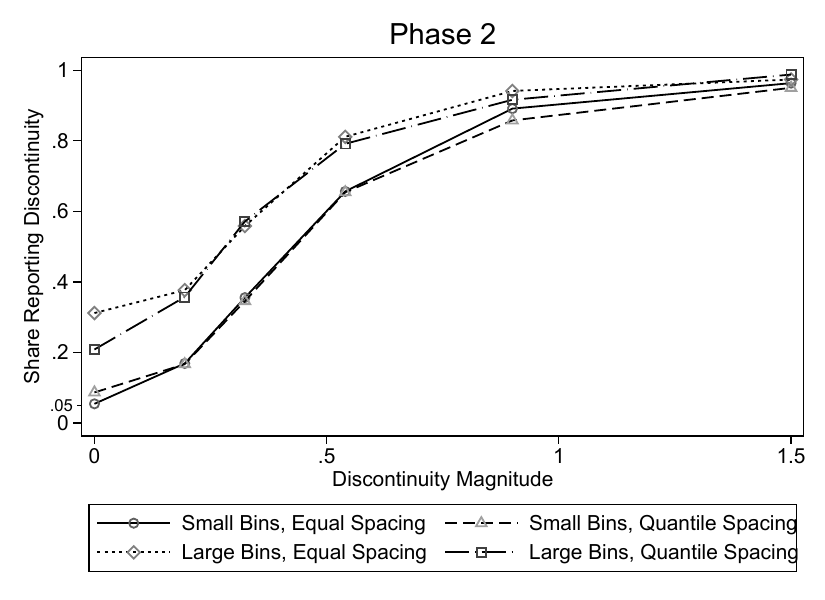}
        \includegraphics[width=3.2in]{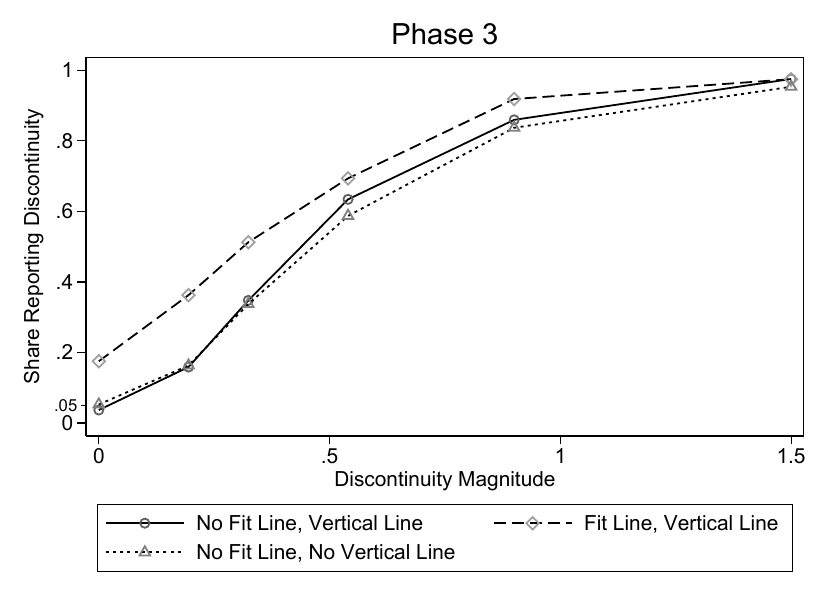}\includegraphics[width=3.2in]{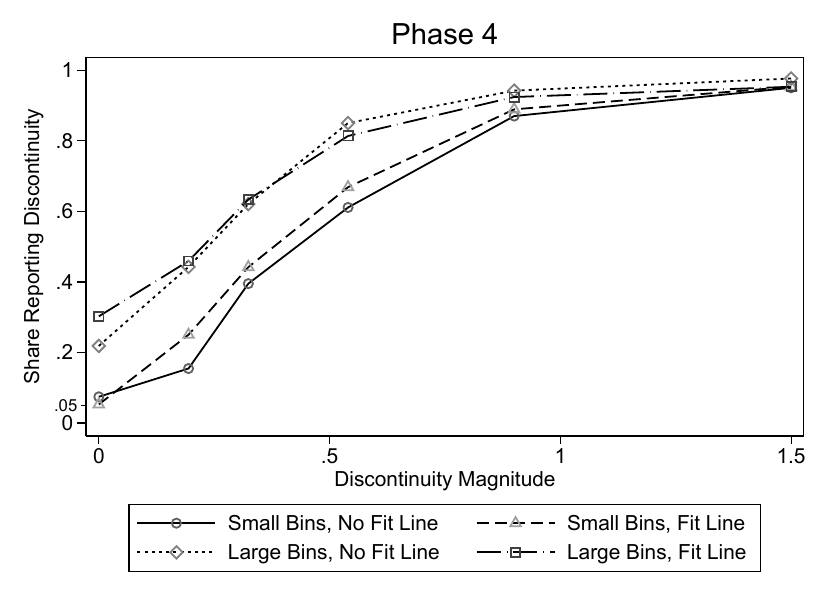}
        \caption{\label{fig:RDD-Power}Power Functions by Experimental Phase}
        \vspace{3mm}
        \begin{minipage}[t]{0.85\columnwidth}%
            {\scriptsize{}Notes: Plotted are empirical power functions from the four non-expert experiment phases. As described in Table \ref{tab:rdd-details-main}, we test the effects of two graphical parameters in each phase while holding fixed the other three parameters. The legend underneath each graph indicates the combination of graphical parameters each phase tests. The power functions are defined in Section \ref{sec:conceptual-framework}. The discontinuity magnitude on the $x$-axis is specified as a multiple of the error standard deviation. The $y$-axis represents the share of respondents classifying a graph as having a discontinuity at the policy threshold.
            \textit{Large bins} corresponds to the \citet{Calonicoetal2015} bin width selector that minimizes the integrated mean squared error of the bin-average estimators of the conditional expectation function; \textit{Small bins} corresponds to the \citet{Calonicoetal2015} bin width selector that aims to approximate the variability of the underlying data; \textit{Quantile spacing} indicates that bins were spaced by quantiles rather than evenly spaced; \textit{Fit line} indicates the presence of parametric fit lines; \textit{Vertical line} indicates the presence of a vertical line at the policy threshold; \textit{Normal scale} corresponds to the default $y$-axis scaling using Stata 14; \textit{Large scale} doubles that default range.}%
        \end{minipage}
    \end{figure}
\end{landscape}

\begin{figure}[H]
    \centering\includegraphics[width=3.5in]{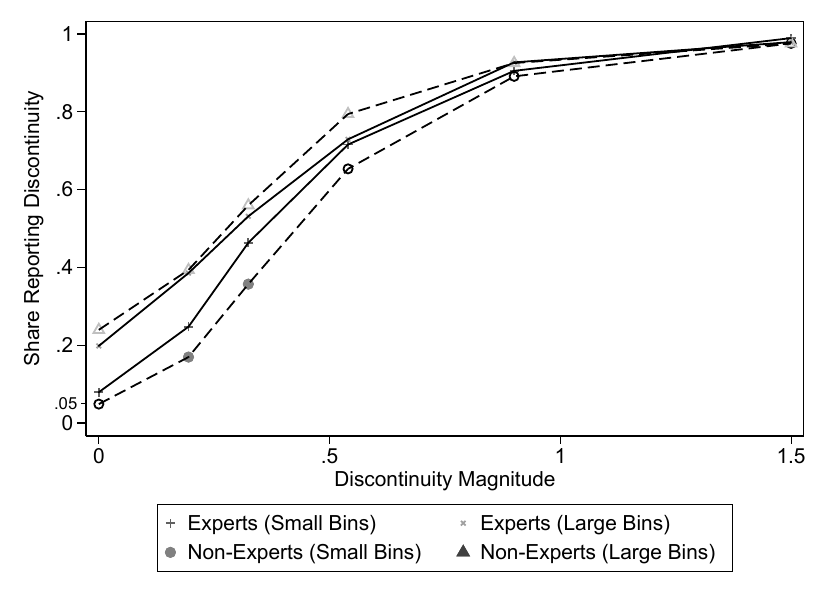}
    \caption{\label{fig:Experts-Nonexperts-ROC}Expert vs Non-Expert Performance}
    \vspace{3mm}
    \begin{minipage}[t]{0.75\columnwidth}%
        {\scriptsize{}\centering{Notes: Plotted are the power functions for the experts and non-experts. Markers in the figure are shown as solid, matching the legend, whenever non-experts perform statistically significantly differently at the 5\% level from experts under the same graphical treatment and at the same discontinuity magnitude. Markers in the figure are shown as the same shape but hollow instead whenever non-experts do not perform statistically significantly differently at the 5\% level from experts under the same graphical treatment and at the same discontinuity magnitude. \textit{Large bins} corresponds to the \citet{Calonicoetal2015} bin width selector that minimizes the integrated mean squared error of the bin-average estimators of the conditional expectation function; \textit{Small bins} corresponds to the \citet{Calonicoetal2015} bin width selector that aims to approximate the variability of the underlying data.}}%
    \end{minipage}
\end{figure}

\begin{figure}[H]
    \centering\includegraphics[width=5.5in]{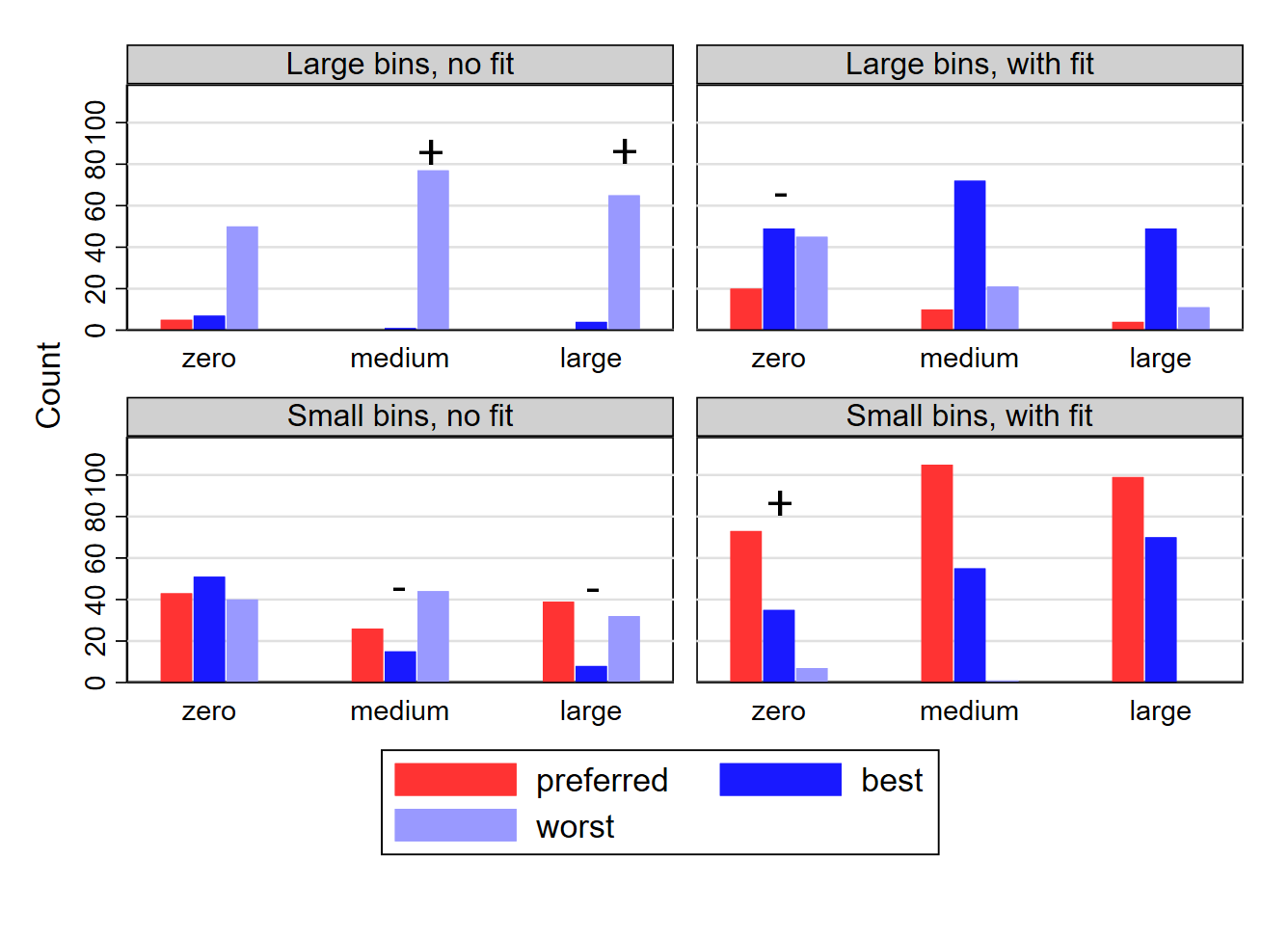}
    \caption{\label{fig:RDD-Experts-Pred-Pref}Expert Preferences and Beliefs
    Regarding Non-Expert Performance}
    \vspace{3mm}
    \begin{minipage}[t]{0.75\columnwidth}%
        {\scriptsize{}Notes: Each panel shows the number of experts who report the given treatment as being their preferred treatment at a given discontinuity level; believe it to be the best-performing treatment among our non-expert sample; or believe it to be the worst-performing treatment among our non-expert sample. For comparison, the treatments that performed best and worst in that sample are marked with a $+$ and $-$ sign respectively. \textit{Large bins} corresponds to the \citet{Calonicoetal2015} bin width selector that minimizes the integrated mean squared error of the bin-average estimators of the conditional expectation function; \textit{Small bins} corresponds to the \citet{Calonicoetal2015} bin width selector that aims to approximate the variability of the underlying data; \textit{with fit/no fit} indicate the presence or absence of parametric fit lines.}%
    \end{minipage}
\end{figure}

\begin{landscape}
    \begin{figure}[H]
        \centering\includegraphics[width=3.5in]{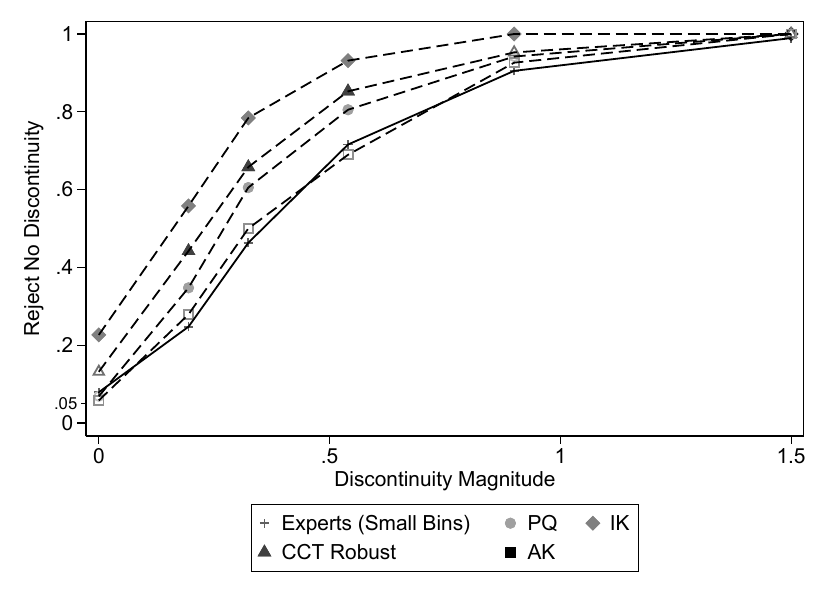}\includegraphics[width=3.5in]{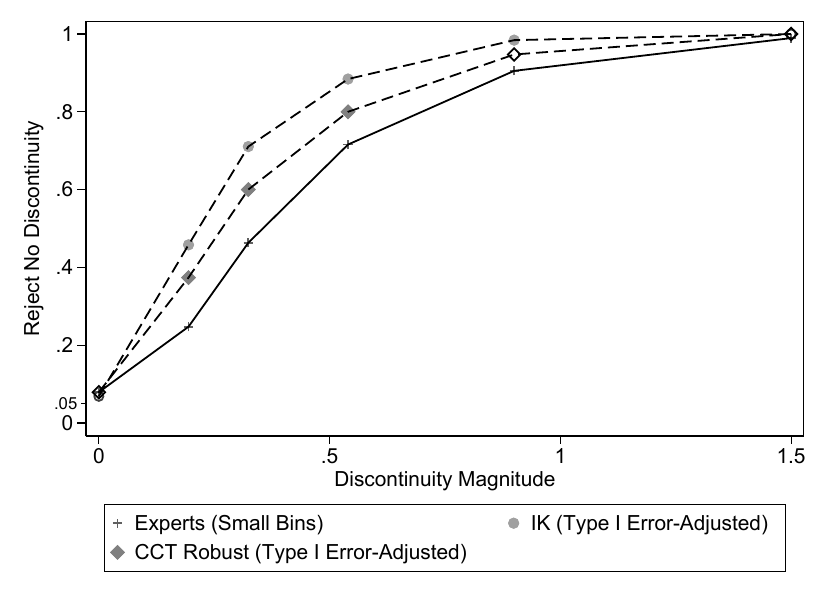}
        \caption{\label{fig:RDD-Power-Experts-vs-All}Expert Visual vs Econometric Inference}
        \vspace{3mm}
        \centering{}%
        \begin{minipage}[t]{0.75\columnwidth}%
            {\scriptsize{}Notes: PQ uses a correctly specified regression model with global piecewise quintics above and below the treatment threshold and assuming homoskedasticity. IK is based on a local linear estimator using the IK bandwidth. CCT Robust is the default RDD inference procedure from CCT's }\texttt{\scriptsize{}rdrobust}{\scriptsize{}. AK uses the }\texttt{\scriptsize{}RDHonest}{\scriptsize{} procedure with the rule-of-thumb bound on each DGP's second derivative. Markers in the figure are shown as solid, matching the legend, whenever the econometric inference procedure performs statistically significantly differently at the 5\% level from expert visual inference at the same discontinuity magnitude. Markers in the figure are shown as the same shape but hollow instead whenever the econometric inference procedure does not perform statistically significantly differently at the 5\% level from expert visual inference at the same discontinuity magnitude. Expert visual inference is shown for the case of \textit{small bins}, corresponding to the \citet{Calonicoetal2015} bin width selector that aims to approximate the variability of the underlying data.}%
        \end{minipage}
    \end{figure}
\end{landscape}

\begin{figure}[H]
    \hspace{-0.4cm}\centering\includegraphics[width=3.70in]{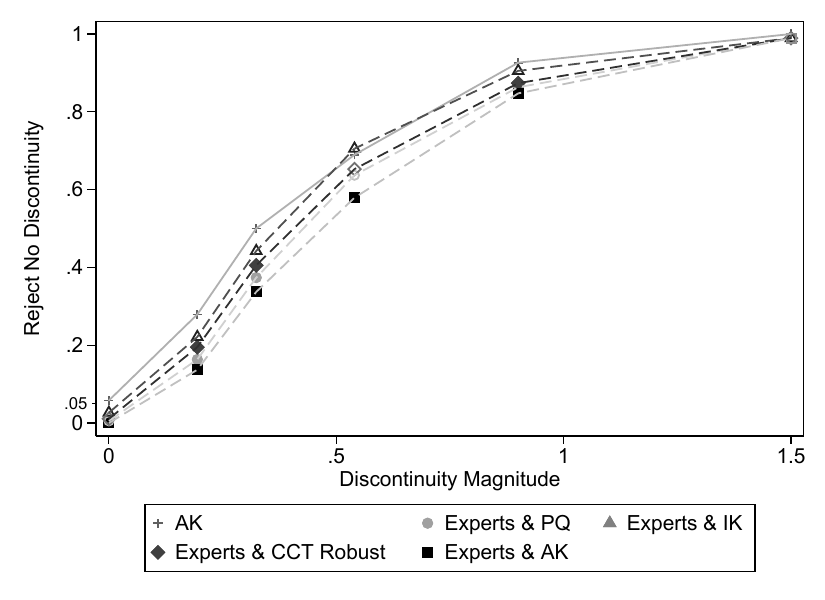}
    \hspace{0.0cm}\includegraphics[width=3.5in]{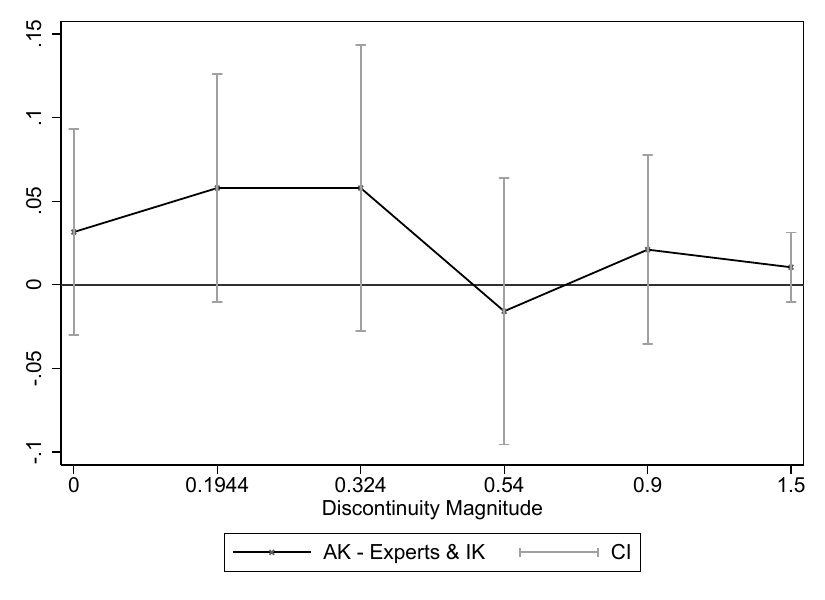}
    \caption{\label{fig:RDD-Visual-Econometric-Inference-Combined-vs-AK}Combined Expert Visual and Econometric
    Inference vs AK}
    \vspace{3mm}
    \centering{}%
    \begin{minipage}[t]{0.75\columnwidth}%
        {\scriptsize{}Notes: Combined expert and visual econometric inference is based on the performance of an inference procedure which infers a discontinuity if and only if both the visual and econometric procedures reject no discontinuity. The relevant expert inference is based on \textit{small bins}, corresponding to the \citet{Calonicoetal2015} bin width selector that aims to approximate the variability of the underlying data. PQ uses a correctly specified regression model with global piecewise quintics above and below the treatment threshold and assuming homoskedasticity. IK is based on a local linear estimator using the IK bandwidth. CCT Robust is the default RDD inference procedure from CCT's }\texttt{\scriptsize{}rdrobust}{\scriptsize{}. AK uses the }\texttt{\scriptsize{}RDHonest}{\scriptsize{} procedure with the rule-of-thumb bound on each DGP's second derivative. Markers in the figure are shown as solid, matching the legend, whenever the combined econometric-expert-visual inference procedure performs statistically significantly differently at the 5\% level  from the AK procedure at the same discontinuity magnitude. Markers in the figure are shown as the same shape but hollow instead whenever the combined econometric-expert-visual inference procedure does not perform statistically significantly differently at the 5\% level  from the AK procedure at the same discontinuity magnitude. In the graph on the bottom, we plot the difference between the combined IK-expert-visual and AK inference procedures, along with 95\% two-way cluster-robust confidence intervals per \citet{cameronetal2011main} that account for potential correlation between visual and econometric inferences at the dataset level (see Online Appendix \ref{subsec:two-way-clustering} for details).}%
    \end{minipage}
\end{figure}

\clearpage{}

\newpage{}

\pagenumbering{arabic}

\appendix

\noindent \begin{center}
~\\
{\Large{}VISUAL INFERENCE AND GRAPHICAL REPRESENTATION IN }\\
{\Large{}REGRESSION DISCONTINUITY DESIGNS}\\ 
\par\end{center}

\bigskip

\begin{singlespace}
\begin{center}
{\large{}Christina Korting, University of Delaware}\\
{\large{}Carl Lieberman, U.S. Census Bureau}\\
{\large{}Jordan Matsudaira, Columbia University}\\
{\large{}Zhuan Pei\daggerfootnote{Corresponding author. Associate Professor, Department of Economics and Jeb E. Brooks School of Public Policy, Cornell University, Ithaca, NY 14853, USA; zhuan.pei@cornell.edu}, Cornell University and IZA}\\
{\large{}Yi Shen, University of Waterloo}\\
\end{center}
\end{singlespace}

\begin{center}
\bigskip
\bigskip
\bigskip
{\Large{}APPENDIX}\\
\par\end{center}

\newpage
\startcontents[sections]
\printcontents[sections]{l}{1}{\setcounter{tocdepth}{2}}
\newpage

\section{Theoretical Details: Estimator Properties, DGP Space, Bin Selectors, and MSE Decomposition}

\renewcommand{\theequation}{A\arabic{equation}}

\setcounter{equation}{0}

\setcounter{footnote}{0}

\subsection{Properties of the Type I and Type II Error Probability Estimators\label{sec:Proofs}}

\setcounter{figure}{0} \renewcommand{\thefigure}{A.\arabic{figure}} 
\setcounter{table}{0} \renewcommand{\thetable}{A.\arabic{table}} 

Following Section \ref{sec:conceptual-framework}, Proposition \ref{prop:g} states the properties of the $g$-specific estimator $\hat{p}(\gamma,g,d)$.
\begin{prop}
\label{prop:g}Under Assumption \ref{assu:Assum_iid} and for a given
DGP $g$
\begin{enumerate}
\item $E[\hat{p}(\gamma,g,d)]=p(\gamma,g,d)$
\item $\hat{p}(\gamma,g,d)\overset{\mathbb{P}}{\to}p(\gamma,g,d)$ as $M\to\infty$
\item $M\cdot\hat{p}(\gamma,g,d)\sim\text{Binomial}(M,p(\gamma,g,d))$
\item $\sqrt{M}(\hat{p}(\gamma,g,d)-p(\gamma,g,d))\Rightarrow N\left(0,p(\gamma,g,d)\cdot(1-p(\gamma,g,d))\right).$
\end{enumerate}
\end{prop}
The proof of the proposition follows trivially from Assumption \ref{assu:Assum_iid}.

Parts 1 and 2 of the proposition state that $\hat{p}(\gamma,g,d)$ is an unbiased and consistent estimator for $p(\gamma,g,d)$. Part 3 states that $\hat{p}(\gamma,g,d)$ has a scaled binomial finite sample distribution, for which methods such as the Clopper-Pearson confidence interval have been developed for finite sample inference, which allows us to construct confidence intervals on the $g$-specific type I and II error probabilities. The standard asymptotic normality result is provided in Part 4, but we do not invoke it in our analysis, as $M$ is equal to eight in our experiments due to resource constraints.

Now we state the properties of the overall estimator $\hat{\bar{p}}(\gamma,d)$.

\begin{prop}
\label{prop:overall}Under Assumptions \ref{assu:Assum_iid} and \ref{assu:Assum_g}
    \begin{enumerate}
        \item $E[\hat{\bar{p}}(\gamma,d)]=\bar{p}(\gamma,d)$
        \item $\hat{\bar{p}}(\gamma,d)\overset{\mathbb{P}}{\to}\bar{p}(\gamma,d)$
        as $J\to\infty$
        \item Conditional on $\{g_{j}\}_{j=1}^{J}$, $M\cdot J\cdot\hat{\bar{p}}(\gamma,d)$
        follows a Poisson binomial distribution
        \item As $J\to\infty$, $\sqrt{J}(\hat{\bar{p}}(\gamma,d)-\bar{p}(\gamma,d))\Rightarrow N\left(0,var_{g\in\mathcal{G}}(\hat{p}(\gamma,g,d))\right)$, with\\
        $var_{g\in\mathcal{G}}(\hat{p}(\gamma,g,d))=\frac{1}{M}\bar{p}(\gamma,d)+(1-\frac{1}{M})E[p(\gamma,g,d)^{2}]-\bar{p}(\gamma,d)^{2}.$
    \end{enumerate}
\end{prop}

\begin{proof}
    Part 1: 
    \begin{align*}
        E[\hat{\bar{p}}(\gamma,d)] & =\frac{1}{J}\sum_{j=1}^{J}E_{g_{j}\in\mathcal{G}}\left[E[\hat{p}(\gamma,g_{j},d)|g_{j}]\right]\\
         & =\frac{1}{J}\sum_{j=1}^{J}E_{g_{j}\in\mathcal{G}}\left[p(\gamma,g_{j},d)\right]\\
         & =\bar{p}(\gamma,d)
    \end{align*}
    
    \noindent Part 2: Assumptions \ref{assu:Assum_iid} and \ref{assu:Assum_g} imply that, for any $J$, $\{\hat{p}(\gamma,g_{j},d)\}_{1\leqslant j\leqslant J}$ is a set of independent and identically distributed random variables. Because $\hat{p}(\gamma,g_{j},d)^{2}$ is uniformly bounded between zero and one, the weak law of large numbers for triangular arrays (see, for example, \citealp{Durett2010}) applies to the set $\{\hat{p}(\gamma,g_{j},d)\}_{1\leqslant j\leqslant J,J=1,2,...}$. Because the expectation of $\hat{p}(\gamma,g_{j},d)$ is $\bar{p}(\gamma,d)$ for each $j$ (Part 1 of the proposition), we have the desired result.\footnote{The application of the law of large numbers for triangular arrays, as opposed to its standard counterpart, allows $\hat{p}(\gamma,g_{j},d)$ to change as $J$ increases. Take $\hat{p}(\gamma,g_{1},d)$, for example. We accommodate the scenarios where $M$, the number of participants who see $g_{1}$, changes with $J$, or where we allocate different individuals to see $g_{1}$ as we increase the number of DGPs sampled.}\bigskip{}
    
    \noindent Part 3: the statement simply follows the definition of a Poisson binomial distribution. \bigskip{}
    
    \noindent Part 4: Similar to the proof of part 2 of the Proposition, the result follows from an application of the Central Limit Theorem for Triangular Arrays (or the Lindeberg-Feller Central Limit Theorem):
    \[
        \sqrt{J}(\hat{\bar{p}}(\gamma,d)-\bar{p}(\gamma,d))\Rightarrow N\left(0,var_{g\in\mathcal{G}}(\hat{p}(\gamma,g,d))\right).
    \]
    The law of total variance implies that 
    \begin{align*}
        & var_{g\in\mathcal{G}}(\hat{p}(\gamma,g,d))\\
        = & E_{g\in\mathcal{G}}[var(\hat{p}(\gamma,g,d)|g)]+var_{g\in\mathcal{G}}(E[\hat{p}(\gamma,g,d)|g])\\
        = & E_{g\in\mathcal{G}}\left[\frac{p(\gamma,g,d)\left(1-p(\gamma,g,d)\right)}{M}\right]\\
         & +var_{g\in\mathcal{G}}(p(\gamma,g,d))
    \end{align*}
    and the last equality follows from Proposition \ref{prop:g}. Because
    \[
        var_{g\in\mathcal{G}}(p(\gamma,g,d))=E_{g\in\mathcal{G}}[p(\gamma,g,d)^{2}]-E_{g\in\mathcal{G}}[p(\gamma,g,d)]^{2}
    \]
    we have 
    \[
        var_{g\in\mathcal{G}}(p(\gamma,g,d))=\frac{1}{M}\bar{p}(\gamma,d)+(1-\frac{1}{M})E_{g\in\mathcal{G}}[p(\gamma,g,d)^{2}]-\bar{p}(\gamma,d)^{2}
    \]
    following simple algebra.
\end{proof}

Parts 1 and 2 of Proposition \ref{prop:overall} establish unbiasedness and consistency of $\hat{\bar{p}}(\gamma,d)$ as an estimator for $\bar{p}(\gamma,d)$. Part 3 provides the finite-sample distribution of $\hat{\bar{p}}(\gamma,d)$; the Poisson binomial distribution is that of a sum of independent but potentially non-identically distributed Bernoulli random variables. Part 4 states the asymptotic normality result for $\hat{\bar{p}}(\gamma,d)$ as $J\to\infty$. In principle, we can consistently estimate the variance and use Part 4 to conduct inference on $\bar{p}(\gamma,d)$ if $J$ is large. As with $M$, however, we choose a moderate $J$ ($11$) in our experiments due to resource constraints, and the asymptotic normality statement is, therefore, more conceptual than practical. Nevertheless, we make the following observations on the variance expression in Part 4 to help understand the variation in the estimator $\hat{\bar{p}}(\gamma,d)$. In essence, the variance expression reflects the block assignment nature of the DGPs to the participants. When $M=1$, each DGP is only assigned to a single participant, the $R_{i}$'s are i.i.d., and $\hat{\bar{p}}(\gamma,d)$ simply follows a (scaled) binomial distribution with variance $\bar{p}(\gamma,d)-\bar{p}(\gamma,d)^{2}$. When $M>1$, each DGP is assigned to a block of participants, the $R_{i}$'s are no longer independent within the same block, and the distribution of $\hat{\bar{p}}(\gamma,d)$ generally deviates from a (scaled) binomial. When $M$ is large, we can precisely estimate $p(\gamma,g,d)$ for each DGP $g$, and the variance of $\hat{\bar{p}}(\gamma,d)$ reduces to that of $p(\gamma,g,d)$ over the space of DGPs.

In the actual analysis of our experimental data, we use the normal approximation of the binomial random variable $\text{Binomial}(M\cdot J,\bar{p}(\gamma,d))$ when $M\cdot J$ is large. Based on a result from \citet{Hoeffding1956} and summarized by \citet{PercasandPercas1985}, using $\text{Binomial}(M\cdot J,\bar{p}(\gamma,d))$ leads to conservative inference on the Poisson binomial random variable $M\cdot J\cdot\hat{\bar{p}}(\gamma,d)$ when conditioning on $\{g_{j}\}_{j=1}^{J}$.\footnote{For a given $(\gamma,d)$, inference for $\hat{\bar{p}}(\gamma,d)$ using $\text{Binomial}(M\cdot J,\bar{p}(\gamma,d))$ is exact when $p(\gamma,g_{j},d)$ is the same for all $j$.} Therefore, conditioning on the set of selected DGPs, using the normal distribution $N\left(\hat{\bar{p}}(\gamma,d),\hat{\bar{p}}(\gamma,d)(1-\hat{\bar{p}}(\gamma,d))/(M\cdot J)\right)$ provides (approximately) conservative inference for $\hat{\bar{p}}(\gamma,d)$ and only relies on the product $M\cdot J$ being large, as opposed to $J$ being large by itself.

\subsection{The DGP Space\label{subsec:DGP-space}}

We formally define the space of DGPs, $\mathcal{G}$. As discussed in Section \ref{sec:conceptual-framework}, the data generating process $g$ has four components: i) the running variable distribution; ii) the conditional expectation function which is continuous at 0; iii) the error term distribution; iv) the sample size. Thus, $\mathcal{G}$ is the product of four spaces.

Let $\mathcal{G}_X$ be the collection of all probability distributions on $(\mathbb{R},\mathcal{B})$ with support on a compact interval that contains $0$ as an interior point, where $\mathcal{B}$ is the Borel $\sigma$-algebra of $\mathbb{R}$. Let $\mathcal{G}_\mu$ be the space of real-valued continuous functions on $\mathbb{R}$. Let $\mathcal{G}_u$ be the collection of mean-zero probability distributions on $(\mathbb{R},\mathcal{B})$. Let $\mathcal{G}_N$ be the set of positive integers. 

We can define the appropriate $\sigma$-algebra for each space. The $\sigma$-algebra for $\mathcal{G}_X$ is generated by the mappings $\nu\to\nu(A)$ from $\mathcal{G}_X$ to $(\mathbb R, \mathcal B)$ for all $A\in \mathcal B$ (p. 49 of \citealp{Kallenberg2017}). The $\sigma$-algebra for $\mathcal{G}_\mu$ is the cylinder $\sigma$-algebra (p. 2 of \citealp{Gusaketal2010}). The $\sigma$-algebra for $\mathcal{G}_u$ is defined analogously to that for $\mathcal{G}_X$. And the $\sigma$-algebra for $\mathcal{G}_N$ is the the power set of positive integers, i.e., the discrete $\sigma$-algebra.

The space of DGPs is 
\[
    \mathcal{G}=\mathcal{G}_X \times \mathcal{G}_\mu \times \mathcal{G}_u \times \mathcal{G}_N,
\]
which is naturally equipped with the product $\sigma$-algebra $\mathcal{F}$. Let the probability measure  $\mathbb{P}$ on $(\mathcal{G},\mathcal{F})$ reflect the distribution of DGPs underlying empirical RD datasets. Because of the infinite dimensionality of $\mathcal{G}$, it is impractical to theoretically characterize $\mathbb{P}$. Instead, we treat the datasets we gather as realizations from DGPs that are sampled from the probability space $(\mathcal{G},\mathcal{F},\mathbb{P})$.

\subsection{Frameworks for Evaluating and Conceptualizing Graphical Methods \label{subsec:Decision-theory}}

In this section, we first present two decision-theoretic frameworks for evaluation graphical methods. They correspond to the second and third criteria mentioned in Section \ref{sec:conceptual-framework}. At the end of the section, we then briefly discuss alternative ways to conceptualize graphs.

In the first framework, we incur a loss from reader $i$ misclassifying graph $T_i(\gamma,g,d)$:
\[
\mathcal{L}^{i}(\gamma,g,d) \equiv R_{i}(T_{i}(\gamma,g,0)) \cdot \kappa \cdot 1_{[d=0]}+(1-R_{i}(T_{i}(\gamma,g,d)))\cdot \varphi \cdot 1_{[d\neq0]}, 
\]
where $\kappa$ and $\varphi$ are the costs of a type I and type II error, respectively. This loss function generalizes the zero-one loss (e.g., p. 20 of \citealp{friedman2001elements}) to allow asymmetric loss for the two error types. This type of loss function has been used in many binary classification problems, e.g., by \citet{ElliottandLieli2013} and \citet{kline2019audits}.

Averaging $\mathcal{L}^{i}$ over $i$, which encapsulates averaging across both readers and graph realizations, leads to the expected loss or risk for DGP $g$:
\[
\mathcal{R}(\gamma,g,d)\equiv p(\gamma,g,0) \cdot \kappa \cdot 1_{[d=0]}+ (1-p(\gamma,g,d)) \cdot \varphi \cdot 1_{[d\neq0]}.
\]
Further integrating $\mathcal{R}$ over the distribution of $g$ leads to:
\[
\bar{\mathcal{R}}(\gamma,d)\equiv \bar{p}(\gamma,0) \cdot \kappa \cdot 1_{[d=0]}+ (1-\bar{p}(\gamma,d)) \cdot \varphi \cdot 1_{[d\neq0]}.
\]
In the remainder of the paper, we refer to $\mathcal{R}$ and $\bar{\mathcal{R}}$ as the classical risks. 

We make three remarks. First, the type I and type II error probabilities are key inputs for the classical risks. Second, unlike econometric inference, it is hard to theoretically control the type I error probability of visual inference under a pre-specified threshold, and different graphical methods often lead to type I and type II error tradeoffs. To choose a method in the presence of these tradeoffs, we need to specify the cost parameters $\kappa$ and $\varphi$ and further average over $d$ according to a prior probability of encountering graphs with different discontinuity levels. These specifications entail subjective judgement, and we discuss our choices in Online Appendix \ref{sec:Risk} when we estimate the classical risks with experimental data. Third, we can define the minimax graphical method as the best performing method under the most adverse DGP. While the vastness of the DGP space makes it difficult to estimate the population maximal risk, the best performing graphical method in terms of the estimated classical risk defined above also does very well under the sample maximal risk among our DGPs.

Our second framework for evaluating graphical methods builds on the recent work on scientific communication by \citet{andrews2020model} (henceforth AS), who propose the so-called communication risk and show that reporting certain statistics achieves lower communication risk than others. Adapted to our context, the (posterior) AS communication risk starts from the same loss function as the classical risk, but it takes an expectation of the loss with respect to the reader's perceived discontinuity distribution.

Formally, the AS risk for reader $i$ when viewing graph $T_{i}$ generated from GGP $(\gamma,g,d)$ is\footnote{Unlike participants in our experiments, the audience of an empirical study has access to other statistical information beyond the RD graph. The AS risk formulation can easily accommodating this by expanding the conditioning set.}
\[
\mathcal{R}_{AS}^{i} (\gamma,g,d) = \min_{\delta \in \{0,1\}}E^{i}\left[\delta \cdot 1_{[D=0]} \cdot \kappa_{i} + (1-\delta) \cdot 1_{[D\neq0]} \cdot \varphi_{i} \mid T_{i}(\gamma,g,d)\right].
\]
We use the capital letter $D$ to denote the reader's perceived discontinuity. More precisely, $D$ is the random variable that represents the discontinuity unknown to the reader, for which she perceives a distribution. The expectation $E^{i}$ is taken with respect to this subjective (posterior) distribution of reader $i$ after she sees a graph. $\kappa_{i}$ and $\varphi_{i}$ are the costs the reader incurs for committing a type I and type II error, respectively, and they are conceptually distinct from the $\kappa$ and $\varphi$ without the subscript $i$ used in the classical risk. While $\kappa$ and $\varphi$ reflect the costs we as evaluators of graphical methods assign to the type I and type II errors, $\kappa_{i}$ and $\varphi_{i}$ are reader-specific and depend on the payoffs in our experimental design. We can think of $\kappa_{i}$ and $\varphi_{i}$ in terms of reader $i$'s utility---this formulation is sufficiently general as to be agnostic of the form of her utility function. A reader minimizing the AS risk is equivalent to her maximizing expected utility. The two values of $\delta$ represent the classification choices, and risk minimization (or expected utility maximization) leads to the reader's discontinuity classification $R_{i}$. Solving the optimization problem implies that she would classify a graph as discontinuous (i.e., $R_{i}=1$) if she is reasonably certain. Denoting her perceived discontinuity probability by $q_{i}\equiv \Pr^{i}(D\neq0|T_{i})$, the reader classifies a discontinuity when $q_{i}$ is above the cutoff $\varsigma_{i} \equiv \kappa_{i}/(\kappa_{i}+\varphi_{i})$, and her posterior AS risk is
\[
\mathcal{R}_{AS}^{i} (\gamma,g,d) = (1-q_{i}) \cdot 1_{[q_{i}\geqslant \varsigma_{i}]} \cdot \kappa_{i} + q_{i} \cdot 1_{[q_{i}< \varsigma_{i}]} \cdot \varphi_{i}, 
\]
which is tent-shaped as a function of $q_{i}$. We focus on the case where $\kappa_{i}=\varphi_{i}=1$ as the payoff is the same in our experimental design regardless of whether the reader commits a type I or a type II error (only the ratio of the two parameters matters in determining the optimal graphical parameter, but choosing a value of $1$ leads to easily interpretable quantities). The risk simplifies to $1-\beta_{i}$ where $\beta_{i} \equiv 0.5+|q_{i}-0.5|$. $\beta_{i}$ is simply the reader's perceived probability of being correct: it attains the highest value of $1$ when she is certain that the graph is continuous ($q_{i}=0$) or discontinuous ($q_{i}=1$), and it attains the lowest value of $0.5$ when she is completely unsure ($q_{i}=0.5$). Therefore, a graphical parameter is preferred under the AS communication risk if it leads to higher reader confidence (this form of the AS risk is also known as the (mis)classification risk, see e.g., \citealp{kitagawa2021constrained}).

Averaging over $i$, we arrive at the AS risk for GGP $(\gamma,g,d)$: 
\[
\mathcal{R}_{AS} (\gamma,g,d) \equiv E_{i}[\mathcal{R}_{AS}^{i} (\gamma,g,d)]=E_{i}[1-\beta_{i}(\gamma,g,d)],
\]
where we make explicit the dependence of $\beta_{i}$ on the GGP $(\gamma,g,d)$. Because we randomly assign graphs to participants, we can interpret the expectation above as first averaging over the realizations of graph $T$ for each participant type $\phi$ as defined in Section \ref{sec:conceptual-framework}, and then averaging over the distribution of $\phi$. The first average is the AS ex ante communication risk if each reader correctly anticipates the objective distribution of $T$ (the ex ante communication risk corresponds to the integral of the ex post communication risk with respect to the reader's subjective prior distribution of $T$, which is hard to ascertain). Because the second average is over the $\phi$-distribution, we can interpret $\mathcal{R}_{AS}$ as a weighted average communication risk per AS where the weights reflect the composition of our experimental subjects. As before, we can further average the DGP-specific risk with respect to the distribution of $g$ to define an overall average communication risk:
\[
\bar{\mathcal{R}}_{AS} (\gamma,d) \equiv E_{g}[\mathcal{R}_{AS} (\gamma,g,d)].
\]

Estimating the AS communication risks requires measures of $\beta_{i}$, each reader's perceived probability of being correct. We do not directly elicit them from the participants in our experiments but can approximate them with each participant's graph-specific choice of a risky or risk-free payment scheme. We discuss the payment schemes in Online Appendix \ref{sec:primary-experiment-bonus-scheme} and compare across graphical methods the estimated values of the communication risk in Online Appendix \ref{sec:Risk}.

Finally, we point out that there are alternative ways to conceptualize graphs beyond the frameworks presented here. While we do not explicitly consider the interest of an analyst in the frameworks above, the analyst may have his own agenda and choose to curate a set of evidence, including graphs, to convince an audience to interpret data in a way that benefits him. \citet{schwartzstein2021using} view graphs through this lens and, as an example, point to fit lines as a device to influence audience perception.\footnote{It is not explicit in the \citet{schwartzstein2021using} formulation, but the analyst may use an interactive process to experiment with different visualizations before choosing a picture to present to the reader. \citet{hullman2021design} provide a conceptual framework of the role of interactive graphical analysis in the context of exploratory analysis and model checking.} Recent studies by \citet{banerjee2020theory} and \cite{spiess2018optimal} study statistical problems---experimental design and covariate adjustments---using decision theory while accounting for analyst preference. Extending this line of thinking to studying graphs may be a promising avenue for future research.

\subsection[CCT Bin Selection Algorithms]{\citet{Calonicoetal2015}  Bin Selection Algorithms\label{sec:Bin-algorithms}}

This section describes the two bin width algorithms from \citet{Calonicoetal2015}. The IMSE algorithm, which minimizes the integrated mean squared error of the bin-average estimators of the CEF and results in fewer, larger bins. The MV algorithm, which aims to approximate the variability of the underlying data and results in more, smaller bins.

Formally, consider without loss of generality the number of bins above the threshold. We denote the relevant quantities with a ``$+$'' subscript. The IMSE bin selector chooses the number of bins $J_{+,N}$ to minimize the integrated mean squared error
\[
    \int_{0}^{\overline{x}}E[(\hat{\mu}_{+}(x)-\mu_{+}(x))^{2}|X=x]f(x)dx=\frac{J_{+,N}}{N}Var_{+}\{1+o_{p}(1)\}+\frac{1}{J_{+,N}^{2}}Bias_{+}\{1+o_{p}(1)\},
\]
where $\mu_{+}$ denotes the CEF, $\hat{\mu}_{+}$ denotes the bin-average estimator of $\mu_{+}$, $f(x)$ is the density of $X$, $N$ is the overall sample size. The constants $Var_{+}$ and $Bias_{+}$ are equal to
\[
    Var_{+}=\frac{1}{\overline{x}}\int_{0}^{\bar{x}}\sigma_{+}^{2}(x)dx
\]
\[
    Bias_{+}=\frac{\overline{x}^{2}}{12}\int_{0}^{\bar{x}}(\mu_{+}^{\prime}(x))^{2}f(x)dx,
\]
with $\bar{x}$ being the upper bound of the support of $X$, $\mu_{+}^{\prime}$ the first derivative of $\mu_{+}$, and $\sigma_{+}^{2}(x)$ the conditional variance of $Y$ given $X=x$.

Solving the minimization problem, we obtain the IMSE-optimal number of bins
\[
    J_{+,N,IMSE}=\left\lceil \left(\frac{2Bias_{+}}{Var_{+}}\right)^{\frac{1}{3}}N^{\frac{1}{3}}\right\rceil 
\]
with an analogous $J_{-,N,IMSE}$ for the number of bins below the threshold.

The goal of the mimicking variance bin selector, on the other hand, is to ``choose the number of bins so that the binned sample means have an asymptotic (integrated) variability approximately equal to the amount of variability of the raw data'' \citep{Calonicoetal2015}. Consider again the problem of choosing the number of bins above the threshold. The MV criterion translates to setting
\[
    \frac{J_{+,N}}{N}Var_{+}=Var(Y|X\geqslant0)\equiv V_{+}.
\]
It follows that the number of bins should be
\[
    J_{+,N}=\frac{V_{+}}{Var_{+}}N.
\]
However, this number of bins is likely to be too large. There are two ways of seeing this. First, it grows linearly with the sample size, which is too fast for the rate conditions in \citet{Calonicoetal2015} that ensure the consistency of $\hat{\mu}_{+}(x)$. Second, as a referee points out, the constant $V_{+}/Var_{+}$ may be larger than one, resulting in more bins above the threshold than the overall sample size. To see this, consider the simple case where $\bar{x}=1$. By the law of total variance, 
\begin{align*}
V_{+}&=Var(Y|X\geqslant0)\\&=E[Var(Y|X)|X\geqslant0]+Var(E[Y|X]|X\geqslant0) .    
\end{align*}
If the density of $X$ conditional on $X\geqslant0$ is uniform on $[0,1]$, then 
\[
    E[Var(Y|X)|X\geqslant0]=\int_{0}^{1}Var(Y|X=x)dx=Var_{+} .
\]
Since $Var(E[Y|X]|X\geqslant0)\geqslant0$, we have $V_{+}\geqslant Var_{+}$, with the equality binding if and only if $E[Y|X]$ is a constant function almost surely for $X\geqslant0$. The same result applies if, instead of assuming a uniform distribution for $X$, we impose homoskedasticity: $Var(Y|X)=\sigma^{2}$ for $X\geqslant0$. 

To avoid the problem having too many bins, \citet{Calonicoetal2015} modify the mimicking variance formula and use the following bin number 
\[
    J_{+,N,MV}=\left\lceil \frac{V_{+}}{Var_{+}}\frac{N}{\log(N)^{2}}\right\rceil .
\]

\subsection{Mean Squared Error Decompositions\label{sec:MSE-Decomp}}

For each graph $T_{i}$ generated from GGP $(\gamma,g,d)$, let $d_{i}(T_{i}(\gamma,g,d))$ be the discontinuity estimate by participant $i$. The average visual discontinuity estimate for this GGP  is given by $\hat{d}(\gamma,g,d) \equiv \frac{1}{M}\sum_{i=1}^{M}d_{i}(T_{i}(\gamma,g,d))$. Meanwhile, let $d_{\hat{\theta}}(W_{i}(g,d))$ and $\hat{d}_{\hat{\theta}}(g,d)\equiv\frac{1}{M}\sum_{i=1}^{M}d_{\hat{\theta}}(W_{i}(g,d))$ be the corresponding estimates by the econometric estimator $\hat{\theta}$. We can easily compare the mean squared errors (MSEs) of the visual and econometric estimates, and we plot their square root in the left panel of Figure \ref{fig:RDD-Expert-Estimator-RMSE}. 

To understand the driving force behind the MSE comparisons, we can further decompose the MSEs into the square of the bias and the variance, by using the following identities:
\begin{align*}
    \frac{1}{M}\sum_{i=1}^{M}(d_{i}(T_{i}(\gamma,g,d))-d)^{2} & =(\hat{d}(\gamma,g,d)-d)^{2}+\frac{1}{M}\sum_{i=1}^{M}(d_{i}(T_{i}(\gamma,g,d))-\hat{d}(\gamma,g,d))^{2}\\
    \frac{1}{M}\sum_{i=1}^{M}(d_{\hat{\theta}}(W_{i}(g,d))-d)^{2} & =(\hat{d}_{\hat{\theta}}(g,d)-d)^{2}+\frac{1}{M}\sum_{i=1}^{M}(d_{\hat{\theta}}(W_{i}(g,d))-\hat{d}_{\hat{\theta}}(g,d))^{2}.
\end{align*}
Note that the interpretation of the second term on the right hand side of the first identity differs slightly from the conventional MSE decomposition, as the randomness here not only comes from the realization of data, but also from the participant ($d_{i}$). Nevertheless, it can still be conveniently understood as variance, as $d_{i}(T_{i}(\gamma,g,d))$ remains independent and identically distributed for different $i$. We use these two identities to decompose the difference between the visual and econometric MSEs  into two parts: one due to bias and one due to variance. To save space, we plot only the decomposition results for the visual-IK comparison in the right panel of Figure \ref{fig:RDD-Expert-Estimator-RMSE}, as IK dominates visual estimates in MSE for every DGP.

\section{Additional Details on DGP Specification and
Graph Creation\label{sec:creation-datasets-graphs-appendix}}

In conjunction with Section \ref{sec:creation-datasets-graphs} in
the main text, this online appendix describes the process by which we specify DGPs and create graphs for our experiments.

We identify a total of 110 empirical RD papers published in top economics journals including the \emph{American Economic Review}, \emph{American  Economic  Journals}, \emph{Econometrica}, \emph{Journal of Business and Economic Statistics}, \emph{Journal  of  Political  Economy}, \emph{Quarterly Journal of Economics}, \emph{Review of Economic Studies}, and \emph{Review of Economics and Statistics} between 1999 and 2017. We then randomly sample 11 papers from this list that have replication data available.

For the specification of the running variable for each DGP, we directly use the empirical distribution of $X$ from the corresponding paper and normalize the running variable to lie in $[-1,1]$. If the support of the running variable is asymmetric, this normalization may create a discontinuity in the density at $x=0$, which violates the assumptions in the identification framework of \citet{Lee2008} (but not the minimally sufficient identifying assumptions of \citealp{Hahnetal1999}). The original running variable in seven of our 11 papers fails the \citet{McCrary2008} test for continuity of the running variable density. Among our normalized running variables, eight fail the test, with two passing only in their non-normalized version and one passing only after normalizing. In our testing, visual inference performance does not vary significantly across DGPs whether the support of their original (unnormalized) running variables is symmetric. Following \citet{ImbensKalyanaraman2012} and \citet{Calonicoetal2014}, we remove any observations where $|X|>0.99$.

Next, we use the resulting datasets and further follow \citet{ImbensKalyanaraman2012} and \citet{Calonicoetal2014} to specify the DGP's CEF via a piecewise global quintic regression. Denoting the left and right intercepts of the CEFs by $\alpha^{-}$ and $\alpha^{+}$, we shift the right arm of the CEF upwards/downwards by the amount $|\alpha^{+}-\alpha^{-}|$ to make the functions $E[\tilde{Y}|X=x]$ continuous at the policy threshold.

Our specification of the distribution of the error term $u$ as i.i.d.\ normal similarly follows \citet{ImbensKalyanaraman2012} and \citet{Calonicoetal2014}. $u$ has mean zero and standard deviation $\sigma$, which is specified as the RMSE of the global quintic regression in the CEF specification above. We generate eight draws of $u$ for each DGP, enough so that every graph we generate for our experiments is seen by no more than one participant. Most of our papers use a continuous running variable, but two feature semi-discrete variables where the running variable takes on many unique values but these points have multiple observations each. In these two cases, we add small amounts of noise from a $N(0,(\frac{1}{\min\{N_{-},N_{+}\}})^{2})$ distribution, where $N_{-}$ and $N_{+}$ are the number of observations below and above the policy threshold, to  match the continuous running variable condition assumed in \citet{Calonicoetal2015}. We add this noise prior to fitting the piecewise quintic CEF. Perturbing the data prior to fitting the quintics introduces  ``measurement error'' that can attenuate estimated discontinuities. Alternatives include fitting the CEF prior to adding the noise or drawing the noise from a uniform distribution to prevent observations from crossing the policy threshold. For a discussion of measurement error in RD, see, for example, \citet{PeiShen2017}.

In practice, the difference between adding the noise before and after fitting the CEF is minor. Figure \ref{fig:Lineup-Protocols-Fudged-Unfudged} presents lineup protocols for these two DGPs. To test the null hypothesis that the DGP resulting from fitting the piecewise quintic prior to adding noise is indistinguishable from the DGP where noise is added first, we randomly place one graph from the former distribution among 19 from the latter. All graphs use a continuous running variable generated by adding noise to the original semi-discrete one. The solutions, i.e., the graphs from the DGP fit on the original data, are in the figure notes. We present both sets of CEFs in Figure \ref{fig:CEFs-Fudged-Unfudged}, with the scale of the $y$-axis chosen to match that of the corresponding graph in the original paper. Although the differences in the tails for the CEF on the right of Figure \ref{fig:CEFs-Fudged-Unfudged} underscore the sensitivity of the piecewise quintic specifications, the difficulty of these lineup protocols and the general similarity of the CEFs suggest that the injection of measurement error is not a major concern.

\section{Robustness to Alternative DGP Specifications\label{sec:dgp-robustness}}

To address concerns about the possible sensitivity of our results to our DGP-specification process, in March 2021, we conducted an additional experimental phase to compare non-expert performances across alternative DGP specifications (while holding fixed the graphical method). We test all four combinations of global quintic/local linear CEFs and homoskedastic/heteroskedastic noises.\footnote{We choose local linear to maximize the contrast from the global piecewise quintic specification. Using local cubic, for example, results in CEFs that looked very similar to the global quintic CEFs. We compare the local linear and local cubic CEFs in Figure \ref{fig:RDD-CEF-local}, and the global quintic CEFs are plotted in Figure \ref{fig:RDD-CEF} as mentioned in the main text.} The setup for these DGPs with respect to normalizing and trimming is identical to the procedures from the full paper. In all treatment arms, we use graphs with evenly spaced small bins with a vertical line at $x=0$, no fit lines, and Stata 14's default spacing. We present example plots across these treatments in Figure \ref{fig:RDD-local-examples}.

We set the zero-discontinuity dataset by subtracting from the RHS the fitted outcome values at $x=0$ based on the local fits on both sides of the treatment cutoff. We continue to specify discontinuities as multiples of the residual standard deviation from the piecewise quintic regressions to ensure comparability, so that the discontinuity levels are identical across alternative DGP specifications.

To allow for heteroskedasticity, we follow the approach by \citet{fan1998efficient} to estimate conditional variances of the residuals. Specifically, we estimate the conditional variance of $Y$ given $X$ by fitting local linear regressions to the squared residuals after we specify the local linear CEF. We multiply the error term with the square root of these conditional variance estimates.

We plot the four resulting power functions from this experiment in Figure \ref{fig:RDD-Power-RDD5}. None of the differences between treatment arms are statistically significant at any discontinuity magnitude, and the repeated treatment's power function is not statistically significantly different from the aggregated power function from the main four phases anywhere. We conclude that our empirical results on visual inference are not driven by the idiosyncrasies of our DGP-specification process.

\section{Design of Experiments and Studies\label{sec:Experiment-Details-App}}

This section outlines the structure of the experiments. We first describe the general sequence of events and survey structure before discussing the sequential rollouts of the experiments. All parts of this study received exemption confirmations from the Institutional Review Boards (IRBs) of Cornell University, Princeton University, Columbia Teachers College, and the University of Waterloo and were pre-registered on Open Science Framework and at the American Economic Association's RCT Registry. The supplemental phase 5 was additionally confirmed exempt by the IRB of the University of Delaware because it was conducted after one of the authors (Korting) changed her affiliation from Cornell. Expert studies run as part of seminars received separate exemption confirmations from their respective local IRBs. Before beginning any of the studies, participants completed consent forms. Participants had the opportunity to exit the study at any time.

\subsection{Non-Expert Experiment}

Each experiment consisted of three parts. In Part 1, participants watched a video tutorial outlining the graph construction (normalization of the running variable, binning the data, etc.) and stating the objective as classifying a discontinuity in the ``true'' underlying relationship at $x=0$. The video tutorial lasted between three to four minutes depending on the participant's graphical representation treatment. For example, participants in treatment arms featuring fit lines received a brief supplementary segment informing them that fit lines only serve as approximations of the true underlying relationship between the running variable and the outcome variable. To ensure participant engagement, the video tutorial featured an attention check about halfway through the video. This took the form of a colored bird (see example in Figure \ref{fig:attention-check-bird}); the voiceover on this slide informed participants that when asked about the color of our ``bird of interest,'' they should report a specific color that did not match the color of the bird in the picture. For example, participants seeing Figure \ref{fig:attention-check-bird} below might be asked to report that the bird is red, not blue. Participants had the opportunity to watch the video several times or to pause and rewatch specific sections as desired. However, participants were not able to return to the video tutorial once they had moved on to the next section.

The video tutorial was followed by a sequence of example questions. Participants were asked to classify tutorial graphs and received feedback on their answers. Participants could navigate between these graphs using a panel of buttons on the right-hand side of the screen (see Figures \ref{fig:example-tasks} and \ref{fig:example-tasks-ctd}).

Part 2 of the experiment consisted of a series of paid classification tasks. Participants were asked to assess whether a given graph featured a discontinuity or not and could choose between two potential bonus options for each graph. We discuss this bonus system in detail in Online Appendix \ref{sec:primary-experiment-bonus-scheme} below. Figure \ref{fig:example-classification-screen} shows a typical classification screen. Participants were not given any feedback regarding their accuracy or earnings between tasks.

Part 3 of the experiment was an exit survey. We solicited basic sample demographics such as age, gender, occupation, education level, and prior experience with statistics. For all phases after the initial pilot, we also asked participants to report whether they had participated in earlier iterations of the study. At the end of the experiment, participants were informed of the total number of graphs they classified correctly as well as their corresponding bonus earnings (see Figure \ref{fig:Result-Screen-Example} for an example). Experiment participants were paid in Amazon gift cards after each phase was complete.

\subsubsection{Experiment Implementation\label{sec:Appendix-Experiment-Timeline}}

The rollout of the experiments was staggered across five phases which were completed between November 2018 and March 2021. Phase 1 served as a pilot and was not part of the pre-registration package. The goal of this pilot was to gauge effect sizes in the case of the bin width and axis scaling treatments in order to determine the required sample size for subsequent (pre-registered) phases. Table \ref{tab:rdd-details-main} outlines the graphical representation treatments implemented in each phase. Each participant's classification task featured 11 graphs, one from each of the 11 specified DGPs. The order of graphs was randomized within subjects. The discontinuity levels were chosen as follows: two graphs featured no discontinuity, one graph featured the extreme discontinuity levels of $\pm1.5\sigma$, and the remaining eight graphs were based on each of the eight discontinuity levels $\pm0.1944\sigma,\pm0.324\sigma,\pm0.54\sigma,\pm0.9\sigma$. Each treatment arm contained 968 unique graphs split across 88 unique participants. Figure \ref{fig:Sequence-of-Events-RDD} outlines the sequence of events within an experiment.

\subsubsection{Participant Bonus Choice\label{sec:primary-experiment-bonus-scheme}}

In this section, we analyze a participant's bonus choice and discuss how we can use it to estimate her confidence in discontinuity classification and calculate the AS risk in Online Appendix \ref{sec:Risk}. First note that participants will choose a classification if and only if their perceived probability that the classification is correct ($\beta$ as defined in Online Appendix \ref{subsec:Decision-theory}) is at least 50\%. Denoting the utility function by $v$, the expected utility under the wager is given by $\beta v(\$0.4)+(1-\beta) v(\$0)$, which is equal to $\beta v(\$0.4)$ if we normalize $v(\$0)$ to zero. Expected utility maximizing participants will choose the sure amount whenever $\beta v(\$0.4)\leqslant v(\$0.2)$, i.e., whenever $\beta \in[0.5,v(\$0.2)/v(\$0.4)]$. A risk neutral subject (with a linear utility function) would therefore only choose the sure amount when they believe their chance of being correct is exactly 0.5 and choose the wager as soon as their confidence exceeds 0.5. If the distribution of $\beta$ is continuous, then few should choose the risk-free payment scheme. In the data, however, we see a number of people choosing the risk-free payment scheme.

An alternative behavioral model that is consistent with the observed bonus choices is loss aversion.\footnote{In principle, we can model risk aversion instead, but \citet{rabin2001anomalies} and \citet{o2018modeling} show that risk aversion over small stakes such as those in our wager implies implausible choices at higher stakes.} In this case, the utility function is replaced with a value function for gains and losses with respect to a reference point, and loss aversion is governed through a parameter $\lambda>1$ multiplying outcomes in the loss domain. In the context of our wager, the sure option of 20 cents serves as a natural focal point to participants, so we consider the simple loss aversion framework put forward by \citet{KahnemanandTversky1979} assuming a fixed reference point of 20 cents. That is, participants would consider the gamble as an opportunity to win 20 cents or lose 20 cents relative to this reference point depending on their answer (as mentioned in Online Appendix \ref{subsec:Decision-theory}, an expected utility maximizing reader's discontinuity decision, unlike her bonus choice, is not affected by loss aversion). 

The expected payoff of the wager in this case is given by $\beta v(\$0.2)-\lambda (1-\beta) v(\$0.2)$, compared to zero under the fixed payment option. Participants choose the wager whenever $\beta \geqslant\lambda/(1+\lambda)$. Therefore, for a given value of $\lambda$, a player's bonus choice indicates the interval in which $\beta$ falls, which we can use to approximate the communication risk by \citet{andrews2020model}. \citet{brown2021meta} conduct a meta-analysis of 607 empirical loss-aversion estimates across 150 studies and find that the average coefficient $\lambda$ lies between 1.8 and 2.1. For simplicity, we abstract away from heterogeneity in $\lambda$ and choose a loss-aversion parameter of $\lambda=2$ throughout our analysis. However, as discussed in Online Appendix \ref{sec:Risk} below, choosing alternative values for $\lambda$ does not alter our graphical method recommendation.

\subsection{Expert Study}

The expert study was run both online recruiting NBER members in applied microeconomic fields (aging, children, development, education, health, health care, industrial organization, labor, and public) and IZA fellows/affiliates, and during three seminar presentations between May and October of 2019. Table \ref{tab:expert-details-appendix} shows the timeline and number of participants for each session. Figure \ref{fig:Sequence-of-Events-Experts} outlines the sequence of events for the expert study.

Unlike in the non-expert experiments, participants in the expert study were not asked to watch a video tutorial or complete practice questions at the start. Instead, Part 1 of the expert study was a classification task paralleling the one non-expert participants completed (see Figures \ref{fig:Expert-Survey-Part1-Intro} and \ref{fig:Expert-Survey-Part1}). In addition, the graphical representation choices in this study were not randomized between subjects. Participants had the opportunity to navigate between the 11 graph classifications using a panel of navigation buttons at the bottom of the screen. This feature was introduced to address feedback we received during the testing of the survey. Figure \ref{fig:Expert-Survey-Part1} shows an example of the task screens for participants.

Part 2 of the expert study asked participants to rank four alternative graphical representation options for RDD. The four treatment combinations corresponded to small or large bins interacted with both fit line treatments, as in phase 4.  The instructions for Part 2 of the expert study are outlined in Figures \ref{fig:Expert-Survey-Part2-Intro} and \ref{fig:Expert-Survey-Part2-Decision-Screen}. Figure  \ref{fig:RDD-Power-Phase-VI-DGP9} shows the power functions for the DGP used in the example graphs in this part of the expert study. We select this DGP based on its visual inference performance aligning most closely with the average performance over all DGPs. Participants were asked to specify which graphing treatment they believe researchers should use to present evidence of the main treatment effect of a study, as well as which graphical options they believe would perform best and worst in our non-expert sample. Participants were also able to report having ``no preference'' across the different graphing options. These three decisions were repeated three times in the context of a zero, small, and large discontinuity.

Part 3 of the expert study was an exit survey. We asked about basic sample demographics such as age, gender, region, main area of research, and prior experience with RDD. After all sessions of the expert study were complete, we randomly selected four participants to receive a payment of \$450 plus \$50 for every correctly classified graph in Part 1 of the study. We did not take the accuracy of the discontinuity magnitude estimate into account to determine payments, only the binary classification decision.

Table \ref{tab:expert-details-appendix} highlights the different participant pools for the sessions of the expert study and lists the graphical representation choices in the classification task (Part 1 of the expert study). All graphs in Part 1 used the default Stata 14 axis scaling with no fit lines and a vertical line at $x=0$. All but one session of the expert study used small bins. The exception was the session run at Princeton's Quantitative Social Science Colloquium in October 2019 in which we used large bins to compare the effect of bin widths on experts and non-experts.

The distributions of graphs seen by the experts and non-experts differ slightly due to the methods used to randomize participants' graphs. Non-expert randomization is based around the order in which the 88 participants for each treatment arm begin the survey, i.e., we randomized treatment order in advance and then assigned each participant a participant-specific survey based on their arrival time. We use the same randomization mechanism for experts, but repeating the survey in different seminars and online phases results in more experts at the ``beginning'' of the randomized assignments. In addition, a few experts did not complete the study after accessing their participant-specific survey, and we removed their partial responses. Therefore, although all 88 datasets are represented over the entire power function of expert visual inference, this is not true at each discontinuity magnitude: there are 85, 82, 84, 84, 87, and 69 datasets at the $0$, $0.1944\sigma$, $0.324\sigma$ $0.54\sigma$, $0.9\sigma$, and $1.5\sigma$ discontinuity magnitudes, respectively (recall that we only have one graph with $1.5\sigma$ discontinuity for every two graphs with other discontinuity magnitudes). Reweighting the data to match the distribution of graphs seen by non-experts does not meaningfully change any of our results.

\section{Evaluating Graphical Methods Using Frameworks from Online Appendix \ref{subsec:Decision-theory} \label{sec:Risk}}

In this section, we provide additional evidence in support of our recommendation of using small bins, no fit lines, even spacing, default $y$-axis scaling, and a vertical line at the policy threshold as a starting point in RDD graphical analysis. In particular, we empirically implement the two decision theoretic frameworks from Online Appendix \ref{subsec:Decision-theory}, which further synthesize our experimental results.

First, we estimate the overall classical risk of each graphical method. For a given $d$, the risk $\bar{\mathcal{R}}(\gamma,d)$ is a simple transformation of the type I or type II error probability and is effectively summarized by the power functions presented in Figure \ref{fig:RDD-Power}. To facilitate the risk comparisons across $\gamma$, however, we will need to aggregate the information in each power function across $d$. Formally, as mentioned in Online Appendix \ref{subsec:Decision-theory}, this task requires the specification of the type I and type II error cost parameters---$\kappa$ and $\varphi$, respectively---and a prior on $d$ that represents the probability of encountering a graph with a certain discontinuity level. Because their choices involve subjective judgement, we present results under different scenarios, and our power functions also allow researchers to estimate the risks with their preferred specification of these quantities. For the cost parameters, we normalize $\varphi$ to 1 (only the ratio of the cost parameters matters for comparing risks) and try both the baseline case of $\kappa=1$ and a benchmark used by \citet{kline2019audits}, $\kappa=4$. For the prior on $d$, we adopt the uninformative prior that it is equally likely to encounter a continuous graph as a discontinuous graph, but there is still the question of how to incorporate the type II error rates across different discontinuities. In our main estimates presented in Table \ref{tab:risks-main}, we use the type II error rate at $0.324\sigma$, the modal (and median)---and therefore the most likely---discontinuity level among the subset of the 11 papers that report a significant discontinuity. In Online Appendix Table \ref{tab:risks-supp}, we also present risk estimates by simply averaging the type II error rates across the five different nonzero discontinuity levels, which lead to the same recommendation.

When weighting type I and II errors equally, the classical risks are fairly similar across treatments. In fact, no risk is statistically significantly different from that of the benchmark small bins/no fit lines treatment (in bold) repeated across experimental phases. When penalizing type I errors more in the \(\kappa=4\) calculations, the benefit of the low type I error rate of the benchmark treatment is more evident. The benchmark treatment has the lowest risk in two of the four phases (phases 2 and 3) and the second-lowest in the other two, where the differences with the best performer are not statistically significant. Further, the large bin treatments in the three phases where they appear (phases 1, 2, and 4) have statistically significantly higher classical risks, as does the fit lines treatment in phase 3.

Second, we estimate the AS communication risks per \citet{andrews2020model}. To do this, we need to measure the perceived probability of being correct, $\beta$, for each participant-graph combination. As noted in Online Appendix \ref{sec:primary-experiment-bonus-scheme}, a participant's bonus scheme choice provides information on $\beta$. Assuming a loss aversion parameter of 2, $\beta$ falls into the interval $(2/3, 1)$ if the participant chooses the $(\$0.4,\$0)$ wager, and it falls into $(1/2,2/3)$ if she chooses the risk-free $(\$0.2,\$0.2)$ option. For the estimates of the AS risk we present, we simply approximate $\beta$ by the midpoint of the two intervals, which is a common way of handling interval data in the empirical literature. The resulting averages of $\beta$ lie between $7/12$ and $5/6$. We plot these values in Figure \ref{fig:Subj-Prob-CorrectBy-DGP} with inference on the differences across treatment arms in Tables \ref{tab:TE-conf-RD-I} through \ref{tab:TE-conf-RD-VI}. We note that the midpoint approximation is coarse and may underrepresent the true curvature in Figure \ref{fig:Subj-Prob-CorrectBy-DGP}---the average $\beta$ at the highest and most obvious discontinuity magnitude is likely very close to one as opposed to $5/6$. That said, the midpoint approach has a great advantage: it uses a simple and transparent mapping from the bonus scheme choices to reader confidence: the risky choice is mapped to $\beta=5/6$ and the risk-free choice to $\beta=7/12$. As a result, using the the midpoint approach to rank graphical methods is equivalent to just using the binary bonus scheme choice, as it uses the same information without imposing untestable structure on the unobserved belief distribution among readers. As an added benefit, the midpoint approach leads to a quantity that is interpretable through the AS framework. This simple mapping also means that the recommended graphical method based on the midpoint approach should not depend on the particular value we choose for the loss aversion parameter, $\lambda$. Finally, we have also tried alternative functional forms for the underlying distribution of $\beta$ and reach similar conclusions.

Approximating $\beta$ using participants' bonus scheme choices allows us to estimate the AS risk $\bar{\mathcal{R}}_{AS}(\gamma,d)$ for each $(\gamma,d)$ combination. But we still need to aggregate across $d$ so that we can compare the overall risks across $\gamma$. We simply proceed with the same weighting scheme as with our classical risks.

We present these results in the bottom half of Table \ref{tab:risks-main}, and the repeated benchmark small bins/no fit lines treatment (in bold) again performs well. 
When giving the same weight to the AS risk at zero and nonzero discontinuities, the benchmark treatment has the lowest risk in phase 2 and the second-lowest in phases 1, 3, and 4. With greater weights at zero discontinuity, it has the lowest risk in phases 1 and 2. As with the classical risks, no treatment achieves a statistically significantly lower AS risk than small bins/no fit lines, and many are significantly outperformed in either weighting scheme.

\section{Subject Characteristics and Performance Predictors\label{sec:Appendix-Additional-RDD-Analysis}}

\subsection{Demographic Balance and Effects}

For our non-expert experiment, we check the validity of our randomization by testing for the balance of covariates across treatment groups. For each phase, we regress the covariates on treatment group indicators and test whether all coefficients are equal. Specifically, we test for the balance of sex, education, age categories, statistical knowledge, passing the attention check, being a first-time participant, and not completing the experiment. The results are in Tables \ref{tab:Cov-Bal-RD-I} through \ref{tab:Cov-Bal-RD-VII}. Across these 44 hypotheses, we cannot reject balance in all but two instances (4.5\% of tests) when using a 5\%-level test: college completion and fraction of participants aged 23-49 in phase 4. In addition, $p$-values for joint tests of significance for all covariates within each phase based on \citet{pei2019poorly} are 0.18, 0.50, 0.65, and 0.14 for phases 1-4. The evidence is consistent with successful randomization.

As can be seen in Table \ref{tab:Cov-Prediction-RD-I-II-III-VI}, participant demographics such as gender, age, college completion, statistical knowledge, and being a repeat participant are generally not predictive of performance in the experiment, with insignificant point estimates of a few percentage points. Out of the 34 factors we test, two have a statistically significant effect on the probability of classifying a graph correctly (5.9\% of tests). In phase 1, failing the attention check is associated with a 5.5pp decrease in the probability of classifying a graph correctly. This effect is smaller and insignificant in phases 2 (4.3pp), 3 (-1.2pp), and 4 (-1.2pp). In phase 3, being 50 or older correlates with a 9.4pp increase in the probability of classifying a graph correctly. This lack of predictive strength makes any imbalances in demographics across phases an even smaller concern. It also suggests that comparisons of treatments across phases are unlikely to be confounded by differing subject pools.

\subsection{DGP Characteristics}

While our results in Section \ref{sec:Results} focus on the treatment effects of each graphical parameter, in this section we focus on other factors of inference performance. In particular, we explore how the characteristics of an RD DGP affect visual inference.

Figure \ref{fig:RDD-DGP-Histograms} plots the density of each DGP along its running variable. These distributions are varied. Some are approximately uniform, some are more normal, and others, such as DGPs 6 and 8, are highly skewed. Some DGPs show clear changes in density at the policy threshold, some of which were present in the original data and some of which were introduced by having asymmetric supports prior to our rescaling the running variable. Figure \ref{fig:ROC-Phase-6-Symmetric-Asymmetric} compares power functions for the eight DGPs with symmetric supports prior to normalization with the three with asymmetric supports in phase 4. There is a noticeable flatness or dip in the share of participants reporting a discontinuity between 0 and $0.1944\sigma$. This pattern appears in every phase, but with only three asymmetric DGPs, this may also be due to the effects of the individual DGPs rather than something about the effects of normalizing the running variable when asymmetry is present.

When the distribution of the running variable is uniform, there is no difference between graphs generated with even spacing and quantile spacing. When the distribution is not uniform, we may expect quantile spacing to perform better, for example by preventing outlier bins with very few observations in more sparsely populated regions of the support. To quantify a DGP's deviation from a uniform distribution, we can compute Gini coefficients based on the distribution of observations within evenly spaced bins.\footnote{We do not compute Gini coefficients with quantile spacing, as all bins have about the same number of observations and the Gini coefficients would all be approximately zero.} We calculate Gini coefficients for both the IMSE and MV bin width algorithms and take the average across data realizations for each DGP. Figure \ref{fig:Gini-Spacing-Performance} explores how the Gini coefficient interacts with bin spacing to test whether quantile spacing outperforms equal spacing when the distribution is more skewed by averaging the probability of a correct response across treatment arms within each DGP. There is no relationship between a DGP's Gini coefficient and performance across spacing types. Among DGPs with higher distributional inequality, as well as for those with intermediate and lower values, neither spacing dominates the other.

In Table \ref{tab:DGP-Direction-Prediction-Power}, we examine how the first derivatives of each arm of the DGP interact with the direction of the discontinuity in impacting type II error rates. Specifically, we ask whether visual inference performs better if the discontinuity is positive or negative within a combination of left and right slope signs. Looking only at graphs with nonzero discontinuities, with which we measure type II error rates, we regress the binary classification decision on dummy variables for the negative discontinuity sign within each slope combination (omitting the positive discontinuity dummy) with DGP fixed effects and clustering at the participant level. These regressions show that visual inference performs better when the discontinuity direction matches the signs of both the left and right derivatives. For example, when the left and right first derivatives are both positive, the type II error rate is 11 percentage points higher when the discontinuity is negative than positive (see Figure \ref{fig:Slope-Discontinuity-Direction} for an illustration). Table \ref{tab:DGP-Direction-Prediction-Size} repeats a similar exercise for type I error rates over the no-discontinuity graphs. We cannot include DGP fixed effects here, as they would eliminate the variation needed to identify impacts. Unlike with type II error rates, we do not find any meaningful results for type I error rates.

\subsection{Dynamic Visual Inference\label{subsec:Dynamic-Visual-Inference}}

In this subsection, we investigate participants' performance over the course of the 11 graphs they see. This exercise provides evidence to alleviate the concern that participants have a preconceived notion that about half of the graphs feature a discontinuity, which may bias our results in favor of good control of type I error rate by visual inference.

First, we trace out the fraction of participants reporting a discontinuity over their 11 graphs. If there is a consistent upward or downward trend, it would be indicative of participants having a fixed number of continuous graphs in mind. As shown in Figure \ref{fig:RDD-classifications-over-time}, however, the trends are flat in the treatment arm (small bins and no fit lines) used to compare visual and econometric inference procedures over the main four non-expert phases and in the expert sample.\footnote{Because we randomize the graph order, the flat discontinuity classification rate is consistent with a flat correct classification rate, which we also verify in the data. That is, no evidence supports \textquotedblleft learning\textquotedblright{} over the 11 graphs in terms of visual inference success, which is concordant with our design in which participants find out the total number of correct discontinuity classifications only after they complete the study.}

We also implement an AR(1) regression to more closely examine the dynamic discontinuity classification over the 11 graphs, which Figure \ref{fig:RDD-classifications-over-time} may miss:
\begin{equation}
    R_{is}=\rho R_{i(s-1)}+c_{i}+\epsilon_{is}.\label{eq:AR1}
\end{equation}
$R_{is}$ is the binary variable indicating whether participant $i$ reports a discontinuity in graph $s$, $c_{i}$ is the individual fixed effect, and $\epsilon_{is}$ is the population residual from the projection of $R_{is}$ on $R_{i(s-1)}$ and $c_{i}$. We expect a negative $\rho$ if participants believe ex-ante that half of the graphs will feature a discontinuity. In fact, we do see negative estimates of $\rho$ ranging between -0.1 and -0.2 as reported in column (1) of Table \ref{tab:dynamic-discontinuity-table}. However, these estimates are subject to a large-$N$ asymptotic bias of order $O(1/S)$ per \citet{Nickell1981} where $S$ is the number of graphs. Inverting equation (18) from \citet{Nickell1981}, we solve for the bias-corrected estimates of $\rho$ and report them in column (2) of Table \ref{tab:dynamic-discontinuity-table}. Although these estimates are still negative, they are considerably closer to zero. For the main four phases of the non-expert phases 1-4, three bias-corrected estimates are statistically insignificant at the 5\% level, while the other one is marginally significant.

Although the estimate of $\rho$ in the expert sample is still significant after bias correction, the fact that it is the largest in magnitude of the five estimates in Table \ref{tab:dynamic-discontinuity-table} is reassuring. Because non-experts appear to achieve lower a type I error rate than experts and are less experienced with RDD graphs, we may be more concerned that the non-experts' performance is driven by having a fixed number of continuous graphs in mind. However, this is refuted by the four estimates of $\rho$ from the non-expert sample being smaller. Therefore, we conclude that the evidence does not support the hypothesis that participants expect half of the graphs to be continuous, and the low type I error rate of RD visual inference is not a result of it.

\section{What \emph{t}-Statistic Can Eyes Detect?\label{sec:t-stat}}

Using our experimental results, we can provide an evidence-based answer to the question of how large a $t$-statistic needs to be in order for readers to visually detect a discontinuity (a question that was also posed by Kirabo Jackson on Twitter). We do so by using an alternative scaling of the discontinuity. Instead of specifying the discontinuity $d$ as a multiple of the error standard deviation $\sigma$ as in Section \ref{sec:Results}, we specify the $x$-axis of the power functions in Figures \ref{fig:RDD-Power-Rescaled} and \ref{fig:RDD-Power-By-DGP-Rescaled} as $d/\sigma\sqrt{(X^{\prime}X)_{dd}^{-1}}$, where $X$ is a regressor matrix containing the 12 polynomial terms of the running variable for a piecewise quintic regression and the
$dd$ subscript denotes the entry of the matrix corresponding to the discontinuity estimator. Because the denominator is the large-sample error of the discontinuity estimator from a (correctly specified) piecewise quintic regression, the $x$-axes of these figures have a $t$-statistic interpretation. Note that because the rescaling results in a distinct set of six discontinuity magnitudes for each of our 11 DGPs, we now have 66 discontinuity magnitudes. Therefore, we bin points along the $x$-axis, and consequently,
these power functions are no longer monotone, as the composition of DGPs changes across bin points.

Our analysis shows that 80\% of the participants correctly report a discontinuity, across all phases of our experiment, as the $t$-statistic from the correctly specified regression exceeds 4 (80\% is a threshold commonly used in power analysis). There are caveats to this already nuanced answer. First, this threshold $t$-statistic is specific to the range of parameters we test. For example, one could substantially increase the $y$-axis scale, which would be sure to make any discontinuity we test in this paper disappear before the human eye. Second, this threshold is specific for detecting discontinuities; in separate experiments for regression kink designs presented in our previous working paper \citep{kortingetal2020wp}, we find a much smaller threshold applies to kink detection.

One might conjecture that the $t$-statistic threshold is sensitive to the sample size. In particular, it seems possible to keep the plot fixed and increase the sample size arbitrarily by adding observations in each bin, which would increase the $t$-statistic without altering visual inference. However, for the plot to look the same while increasing sample size, one needs to change $\sigma$ as well. In fact, $\sigma$ needs to increase proportionally to $\sqrt{N}$, leading to an invariant $t$-statistic. 

\section{Visual versus Econometric Inference: Details}\label{sec:visual_v_metrics_details}

In this section, we provide details on the comparisons between visual and econometric inferences. First, we discuss in more depth the three nonparametric econometric inference methods---IK, CCT, and AK---and present additional empirical results on their performances vis-a-vis visual inference. Second, in the subsection below, we describe how we construct the confidence intervals for the differences between visual and econometric inferences. 

The IK inference method refers to the practice of using a local linear estimator with the IK bandwidth selector to estimate the treatment effect and ignoring the asymptotic bias when constructing the confidence interval. Specifically, the IK bandwidth aims to minimize the asymptotic MSE and has the form: 
\[
\hat{h}_{IK}=C_{IK}\cdot\left(\frac{\hat{\sigma}^{2}(0^{+})+\hat{\sigma}^{2}(0^{-})}{\hat{f}(0)\cdot(\hat{\mu}^{(2)}(0^{+})-\hat{\mu}^{(2)}(0^{-}))^{2}+\hat{r}}\right)^{1/5}\cdot N^{-1/5}
\]
where $C_{IK}$ is a constant determined by the kernel, $f(\cdot)$ is the density of the running variable, $\mu^{(2)}(0^{\pm})$ and $\sigma^{2}(0^{\pm})$ represent the second derivatives of the CEF and conditional variances on both sides of the threshold, $r$ is a regularization term to prevent very large bandwidths, and the hats on the various quantities indicate that they are estimated. A local linear estimator $\hat{\tau}$ with bandwidth $h$ for the RD treatment effect parameter $\tau$ has the asymptotic distribution
\[
\sqrt{Nh}(\hat{\tau}-\tau-h^{2}B)\Rightarrow N(0,\Omega).
\]
$B$ is the asymptotic bias constant, which is a function of the second derivatives, $\mu^{(2)}(0^{+})$ and $\mu^{(2)}(0^{-})$. When $h$ shrinks at the MSE-optimal rate of $h=O(N^{-1/5})$, as is the case for IK, the asymptotically normal distribution of $\hat{\tau}$ is centered at the constant $\tau+\sqrt{Nh^{5}}B$. Ignoring the bias term $\sqrt{Nh^{5}}B$ in constructing the confidence interval can be justified on the ground that $B$ is small locally, e.g., when the CEF is approximately piecewise linear around the threshold. Or it can be justified on the ground that ``undersmoothing'' using a bandwidth with a slightly larger shrinkage rate would eliminate the bias term but make little difference empirically. But it is harder to justify the ``conventional'' inference method in the case where $\hat{h}_IK$ is large, and \citet{Calonicoetal2014} document that the resulting confidence interval may have coverage rates considerably below the nominal level.

The CCT inference procedure constructs confidence intervals by directly incorporating the asymptotic bias term. It first estimates the bias $B$ using another (``pilot'') bandwidth $b$, which is MSE-optimal for estimating $\hat{\mu}^{(2)}(0^{+})-\hat{\mu}^{(2)}(0^{-})$. Then, it creates a bias-corrected estimator $\hat{\tau}^{bc}=\hat{\tau}-h^{2}\hat{B}$, whose asymptotically normal distribution is centered at the true RD treatment effect $\tau$. Finally, it adopts an innovative asymptotic framework, in which $h/b$ converges to a positive constant ($h/b$ converges to zero under standard asymptotics), to account for the sampling variation in the estimator $\hat{B}$. 

The AK inference procedure produces asymptotically valid and minimax (near-)optimal confidence intervals over the Taylor class of conditional expectation functions
\[
\left\{ \mu_{\pm}:|\mu_{\pm}(x)-\mu_{\pm}^{\prime}(0)x|\leqslant C_{T}|x|^{2}\text{ for all }x\in\chi_{\pm}\right\} .
\]
$\chi_{\pm}$ is the support of the running variable on two sides of 0, and $C_{T}$ is a researcher-supplied bound on the second derivatives. $C_{T}$ is the key tuning parameter and dictates the bias magnitude of a local linear estimator---unlike CCT, AK's bias correction is nonrandom. We use the default rule-of-thumb procedure in the \texttt{RDHonest} package to estimate the tuning parameter.\footnote{\citet{ArmstrongKolesar2017} propose analogous confidence intervals that maintain coverage and enjoy minimax optimality over a H\"{o}lder class of functions, which is determined by a global, as opposed to local, bound on the second derivative of the CEF. Though not presented here, we find the corresponding power functions to be similar. Like \citet{ArmstrongandKolesar2017,ArmstrongKolesar2017}, \citet{imbens2019optimized} also adapt the idea of \citet{donoho1994statistical}. They propose an RD estimator through numerical optimization that is minimax mean-squared-error optimal over CEFs with a global second derivative bound. Because the corresponding inference procedure performs similarly to \citet{ArmstrongandKolesar2017} in simulations by \citet{Peietal2018local} in their 2018 working paper version and can be computationally demanding, we do not implement the \citet{imbens2019optimized} procedure here.} 

In the remainder of this section, we present two additional sets of empirical results. First, we show power functions of the econometric inference methods where we impose our knowledge of the DGP. This exercise sheds light on the factors underlying the econometric methods' performances. Second, we compare the accuracy of the visual and econometric estimates of discontinuity magnitudes. The comparison goes beyond the binary classification we focus on in much of the paper.

We present the first set of additional results in Online Appendix Figure \ref{fig:RDD-Power-Experts-vs-Theoretical}, where we impose our knowledge of the DGP when implementing the econometric inference procedures. Specifically, we use the theoretical MSE-optimal bandwidth for IK, this bandwidth and the theoretical asymptotic estimator bias for CCT, and the true second derivative bound for AK.\footnote{Because we specify the $X$-distribution as the empirical running variable distribution from the published study, we need to estimate its density at $0$ in computing the theoretical MSE-optimal bandwidth, and it is the only quantity we estimate.} While the AK result is similar to when we use the estimated turning parameter (it is possible that our quintic specification is friendly to the rule-of-thumb estimator, which relies on a quartic regression), the IK type I error rate is considerably lower, leaving little for CCT to improve upon. Although the theoretical MSE-optimal bandwidth is still ``too large'' based on its $N^{-1/5}$ rate of shrinkage, the driving force of IK's high type I error rate appears to be the noisy estimates of the constants in the bandwidth formula (similar patterns also emerge from the simulation results in \citealp{Calonicoetal2014}). 

We present the second set of additional results in Figure \ref{fig:RDD-Expert-Estimator-RMSE}, where we compare the point-estimate accuracy across methods. We plot the RMSE of each method by DGP, averaged over all discontinuity magnitudes. We do this by DGP to prevent scaling issues, as units for each DGP are different and results from one DGP are not directly comparable to those from another. We also take steps to make human and econometric point estimates comparable. Because we ask participants to round their estimates to the nearest hundredth, we round estimators with the same precision. We similarly replace econometric estimates with 0 when the test fails to reject the null hypothesis.

Although there are a handful of DGPs where experts perform similarly to the econometric estimators, their point estimates generally have a greater RMSE than the estimators, and experts overall are worse at estimating magnitudes. Specifically, IK has a lower RMSE than visual inference for every single DGP, which is driven by the variance component of the MSE in a majority of cases as shown in the right panel of Figure \ref{fig:RDD-Expert-Estimator-RMSE} (see Online Appendix \ref{sec:MSE-Decomp} for details of decomposing the MSE differences). In summary, while the average human performs quite well at identifying the existence of discontinuities, her ability to estimate their magnitude is not as strong.

\subsection{Inference on the Difference between Visual and Econometric Inferences\label{subsec:two-way-clustering}}

We describe the construction of the confidence intervals, which are plotted, for example, in Figure \ref{fig:RDD-Experts-Vs-All-Diff}. At each discontinuity level, we estimate the difference between visual and econometric inferences via regressions where we stack visual and econometric classifications for each graph of that discontinuity. With stacking, we can construct standard errors that account for potential correlation across visual and econometric classifications for the same graph and potential correlation in visual classifications of different graphs by the same participant.

We now detail the standard error calculation. We use $R_{i}^{k}$ to denote the classification of graph $k$ by participant $i$ and  $R_{\hat{\theta}}^{k}$ to denote the classification of (the dataset underlying) graph $k$ by econometric procedure $\hat{\theta}$. At each discontinuity magnitude, we regress the vector $(R_{i}^{k},R_{\hat{\theta}}^{k})^{\prime}$ on a constant and an indicator for visual inference, i.e., the regressor matrix $\left(\begin{array}{cc}
1 & 1\\
1 & 0
\end{array}\right)$. The difference between visual and econometric inferences is given by the coefficient on the visual inference indicator. We compute the standard error of the difference by two-way clustering per \citet{cameronetal2011main}. We cluster by dataset and by participant identifier (each human participant has a participant identifier, and we assign a unique artificial ``participant identifier'' to each observation corresponding to econometric classification that is different from any of the human participant identifiers). Clustering by dataset accounts for correlation between visual and econometric inferences, and clustering by participant identifier captures correlation of visual inferences by the same individual.

We also apply this procedure to generate the confidence intervals on the differences between the combined visual-IK and AK procedures plotted in Figure \ref{fig:RDD-Visual-Econometric-Inference-Combined-vs-AK}. For each graph $k$, we simply replace the visual classification $R_{i}^{k}$ with the combined visual-IK classification, and the rest of the confidence interval construction proceeds in the same way.

\clearpage{}






\newpage{}

\section*{Online Appendix Tables and Figures}

\begin{table}[H]
    \caption{\label{tab:TE-RD-I}Effects of Graphical Methods on Visual Inference: Phase 1}
    \centering\resizebox*{4.0in}{!}{%
\begin{tabular*}{1.0\hsize}{@{\hskip\tabcolsep\extracolsep\fill}l c c c c c c} 
\vspace{-.45cm} 
&\multicolumn{6}{c}{Dependent variable: player reports discontinuity} \\ 
\cline{2-7} \\
&\multicolumn{6}{c}{True discontinuity magnitude = } \\ 
\vspace{-.45cm} 
 &\multicolumn{6}{c}{} \\
 &    0 &  0.1944 $\sigma$  & 0.324 $\sigma$ & 0.54 $\sigma$ & 0.9 $\sigma$ &  1.5 $\sigma$   \\
\cline{2-7} \\
Large bins; Default y-axis   & 0.197 & 0.209 & 0.177 & 0.140 & -0.024 & -0.012 \\
  & (0.039) & (0.045) & (0.042) & (0.046) & (0.030) & (0.019) \\
Small bins; Large y-axis   & -0.024 & -0.018 & -0.056 & 0.008 & -0.034 & -0.035 \\
  & (0.029) & (0.043) & (0.042) & (0.043) & (0.029) & (0.026) \\
Large bins; Large y-axis   & 0.142 & 0.186 & 0.120 & 0.093 & 0.013 & -0.001 \\
  & (0.036) & (0.040) & (0.042) & (0.042) & (0.025) & (0.017) \\
Small bins; Default y-axis (Mean)   & 0.054 & 0.181 & 0.367 & 0.669 & 0.922 & 0.988 \\
Number of graphs  & 660 & 660 & 660 & 660 & 660 & 330 \\
Number of players  & 330 & 330 & 330 & 330 & 330 & 330 \\
\cline{1-7} \\
\end{tabular*} 
}

    \begin{minipage}[t]{0.8\columnwidth}%
        {\scriptsize{}Notes: We have a stratified randomized experiment, where each of the 11 strata is determined by the DGPs seen for every discontinuity magnitude. The treatment effect estimates come from regressing the participants' responses on treatment indicators and stratum fixed effects. We obtain standard errors using the procedure from \citet{Bugnietal2019}
        for stratified experiments. \textit{Large bins} corresponds to the \citet{Calonicoetal2015} bin width selector that minimizes the integrated mean squared error of the bin-average estimators of the conditional expectation function; \textit{small bins} corresponds to the \citet{Calonicoetal2015} bin width selector that aims to approximate the variability of the underlying data; \textit{default} $y$-axis scaling corresponds to Stata 14's default axis scaling; \textit{large} $y$-axis scaling doubles this default axis scaling.}%
    \end{minipage}
\end{table}

\begin{table}[H]
    \caption{\label{tab:TE-RD-II}Effects of Graphical Methods on Visual Inference: Phase 2}
    \centering
    \resizebox*{4.0in}{!}{%
\begin{tabular*}{1.0\hsize}{@{\hskip\tabcolsep\extracolsep\fill}l c c c c c c} 
\vspace{-.45cm} 
&\multicolumn{6}{c}{Dependent variable: player reports discontinuity} \\ 
\cline{2-7} \\
&\multicolumn{6}{c}{True discontinuity magnitude = } \\ 
\vspace{-.45cm} 
 &\multicolumn{6}{c}{} \\
 &    0 &  0.1944 $\sigma$  & 0.324 $\sigma$ & 0.54 $\sigma$ & 0.9 $\sigma$ &  1.5 $\sigma$   \\
\cline{2-7} \\
Large bins; Equal spacing   & 0.256 & 0.207 & 0.204 & 0.155 & 0.048 & 0.011 \\
  & (0.042) & (0.039) & (0.045) & (0.039) & (0.033) & (0.029) \\
Small bins; Quantile spacing   & 0.034 & 0.009 & -0.012 & -0.008 & -0.030 & -0.013 \\
  & (0.033) & (0.039) & (0.046) & (0.046) & (0.036) & (0.033) \\
Large bins; Qunatile spacing   & 0.157 & 0.185 & 0.217 & 0.138 & 0.030 & 0.024 \\
  & (0.038) & (0.040) & (0.046) & (0.039) & (0.031) & (0.023) \\
Small bins; Equal spacing (Mean)   & 0.054 & 0.169 & 0.355 & 0.657 & 0.892 & 0.964 \\
Number of graphs  & 650 & 650 & 650 & 650 & 650 & 325 \\
Number of players  & 325 & 325 & 325 & 325 & 325 & 325 \\
\cline{1-7} \\
\end{tabular*} 
}

    \begin{minipage}[t]{0.8\columnwidth}%
        {\scriptsize{}Notes: We have a stratified randomized experiment, where each of the 11 strata is determined by the DGPs seen for every discontinuity magnitude. The treatment effect estimates come from regressing the participants' responses on treatment indicators and stratum fixed effects. We obtain standard errors using the procedure from \citet{Bugnietal2019} for stratified experiments. \textit{Large bins} corresponds to the \citet{Calonicoetal2015} bin width selector that minimizes the integrated mean squared error of the bin-average estimators of the conditional expectation function; \textit{small bins} corresponds to the \citet{Calonicoetal2015} bin width selector that aims to approximate the variability of the underlying data.}%
    \end{minipage}
\end{table}

\begin{table}[H]
    \caption{\label{tab:TE-RD-III}Effects of Graphical Methods on Visual Inference: Phase 3}
    \centering\resizebox*{4.0in}{!}{%
\begin{tabular*}{1.0\hsize}{@{\hskip\tabcolsep\extracolsep\fill}l c c c c c c} 
\vspace{-.45cm} 
&\multicolumn{6}{c}{Dependent variable: player reports discontinuity} \\ 
\cline{2-7} \\
&\multicolumn{6}{c}{True discontinuity magnitude = } \\ 
\vspace{-.45cm} 
 &\multicolumn{6}{c}{} \\
 &    0 &  0.1944 $\sigma$  & 0.324 $\sigma$ & 0.54 $\sigma$ & 0.9 $\sigma$ &  1.5 $\sigma$   \\
\cline{2-7} \\
Fit lines; Vertical line   & 0.142 & 0.215 & 0.172 & 0.051 & 0.057 & -0.002 \\
  & (0.038) & (0.046) & (0.044) & (0.045) & (0.029) & (0.025) \\
No fit lines; No vertical line   & 0.017 & 0.012 & -0.007 & -0.052 & -0.026 & -0.023 \\
  & (0.024) & (0.043) & (0.045) & (0.043) & (0.031) & (0.025) \\
No fit lines; Vertical line (Mean)   & 0.037 & 0.159 & 0.348 & 0.634 & 0.860 & 0.976 \\
Number of graphs  & 496 & 496 & 496 & 496 & 496 & 248 \\
Number of players  & 248 & 248 & 248 & 248 & 248 & 248 \\
\cline{1-7} \\
\end{tabular*} 
}

    \begin{minipage}[t]{0.8\columnwidth}%
        {\scriptsize{}Notes: We have a stratified randomized experiment, where each of the 11 strata is determined by the DGPs seen for every discontinuity magnitude. The treatment effect estimates come from regressing the participants' responses on treatment indicators and stratum fixed effects. We obtain standard errors using the procedure from \citet{Bugnietal2019} for stratified experiments. \textit{Fit lines} indicates the presence of parametric fit lines; \textit{vertical line} indicates the presence of a vertical line at the treatment threshold.}%
    \end{minipage}
\end{table}

\begin{table}[H]
    \caption{\label{tab:TE-RD-VI}Effects of Graphical Methods on Visual Inference: Phase 4}
    \centering
    \resizebox*{4.0in}{!}{%
\begin{tabular*}{1.0\hsize}{@{\hskip\tabcolsep\extracolsep\fill}l c c c c c c} 
\vspace{-.45cm} 
&\multicolumn{6}{c}{Dependent variable: player reports discontinuity} \\ 
\cline{2-7} \\
&\multicolumn{6}{c}{True discontinuity magnitude = } \\ 
\vspace{-.45cm} 
 &\multicolumn{6}{c}{} \\
 &    0 &  0.1944 $\sigma$  & 0.324 $\sigma$ & 0.54 $\sigma$ & 0.9 $\sigma$ &  1.5 $\sigma$   \\
\cline{2-7} \\
Large bins; No fit lines   & 0.145 & 0.283 & 0.223 & 0.239 & 0.074 & 0.027 \\
  & (0.038) & (0.044) & (0.047) & (0.042) & (0.033) & (0.030) \\
Small bins; Fit lines   & -0.019 & 0.088 & 0.047 & 0.059 & 0.024 & 0.002 \\
  & (0.030) & (0.042) & (0.049) & (0.046) & (0.036) & (0.034) \\
Large bins; Fit lines   & 0.230 & 0.300 & 0.239 & 0.207 & 0.058 & 0.003 \\
  & (0.045) & (0.046) & (0.047) & (0.044) & (0.034) & (0.034) \\
Small bins; No fit lines (Mean)   & 0.074 & 0.154 & 0.395 & 0.611 & 0.870 & 0.951 \\
Number of graphs  & 680 & 680 & 680 & 680 & 680 & 340 \\
Number of players  & 340 & 340 & 340 & 340 & 340 & 340 \\
\cline{1-7} \\
\end{tabular*} 
}

    \begin{minipage}[t]{0.8\columnwidth}%
        {\scriptsize{}Notes: We have a stratified randomized experiment, where each of the 11 strata is determined by the DGPs seen for every discontinuity magnitude. The treatment effect estimates come from regressing the participants' responses on treatment indicators and stratum fixed effects. We obtain standard errors using the procedure from \citet{Bugnietal2019} for stratified experiments. \textit{Large bins} corresponds to the \citet{Calonicoetal2015} bin width selector that minimizes the integrated mean squared error of the bin-average estimators of the conditional expectation function; \textit{small bins} corresponds to the \citet{Calonicoetal2015} bin width selector that aims to approximate the variability of the underlying data; \textit{fit line} indicates the presence of parametric fit lines.}%
    \end{minipage}
\end{table}

\begin{table}[H]
    \caption{\label{tab:TE-RD-VII}Effects of DGP Specification on Visual Inference: Supplemental Phase 5}
    \centering
    \resizebox*{4.0in}{!}{%
\begin{tabular*}{1.0\hsize}{@{\hskip\tabcolsep\extracolsep\fill}l c c c c c c} 
\vspace{-.45cm} 
&\multicolumn{6}{c}{Dependent variable: player reports discontinuity} \\ 
\cline{2-7} \\
&\multicolumn{6}{c}{True discontinuity magnitude = } \\ 
\vspace{-.45cm} 
 &\multicolumn{6}{c}{} \\
 &    0 &  0.1944 $\sigma$  & 0.324 $\sigma$ & 0.54 $\sigma$ & 0.9 $\sigma$ &  1.5 $\sigma$   \\
\cline{2-7} \\
Homoskedastic local linear   & -0.018 & -0.044 & 0.013 & -0.006 & 0.001 & 0.000 \\
  & (0.025) & (0.042) & (0.043) & (0.045) & (0.026) & (0.027) \\
Heteroskedastic global quintic   & -0.025 & -0.052 & -0.010 & -0.046 & 0.009 & 0.023 \\
  & (0.021) & (0.040) & (0.043) & (0.046) & (0.026) & (0.022) \\
Heteroskedastic local linear   & -0.011 & -0.063 & 0.044 & 0.006 & 0.008 & 0.022 \\
  & (0.023) & (0.039) & (0.041) & (0.047) & (0.028) & (0.023) \\
Homoskedastic global quintic (Mean)   & 0.047 & 0.218 & 0.359 & 0.676 & 0.912 & 0.965 \\
Number of graphs  & 678 & 678 & 678 & 678 & 678 & 339 \\
Number of players  & 339 & 339 & 339 & 339 & 339 & 339 \\
\cline{1-7} \\
\end{tabular*} 
}

    \begin{minipage}[t]{0.8\columnwidth}%
        {\scriptsize{}Notes:  In the supplemental Phase 5, we test the sensitivity of visual inference to alternative specifications of the data generating processes (using a local linear instead of global quintic specification and allowing for heteroskedastic errors). We have a stratified randomized experiment, where each of the 11 strata is determined by the DGPs seen for every discontinuity magnitude. The treatment effect estimates come from regressing the participants' responses on treatment indicators and stratum fixed effects. We obtain standard errors using the procedure from \citet{Bugnietal2019} for stratified experiments.}%
    \end{minipage}
\end{table}
 
\begin{table}[H]
    \caption{Risks: Classical and Andrews and Shapiro (AS)}
    \label{tab:risks-main}
    \centering{\footnotesize{%
    \begin{tabular}{llllll}
    \hline
    & {\footnotesize{} (1)} & {\footnotesize{}(2)} & {\footnotesize{}(3)} & {\footnotesize{}(4)} \\
    {\footnotesize{}Treatment} & {\footnotesize{}\makecell[l]{Type I Error Rate\\($d=0$)}} & {\footnotesize{}\makecell[l]{Type II Error Rate\\($\lvert d \rvert=0.324\sigma$)}} & {\footnotesize{}\makecell[l]{Classical Risk\\(Equal Weights)}} & {\footnotesize{}\makecell[l]{Classical Risk\\($4\times$ Weight at $d=0$)}} \\
    \hline 
    & \multicolumn{4}{c}{Phase 1}  \\\cline{2-5}
    \textbf{Small bins/default axes} & 0.055 & 0.634 & 0.688 & 0.853 \\
    Large bins/default axes & 0.257 & 0.461 & 0.717 [0.647] & 1.488 [0.000] \\
    Small bins/large axes & 0.036 & 0.692 & 0.729 [0.504] & 0.838 [0.835] \\
    Large bins/large axes & 0.198 & 0.520 & 0.718 [0.747] & 1.313 [0.002] \\
    & \multicolumn{4}{c}{Phase 2} \\\cline{2-5}
    \textbf{Small bins/even spacing} & 0.053 & 0.648 & 0.702 & 0.861 \\
    Large bins/even spacing & 0.306 & 0.439 & 0.745 [0.355] & 1.663 [0.000] \\
    Small bins/quantile spacing & 0.088 & 0.655 & 0.743 [0.414] & 1.008 [0.270] \\
    Large bins/quantile spacing & 0.211 & 0.421 & 0.632 [0.322] & 1.264 [0.006] \\
    & \multicolumn{4}{c}{Phase 3} \\\cline{2-5}
    \textbf{No fit lines/vertical line} & 0.036 & 0.659 & 0.694 & 0.802 \\
    Fit lines/vertical line & 0.179 & 0.485 & 0.664 [0.613] & 1.200 [0.010] \\
    No fit lines/no vertical line & 0.052 & 0.664 & 0.716 [0.658] & 0.873 [0.487] \\
    & \multicolumn{4}{c}{Phase 4} \\\cline{2-5}
    \textbf{Small bins/no fit lines} & 0.073 & 0.609 & 0.682 & 0.900 \\
    Large bins/no fit lines & 0.218 & 0.379 & 0.597 [0.174] & 1.250 [0.022] \\
    Small bins/fit lines & 0.054 & 0.555 & 0.609 [0.220] & 0.769 [0.285] \\
    Large bins/fit lines & 0.304 & 0.367 & 0.671 [0.908] & 1.582 [0.000] \\
    \hline
    {\footnotesize{}Treatment} &  {\footnotesize{}\makecell[l]{AS Risk\\($d=0$)}} & {\footnotesize{}\makecell[l]{AS Risk\\($\lvert d \rvert=0.324\sigma$)}} & {\footnotesize{}\makecell[l]{AS Risk\\(Equal Weights)}} & {\footnotesize{}\makecell[l]{AS Risk\\($4\times$ Weight at $d=0$)}} \\ 
    \hline
    & \multicolumn{4}{c}{Phase 1}\\\cline{2-5}
    \textbf{Small bins/default axes} & 0.180 & 0.208 & 0.388 & 0.927 \\
    Large bins/default axes & 0.215 & 0.221 & 0.436 [0.002] & 1.081 [0.000] \\
    Small bins/large axes & 0.182 & 0.204 & 0.386 [0.913] & 0.931 [0.924] \\
    Large bins/large axes & 0.201 & 0.213 & 0.415 [0.078] & 1.018 [0.024] \\
    & \multicolumn{4}{c}{Phase 2} \\\cline{2-5}
    \textbf{Small bins/even spacing} & 0.183 & 0.214 & 0.397 & 0.947 \\
    Large bins/even spacing & 0.226 & 0.215 & 0.441 [0.018] & 1.119 [0.001] \\
    Small bins/quantile spacing & 0.212 & 0.233 & 0.445 [0.005] & 1.081 [0.002] \\
    Large bins/quantile spacing & 0.229 & 0.218 & 0.447 [0.009] & 1.134 [0.000] \\
    & \multicolumn{4}{c}{Phase 3} \\\cline{2-5}
     \textbf{No fit lines/vertical line} & 0.192 & 0.210 & 0.402 & 0.980 \\
    Fit lines/vertical line & 0.183 & 0.218 & 0.401 [0.997] & 0.950 [0.510] \\
    No fit lines/no vertical line & 0.190 & 0.214 & 0.403 [0.891] & 0.972 [0.882] \\
    & \multicolumn{4}{c}{Phase 4} \\\cline{2-5}
    \textbf{Small bins/no fit lines} & 0.198 & 0.227 & 0.425 & 1.017 \\
    Large bins/no fit lines & 0.223 & 0.217 & 0.439 [0.397] & 1.107 [0.063] \\
    Small bins/fit lines & 0.192 & 0.216 & 0.408 [0.398] & 0.984 [0.506] \\
    Large bins/fit lines & 0.219 & 0.216 & 0.435 [0.586] & 1.090 [0.134] \\
    \hline 
    \end{tabular}{\footnotesize\par}
    }}
    \medskip
    \begin{singlespace}
    \noindent %
    \begin{minipage}[t]{6.5in}%
    \begin{singlespace}
        {\scriptsize{}Notes: For both the classical and the Andrews and Shapiro risk measures, column (3) is simply the sum of columns (1) and (2); column (4) is equal to four times column (1) plus column (2). Risks at $\lvert d \rvert=0.324\sigma$ are calculated using responses at that discontinuity magnitude, the mode and median of the subset of our 11 DGPs corresponding to papers that report discontinuities. In brackets in columns (3) and (4) are the $p$-values for testing whether the difference in risks relative to the first and benchmark treatment (in bold) within each phase is zero. We obtain the $p$-values by regressing risks on treatment indicators and stratum fixed effects, where we define the 11 strata by the DGPs seen for every discontinuity magnitude, and conducting inference using the procedure from \citet{Bugnietal2019} for stratified experiments. \textit{Large bins} corresponds to the \citet{Calonicoetal2015} bin width selector that minimizes the integrated mean squared error of the bin-average estimators of the conditional expectation function; \textit{small bins} corresponds to the \citet{Calonicoetal2015} bin width selector that aims to approximate the variability of the underlying data; \textit{default axes} corresponds to Stata 14's default axis scaling; \textit{large axes} doubles this default axis scaling; \textit{vertical line} indicates the presence of a vertical line at the treatment threshold. In each phase, we highlight the repeated benchmark small bins/no fit lines treatment in bold.}
    \end{singlespace}
    \end{minipage}
    \end{singlespace}
\end{table}
 
\begin{table}[H]
    \caption{Risks: Classical and Andrews and Shapiro (AS), Using All Non-Zero Discontinuities}
    \label{tab:risks-supp}
    \centering{\footnotesize{%
        \begin{tabular}{llllll}
            \hline
            & {\footnotesize{} (1)} & {\footnotesize{}(2)} & {\footnotesize{}(3)} & {\footnotesize{}(4)} \\
            {\footnotesize{}Treatment} & {\footnotesize{}\makecell[l]{Type I Error Rate\\($d=0)$}} & {\footnotesize{}\makecell[l]{Type II Error Rate\\($d\neq 0$)}} & {\footnotesize{}\makecell[l]{Classical Risk\\(Equal Weights)}} & {\footnotesize{}\makecell[l]{Classical Risk\\($4\times$ Weight at $d=0$)}} \\
            \hline 
            & \multicolumn{4}{c}{Phase 1}  \\\cline{2-5}
            \textbf{Small bins/default axes} & 0.055 & 0.374 & 0.428 & 0.593 \\
            Large bins/default axes & 0.257 & 0.277 & 0.534 [0.042] & 1.305 [0.000] \\
            Small bins/large axes & 0.036 & 0.403 & 0.440 [0.947] & 0.549 [0.537] \\
            Large bins/large axes & 0.198 & 0.294 & 0.492 [0.168] & 1.087 [0.001] \\
            & \multicolumn{4}{c}{Phase 2}  \\\cline{2-5}
            \textbf{Small bins/even spacing} & 0.053 & 0.394 & 0.447 & 0.606 \\
            Large bins/even spacing & 0.306 & 0.266 & 0.572 [0.005] & 1.490 [0.000] \\
            Small bins/quantile spacing & 0.088 & 0.407 & 0.495 [0.216] & 0.760 [0.263] \\
            Large bins/quantile spacing & 0.211 & 0.275 & 0.485 [0.497] & 1.117 [0.001] \\
            & \multicolumn{4}{c}{Phase 3}  \\\cline{2-5}
            \textbf{No fit lines/vertical line} & 0.036 & 0.405 & 0.440 & 0.548 \\
            Fit lines/vertical line & 0.179 & 0.308 & 0.487 [0.374] & 1.024 [0.002]\\
            No fit lines/no vertical line & 0.052 & 0.424 & 0.476 [0.277] & 0.633 [0.387] \\
            & \multicolumn{4}{c}{Phase 4}  \\\cline{2-5}
            \textbf{Small bins/no fit lines} & 0.073 & 0.406 & 0.479 & 0.698 \\
            Large bins/no fit lines & 0.218 & 0.233 & 0.451 [0.342] & 1.104 [0.009] \\
            Small bins/fit lines & 0.054 & 0.360 & 0.413 [0.072] & 0.574 [0.304] \\
            Large bins/fit lines & 0.304 & 0.244 & 0.548 [0.270] & 1.459 [0.000] \\
            \hline
            {\footnotesize{}Treatment} &  {\footnotesize{}\makecell[l]{AS Risk\\($d=0$)}} & {\footnotesize{}\makecell[l]{AS Risk\\($d\neq0$)}} & {\footnotesize{}\makecell[l]{AS Risk\\(Equal Weights)}} & {\footnotesize{}\makecell[l]{AS Risk\\($4\times$ Weight at $d=0$)}} \\ 
            \hline
            & \multicolumn{4}{c}{Phase 1}  \\\cline{2-5}
            \textbf{Small bins/default axes} & 0.180 & 0.201 & 0.381 & 0.920 \\
            Large bins/default axes & 0.215 & 0.206 & 0.421 [0.003] & 1.066 [0.000] \\
            Small bins/large axes & 0.182 & 0.205 & 0.387 [0.627] & 0.932 [0.765] \\
            Large bins/large axes & 0.201 & 0.200 & 0.402 [0.133] & 1.006 [0.044] \\
            & \multicolumn{4}{c}{Phase 2}  \\\cline{2-5}
            \textbf{Small bins/even spacing} & 0.183 & 0.199 & 0.382 & 0.932 \\
            Large bins/even spacing & 0.226 & 0.200 & 0.425 [0.007] & 1.103 [0.000] \\
            Small bins/quantile spacing & 0.212 & 0.218 & 0.430 [0.001] & 1.066 [0.003] \\
            Large bins/quantile spacing & 0.229 & 0.205 & 0.435 [0.000] & 1.122 [0.000] \\
            & \multicolumn{4}{c}{Phase 3}  \\\cline{2-5}
            \textbf{No fit lines/vertical line} & 0.192 & 0.202 & 0.394 & 0.972 \\
            Fit lines/vertical line & 0.183 & 0.198 & 0.381 [0.319] & 0.930 [0.333] \\
            No fit lines/no vertical line & 0.190 & 0.204 & 0.393 [0.949] & 0.962 [0.835] \\
            & \multicolumn{4}{c}{Phase 4}  \\\cline{2-5}
            \textbf{Small bins/no fit lines} & 0.198 & 0.208 & 0.405 & 0.998 \\
            Large bins/no fit lines & 0.223 & 0.200 & 0.423 [0.277] & 1.091 [0.059] \\
            Small bins/fit lines & 0.192 & 0.198 & 0.390 [0.318] & 0.967 [0.508] \\
            Large bins/fit lines & 0.219 & 0.207 & 0.425 [0.224] & 1.081 [0.083] \\
            \hline 
        \end{tabular}{\footnotesize\par}
    }}
    \medskip
    \begin{singlespace}
    \noindent %
    \begin{minipage}[t]{6.5in}%
        \begin{singlespace}
            {\scriptsize{}Notes: For both the classical and the Andrews and Shapiro risk measures, column (3) is simply the sum of columns (1) and (2); column (4) is equal to four times column (1) plus column (2). Values when $d \neq 0$ weight all discontinuity magnitudes equally. In brackets in columns (3) and (4) are the $p$-values for testing whether the difference in risks relative to the first and benchmark treatment (in bold) within each phase is zero. We obtain the $p$-values by regressing risks on treatment indicators and stratum fixed effects, where we define the 11 strata by the DGPs seen for every discontinuity magnitude, and conducting inference using the procedure from \citet{Bugnietal2019} for stratified experiments. \textit{Large bins} corresponds to the \citet{Calonicoetal2015} bin width selector that minimizes the integrated mean squared error of the bin-average estimators of the conditional expectation function; \textit{small bins} corresponds to the \citet{Calonicoetal2015} bin width selector that aims to approximate the variability of the underlying data; \textit{default axes} corresponds to Stata 14's default axis scaling; \textit{large axes} doubles this default axis scaling; \textit{vertical line} indicates the presence of a vertical line at the treatment threshold. In each we phase, we highlight the repeated benchmark small bins/no fit lines treatment in bold.}
        \end{singlespace}
    \end{minipage}
    \end{singlespace}
\end{table}

\begin{table}[H]
    \caption{\label{tab:TE-conf-RD-I}Effects of Graphical Methods on Subjective Probability of Correct Classification: Phase 1}
    \centering\resizebox*{4.0in}{!}{%
\begin{tabular*}{1.0\hsize}{@{\hskip\tabcolsep\extracolsep\fill}l c c c c c c} 
\vspace{-.45cm} 
&\multicolumn{6}{c}{Dependent variable: subjective probability correct} \\ 
\cline{2-7} \\
&\multicolumn{6}{c}{True discontinuity magnitude = } \\ 
\vspace{-.45cm} 
 &\multicolumn{6}{c}{} \\
 &    0 &  0.1944 $\sigma$  & 0.324 $\sigma$ & 0.54 $\sigma$ & 0.9 $\sigma$ &  1.5 $\sigma$   \\
\cline{2-7} \\
Large bins; Default y-axis   & -0.035 & -0.001 & -0.014 & 0.002 & -0.012 & 0.007 \\
  & (0.009) & (0.012) & (0.011) & (0.012) & (0.009) & (0.006) \\
Small bins; Large y-axis   & -0.001 & 0.006 & 0.003 & -0.021 & -0.006 & 0.000 \\
  & (0.007) & (0.010) & (0.011) & (0.011) & (0.008) & (0.007) \\
Large bins; Large y-axis   & -0.020 & 0.003 & -0.007 & 0.008 & -0.004 & 0.003 \\
  & (0.009) & (0.011) & (0.010) & (0.011) & (0.008) & (0.007) \\
Small bins; Default y-axis (Mean)   & 0.820 & 0.791 & 0.793 & 0.787 & 0.812 & 0.826 \\
Number of graphs  & 660 & 660 & 660 & 660 & 660 & 330 \\
Number of players  & 330 & 330 & 330 & 330 & 330 & 330 \\
\cline{1-7} \\
\end{tabular*} 
}

    \begin{minipage}[t]{0.8\columnwidth}%
        {\scriptsize{}Notes: We have a stratified randomized experiment, where each of the 11 strata is determined by the DGPs seen for every discontinuity magnitude. The treatment effect estimates come from regressing the participants' imputed subjective probabilities of a correct classification on treatment indicators and stratum fixed effects. We obtain standard errors using the procedure from \citet{Bugnietal2019} for stratified experiments. \textit{Large bins} corresponds to the \citet{Calonicoetal2015} bin width selector that minimizes the integrated mean squared error of the bin-average estimators of the conditional expectation function; \textit{small bins} corresponds to the \citet{Calonicoetal2015} bin width selector that aims to approximate the variability of the underlying data; \textit{default} $y$-axis scaling corresponds to Stata 14's default axis scaling; \textit{large} $y$-axis scaling doubles this default axis scaling.}%
    \end{minipage}
\end{table}

\begin{table}[H]
    \caption{\label{tab:TE-conf-RD-II}Effects of Graphical Methods on Subjective Probability of Correct Classification: Phase 2}
    \centering\resizebox*{4.0in}{!}{%
\begin{tabular*}{1.0\hsize}{@{\hskip\tabcolsep\extracolsep\fill}l c c c c c c} 
\vspace{-.45cm} 
&\multicolumn{6}{c}{Dependent variable: subjective probability correct} \\ 
\cline{2-7} \\
&\multicolumn{6}{c}{True discontinuity magnitude = } \\ 
\vspace{-.45cm} 
 &\multicolumn{6}{c}{} \\
 &    0 &  0.1944 $\sigma$  & 0.324 $\sigma$ & 0.54 $\sigma$ & 0.9 $\sigma$ &  1.5 $\sigma$   \\
\cline{2-7} \\
Large bins; Equal spacing   & -0.044 & -0.002 & -0.002 & 0.001 & 0.001 & 0.003 \\
  & (0.011) & (0.009) & (0.012) & (0.011) & (0.009) & (0.008) \\
Small bins; Quantile spacing   & -0.028 & -0.014 & -0.021 & -0.035 & -0.012 & -0.003 \\
  & (0.010) & (0.010) & (0.012) & (0.012) & (0.009) & (0.009) \\
Large bins; Qunatile spacing   & -0.046 & -0.035 & -0.004 & 0.007 & 0.001 & 0.009 \\
  & (0.011) & (0.011) & (0.011) & (0.010) & (0.008) & (0.007) \\
Small bins; Equal spacing (Mean)   & 0.817 & 0.802 & 0.787 & 0.794 & 0.809 & 0.823 \\
Number of graphs  & 650 & 650 & 650 & 650 & 650 & 325 \\
Number of players  & 325 & 325 & 325 & 325 & 325 & 325 \\
\cline{1-7} \\
\end{tabular*} 
}

    \begin{minipage}[t]{0.8\columnwidth}%
        {\scriptsize{}Notes: We have a stratified randomized experiment, where each of the 11 strata is determined by the DGPs seen for every discontinuity magnitude. The treatment effect estimates come from regressing the participants' imputed subjective probabilities of a correct classification on treatment indicators and stratum fixed effects. We obtain standard errors using the procedure from \citet{Bugnietal2019} for stratified experiments. \textit{Large bins} corresponds to the \citet{Calonicoetal2015} bin width selector that minimizes the integrated mean squared error of the bin-average estimators of the conditional expectation function; \textit{small bins} corresponds to the \citet{Calonicoetal2015} bin width selector that aims to approximate the variability of the underlying data.}%
    \end{minipage}
\end{table}

\begin{table}[H]
    \caption{\label{tab:TE-conf-RD-III}Effects of Graphical Methods on Subjective Probability of Correct Classification: Phase 3}
    \centering\resizebox*{4.0in}{!}{%
\begin{tabular*}{1.0\hsize}{@{\hskip\tabcolsep\extracolsep\fill}l c c c c c c} 
\vspace{-.45cm} 
&\multicolumn{6}{c}{Dependent variable: subjective probability correct} \\ 
\cline{2-7} \\
&\multicolumn{6}{c}{True discontinuity magnitude = } \\ 
\vspace{-.45cm} 
 &\multicolumn{6}{c}{} \\
 &    0 &  0.1944 $\sigma$  & 0.324 $\sigma$ & 0.54 $\sigma$ & 0.9 $\sigma$ &  1.5 $\sigma$   \\
\cline{2-7} \\
Fit lines; Vertical line   & 0.008 & 0.010 & -0.007 & 0.020 & -0.004 & -0.000 \\
  & (0.008) & (0.011) & (0.011) & (0.011) & (0.008) & (0.008) \\
No fit lines; No vertical line   & 0.002 & 0.002 & -0.003 & 0.007 & -0.009 & -0.006 \\
  & (0.009) & (0.011) & (0.012) & (0.011) & (0.009) & (0.008) \\
No fit lines; Vertical line (Mean)   & 0.808 & 0.794 & 0.791 & 0.780 & 0.814 & 0.826 \\
Number of graphs  & 496 & 496 & 496 & 496 & 496 & 248 \\
Number of players  & 248 & 248 & 248 & 248 & 248 & 248 \\
\cline{1-7} \\
\end{tabular*} 
}

    \begin{minipage}[t]{0.8\columnwidth}%
        {\scriptsize{}Notes: We have a stratified randomized experiment, where each of the 11 strata is determined by the DGPs seen for every discontinuity magnitude. The treatment effect estimates come from regressing the participants' imputed subjective probabilities of a correct classification on treatment indicators and stratum fixed effects. We obtain standard errors using the procedure from \citet{Bugnietal2019} for stratified experiments. \textit{Fit line} indicates the presence of parametric fit lines; \textit{vertical line} indicates the presence of a vertical line at the treatment threshold.}%
    \end{minipage}
\end{table}

\begin{table}[H]
    \caption{\label{tab:TE-conf-RD-VI}Effects of Graphical Methods on Subjective Probability of Correct Classification: Phase 4}
    \centering\resizebox*{4.0in}{!}{%
\begin{tabular*}{1.0\hsize}{@{\hskip\tabcolsep\extracolsep\fill}l c c c c c c} 
\vspace{-.45cm} 
&\multicolumn{6}{c}{Dependent variable: subjective probability correct} \\ 
\cline{2-7} \\
&\multicolumn{6}{c}{True discontinuity magnitude = } \\ 
\vspace{-.45cm} 
 &\multicolumn{6}{c}{} \\
 &    0 &  0.1944 $\sigma$  & 0.324 $\sigma$ & 0.54 $\sigma$ & 0.9 $\sigma$ &  1.5 $\sigma$   \\
\cline{2-7} \\
Large bins; No fit lines   & -0.025 & -0.006 & 0.010 & 0.014 & 0.016 & 0.005 \\
  & (0.011) & (0.010) & (0.012) & (0.011) & (0.008) & (0.009) \\
Small bins; Fit lines   & 0.005 & 0.004 & 0.011 & 0.006 & 0.022 & 0.004 \\
  & (0.010) & (0.011) & (0.013) & (0.012) & (0.008) & (0.009) \\
Large bins; Fit lines   & -0.021 & -0.017 & 0.011 & 0.009 & 0.002 & 0.001 \\
  & (0.011) & (0.011) & (0.012) & (0.012) & (0.010) & (0.010) \\
Small bins; No fit lines (Mean)   & 0.803 & 0.796 & 0.773 & 0.786 & 0.801 & 0.816 \\
Number of graphs  & 680 & 680 & 680 & 680 & 680 & 340 \\
Number of players  & 340 & 340 & 340 & 340 & 340 & 340 \\
\cline{1-7} \\
\end{tabular*} 
}

    \begin{minipage}[t]{0.8\columnwidth}%
        {\scriptsize{}Notes: We have a stratified randomized experiment, where each of the 11 strata is determined by the DGPs seen for every discontinuity magnitude. The treatment effect estimates come from regressing the participants' imputed subjective probabilities of a correct classification on treatment indicators and stratum fixed effects. We obtain standard errors using the procedure from \citet{Bugnietal2019} for stratified experiments. \textit{Large bins} corresponds to the \citet{Calonicoetal2015} bin width selector that minimizes the integrated mean squared error of the bin-average estimators of the conditional expectation function; \textit{small bins} corresponds to the \citet{Calonicoetal2015} bin width selector that aims to approximate the variability of the underlying data. \textit{Fit line} indicates the presence of parametric fit lines.}%
    \end{minipage}
\end{table}

\begin{table}[htbp]
    \caption{Timeline of Expert Study\label{tab:expert-details-appendix}}
    \centering\includegraphics[scale=0.9]{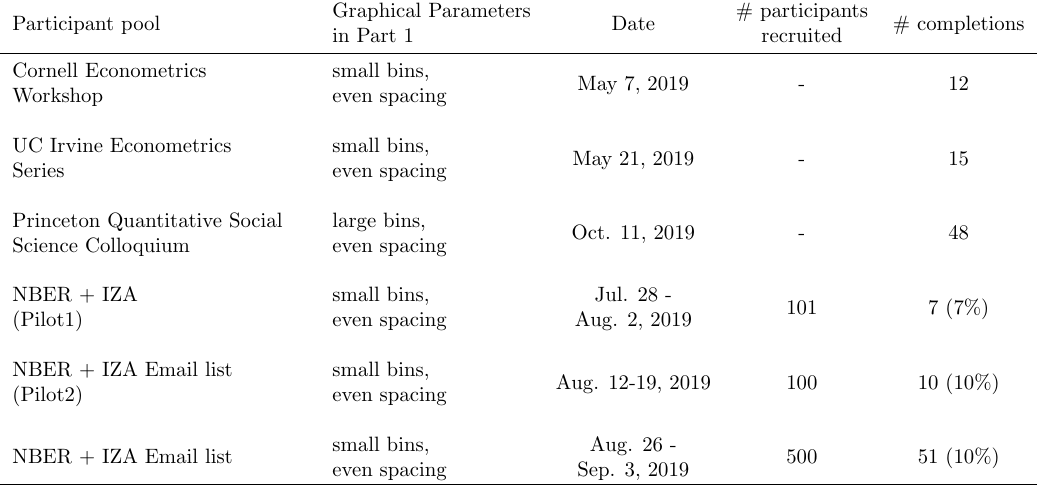}
    \begin{minipage}[t]{0.95\columnwidth}%
        {\scriptsize{}Notes: \textit{Large bins} corresponds to the \citet{Calonicoetal2015} bin width selector that minimizes the integrated mean squared error of the bin-average estimators of the conditional expectation function; \textit{small bins} corresponds to the \citet{Calonicoetal2015} bin width selector that aims to approximate the variability of the underlying data. }%
    \end{minipage}
\end{table}

\begin{table}[H]
    \caption{\label{tab:Cov-Bal-RD-I}Covariate Means and Balance across Treatment Arms: Phase 1}
    \centering\resizebox*{4.0in}{!}{%
\begin{tabular*}{1.0\hsize}{@{\hskip\tabcolsep\extracolsep\fill}l c c c c c } 
 &\multicolumn{5}{c}{} \\
 &    T1 &  T2  & T3 & T4 & p-value   \\
\cline{2-6} \\
Female   & 0.687 & 0.738 & 0.759 & 0.688 & 0.649  \\
Completed college   & 0.518 & 0.512 & 0.482 & 0.625 & 0.268  \\
Age 18 to 22   & 0.422 & 0.571 & 0.554 & 0.425 & 0.088  \\
Age 23 to 49   & 0.494 & 0.381 & 0.398 & 0.487 & 0.318  \\
Age 50 or older   & 0.060 & 0.048 & 0.048 & 0.087 & 0.742  \\
Graduate stats knowledge   & 0.145 & 0.119 & 0.133 & 0.050 & 0.091  \\
Passed attention check   & 0.867 & 0.857 & 0.843 & 0.912 & 0.516  \\
Attrition rate   & 0.057 & 0.045 & 0.057 & 0.091 & 0.690  \\
Participants & 83 & 84 & 83 & 80 & \\
\cline{1-6} \\
\end{tabular*} 
}

    \begin{singlespace}
    \noindent \centering{}%
    \begin{minipage}[t]{4in}%
        \noindent {\scriptsize{}Notes: The $p$-value for the joint significance of the covariates is 0.175. T1-T4 represent the four treatment arms considered in phase 1.}%
    \end{minipage}
    \end{singlespace}
\end{table}

\begin{table}[H]
    \caption{\label{tab:Cov-Bal-RD-II}Covariate Means and Balance across Treatment Arms: Phase 2}
    \centering\resizebox*{4.0in}{!}{%
\begin{tabular*}{1.0\hsize}{@{\hskip\tabcolsep\extracolsep\fill}l c c c c c } 
 &\multicolumn{5}{c}{} \\
 &    T1 &  T2  & T3 & T4 & p-value   \\
\cline{2-6} \\
Female   & 0.614 & 0.623 & 0.617 & 0.726 & 0.325  \\
Completed college   & 0.494 & 0.429 & 0.543 & 0.512 & 0.526  \\
Age 18 to 22   & 0.518 & 0.584 & 0.481 & 0.488 & 0.548  \\
Age 23 to 49   & 0.386 & 0.364 & 0.481 & 0.429 & 0.447  \\
Age 50 or older   & 0.084 & 0.039 & 0.037 & 0.071 & 0.485  \\
Graduate stats knowledge   & 0.072 & 0.143 & 0.099 & 0.119 & 0.499  \\
Passed attention check   & 0.928 & 0.870 & 0.877 & 0.940 & 0.306  \\
First time player   & 0.711 & 0.805 & 0.802 & 0.762 & 0.466  \\
Attrition rate   & 0.057 & 0.125 & 0.080 & 0.045 & 0.264  \\
Participants & 83 & 77 & 81 & 84 & \\
\cline{1-6} \\
\end{tabular*} 
}

    \begin{singlespace}
    \noindent \centering{}%
    \begin{minipage}[t]{4in}%
        \noindent {\scriptsize{}Notes: The $p$-value for the joint significance of the covariates is 0.496. T1-T4 represent the four treatment arms considered in phase 2.}%
    \end{minipage}
    \end{singlespace}
\end{table}

\begin{table}[H]
    \caption{\label{tab:Cov-Bal-RD-III}Covariate Means and Balance across Treatment Arms: Phase 3}
    \centering\resizebox*{4.0in}{!}{%
\begin{tabular*}{1.0\hsize}{@{\hskip\tabcolsep\extracolsep\fill}l c c c c } 
 &\multicolumn{4}{c}{} \\
 &    T1 &  T2  & T3 &  p-value   \\
\cline{2-5} \\
Female   & 0.683 & 0.725 & 0.616 &  0.323  \\
Completed college   & 0.659 & 0.650 & 0.558 &  0.339  \\
Age 18 to 22   & 0.439 & 0.463 & 0.523 &  0.529  \\
Age 23 to 49   & 0.488 & 0.512 & 0.419 &  0.449  \\
Age 50 or older   & 0.073 & 0.025 & 0.047 &  0.345  \\
Graduate stats knowledge   & 0.159 & 0.113 & 0.070 &  0.184  \\
Passed attention check   & 0.902 & 0.900 & 0.907 &  0.988  \\
First time player   & 0.488 & 0.500 & 0.453 &  0.823  \\
Attrition rate   & 0.068 & 0.091 & 0.023 & 0.087  \\
Participants & 82 & 80 & 86 & \\
\cline{1-5} \\
\end{tabular*} 
}

    \begin{singlespace}
    \noindent \centering{}%
    \begin{minipage}[t]{4in}%
        \noindent {\scriptsize{}Notes: The $p$-value for the joint significance of the covariates is 0.653. T1-T4 represent the four treatment arms considered in phase 3.}%
    \end{minipage}
    \end{singlespace}
\end{table}

\begin{table}[H]
    \caption{\label{tab:Cov-Bal-RD-VI}Covariate Balance across Treatment Arms: Phase 4}
    \centering\resizebox*{4.0in}{!}{%
\begin{tabular*}{1.0\hsize}{@{\hskip\tabcolsep\extracolsep\fill}l c c c c c } 
 &\multicolumn{5}{c}{} \\
 &    T1 &  T2  & T3 & T4 & p-value   \\
\cline{2-6} \\
Female   & 0.654 & 0.782 & 0.733 & 0.744 & 0.327  \\
Completed college   & 0.543 & 0.655 & 0.442 & 0.535 & 0.040  \\
Age 18 to 22   & 0.568 & 0.471 & 0.663 & 0.570 & 0.084  \\
Age 23 to 49   & 0.420 & 0.506 & 0.291 & 0.407 & 0.031  \\
Age 50 or older   & 0.012 & 0.023 & 0.035 & 0.023 & 0.797  \\
Graduate stats knowledge   & 0.160 & 0.115 & 0.070 & 0.047 & 0.067  \\
Passed attention check   & 0.914 & 0.897 & 0.953 & 0.907 & 0.427  \\
First time player   & 0.556 & 0.563 & 0.628 & 0.581 & 0.771  \\
Attrition rate   & 0.080 & 0.011 & 0.023 & 0.023 & 0.185  \\
Participants & 81 & 87 & 86 & 86 & \\
\cline{1-6} \\
\end{tabular*} 
}

    \begin{singlespace}
    \noindent \centering{}%
    \begin{minipage}[t]{4in}%
        \noindent {\scriptsize{}Notes: The $p$-value for the joint significance of the covariates is 0.137. T1-T4 represent the four treatment arms considered in phase 4.}%
    \end{minipage}
    \end{singlespace}
\end{table}

\begin{table}[H]
    \caption{\label{tab:Cov-Bal-RD-VII}Covariate Means and Balance across Treatment Arms: Supplemental Phase 5}
    \centering\resizebox*{4.0in}{!}{%
\begin{tabular*}{1.0\hsize}{@{\hskip\tabcolsep\extracolsep\fill}l c c c c c } 
 &\multicolumn{5}{c}{} \\
 &    T1 &  T2  & T3 & T4 & p-value   \\
\cline{2-6} \\
Female   & 0.659 & 0.706 & 0.694 & 0.774 & 0.388  \\
Completed college   & 0.576 & 0.600 & 0.635 & 0.560 & 0.766  \\
Age 18 to 22   & 0.518 & 0.518 & 0.424 & 0.440 & 0.466  \\
Age 23 to 49   & 0.435 & 0.447 & 0.518 & 0.500 & 0.651  \\
Age 50 or older   & 0.047 & 0.035 & 0.059 & 0.024 & 0.670  \\
Graduate stats knowledge   & 0.176 & 0.094 & 0.082 & 0.155 & 0.186  \\
Passed attention check   & 0.941 & 0.906 & 0.929 & 0.964 & 0.440  \\
First time player   & 0.671 & 0.776 & 0.753 & 0.690 & 0.357  \\
Attrition rate   & 0.034 & 0.034 & 0.034 & 0.045 & 0.976  \\
Participants & 85 & 85 & 85 & 84 & \\
\cline{1-6} \\
\end{tabular*} 
}

    \begin{singlespace}
    \noindent \centering{}%
    \begin{minipage}[t]{4in}%
        \noindent {\scriptsize{}Notes: The $p$-value for the joint significance of the covariates is 0.202. T1-T4 represent the four treatment arms considered in phase 5.}%
    \end{minipage}
    \end{singlespace}
\end{table}

\begin{table}[H]
    \caption{\label{tab:Cov-Prediction-RD-I-II-III-VI}Predictions of Correct Classification with Participant Characteristics}
    \centering\resizebox*{4.0in}{!}{%
\begin{tabular*}{1.0\hsize}{@{\hskip\tabcolsep\extracolsep\fill}l c c c c } 
\toprule 
 &\multicolumn{4}{c}{} \\
 &    Phase 1 &  Phase 2  & Phase 3 & Phase 4  \\
\cline{2-5} \\
\midrule
Female   & -0.001 & -0.007 & -0.019 & 0.017 \\
   & (0.014) & (0.014) & (0.018) & (0.017)  \\
Completed college   & 0.004 & -0.018 & -0.025 & 0.015 \\
   & (0.019) & (0.018) & (0.019) & (0.023)  \\
Age 23 to 49   & -0.006 & 0.003 & -0.012 & 0.009 \\
   & (0.021) & (0.019) & (0.021) & (0.023)  \\
Age 50 or older   & 0.004 & 0.037 & 0.094 & 0.055 \\
   & (0.026) & (0.031) & (0.033) & (0.073)  \\
Graduate stats knowledge   & -0.013 & 0.030 & 0.053 & -0.048 \\
   & (0.022) & (0.022) & (0.028) & (0.029)  \\
Passed attention check   & 0.055 & 0.043 & -0.014 & -0.011 \\
   & (0.023) & (0.026) & (0.029) & (0.031)  \\
First time player   & NA & 0.000 & -0.001 & 0.018 \\
   & (NA) & (0.016) & (0.017) & (0.015)  \\
\cline{1-5} \\
Joint test p-value & 0.321 & 0.460 & 0.009 & 0.389 \\
\bottomrule \\
\end{tabular*} 
}

    \begin{minipage}[t]{4in}%
        \noindent {\scriptsize{}Notes: This table presents the coefficients from regressing whether study participants are correct in classifying an RD graph on their characteristics. Each of the four columns represents results from one of the main four phases. Standard errors are clustered at the participant level.}%
    \end{minipage}{\scriptsize\par}
\end{table}

\begin{table}[H]
    \caption{\label{tab:DGP-Direction-Prediction-Power}Predictions of Type II Error Rates with DGP Characteristics}
    \centering\resizebox*{4.0in}{!}{%
\begin{tabular*}{1.0\hsize}{@{\hskip\tabcolsep\extracolsep\fill}l c c c c} 
 &\multicolumn{4}{c}{} \\
 &    Phase 1 &  Phase 2  & Phase 3 & Phase 4  \\
\cline{2-5} \\
LHS deriv-, RHS deriv-, discont-   & 0.244 & 0.262 & 0.267 & 0.308 \\
   & (0.032) & (0.033) & (0.039) & (0.031) \\
LHS deriv+, RHS deriv-, discont-   & -0.005 & 0.029 & -0.020 & 0.096 \\
   & (0.039) & (0.042) & (0.047) & (0.041) \\
LHS deriv-, RHS deriv+, discont-   & 0.024 & 0.103 & 0.173 & 0.102 \\
   & (0.058) & (0.057) & (0.068) & (0.054) \\
LHS deriv+, RHS deriv+, discont-   & -0.111 & -0.169 & -0.103 & -0.110 \\
   & (0.022) & (0.022) & (0.026) & (0.022) \\
\cline{1-5} \\
\end{tabular*} 
}

    \begin{minipage}[t]{4in}%
        \noindent {\scriptsize{}Notes: This table presents the coefficients from regressions of whether the participant is correct in classifying discontinuous graphs on the interactions of the signs of the left- and right-hand side first derivatives of the CEF at the policy threshold with the sign of discontinuity. \textit{LHS deriv}$+/-$ and \textit{RHS deriv}$+/-$ denote the sign of the left- and right-hand-side derivatives of the CEF, respectively. \textit{discont}$+/-$ denotes the sign of the discontinuity. Each of the four columns represents results from one of the main four phases. Standard errors are clustered at the participant level.}%
    \end{minipage}{\scriptsize\par}
\end{table}

\begin{table}[H]
    \caption{\label{tab:DGP-Direction-Prediction-Size}Predictions of Type I Error Rate with DGP Characteristics}
    \centering\resizebox*{4.0in}{!}{%
\begin{tabular*}{1.0\hsize}{@{\hskip\tabcolsep\extracolsep\fill}l c c c c} 
 &\multicolumn{4}{c}{} \\
 &    Phase 1 &  Phase 2  & Phase 3 & Phase 4   \\
\cline{2-5} \\
LHS deriv-, RHS deriv-   & 0.052 & 0.049 & -0.003 & 0.093 \\
   & (0.035) & (0.037) & (0.029)  & (0.036) \\
LHS deriv+, RHS deriv-   & -0.040 & -0.018 & 0.072 & -0.048 \\
   & (0.032) & (0.035) & (0.037)  & (0.032) \\
LHS deriv-, RHS deriv+   & -0.124 & -0.168 & -0.032 & -0.074 \\
   & (0.026) & (0.022) & (0.033)  & (0.039) \\
LHS deriv+, RHS deriv+ (Mean) & 0.141 & 0.168 & 0.077  & 0.154 \\
   & (0.020) & (0.022) & (0.018) & (0.020) \\
\cline{1-5} \\
\end{tabular*} 
}

    \begin{minipage}[t]{4in}%
        \noindent {\scriptsize{}Notes: This table presents the coefficients from regressions of whether the participant is correct in classifying continuous graphs on the interactions of the signs of left- and right-hand-side first derivatives of the CEF at the policy threshold. \textit{LHS deriv}$+/-$ and \textit{RHS deriv}$+/-$ denote the sign of the left- and right-hand-side derivatives of the CEF at the policy threshold, respectively. \textit{discont}$+/-$ denotes the sign of the discontinuity. Each of the four columns represents results from one of the main four phases. Standard errors are clustered at the participant level.}%
    \end{minipage}{\scriptsize\par}
\end{table}

\begin{table}[H]
    \caption{Dynamic Discontinuity Classification: AR(1) Regression\label{tab:dynamic-discontinuity-table}}
    \centering{\footnotesize{}}%
    \begin{tabular}{ccc}
    \hline
     & {\footnotesize{}(1)} & {\footnotesize{}(2)}\tabularnewline
    {\footnotesize{}Phase} & {\footnotesize{}Estimate of $\rho$ from Eq. (\ref{eq:AR1})} & {\footnotesize{}Bias-Corrected Estimate of $\rho$}\tabularnewline
    \hline 
    {\footnotesize{}1} & {\footnotesize{}-0.127} & {\footnotesize{}-0.041}\tabularnewline
     & {\footnotesize{}(0.038)} & {\footnotesize{}(0.046)}\tabularnewline
    {\footnotesize{}2} & {\footnotesize{}-0.170} & {\footnotesize{}-0.088}\tabularnewline
     & {\footnotesize{}(0.035)} & {\footnotesize{}(0.039)}\tabularnewline
    {\footnotesize{}3} & {\footnotesize{}-0.099} & {\footnotesize{}-0.009}\tabularnewline
     & {\footnotesize{}(0.029)} & {\footnotesize{}(0.036)}\tabularnewline
    {\footnotesize{}4} & {\footnotesize{}-0.126} & {\footnotesize{}-0.034}\tabularnewline
     & {\footnotesize{}(0.035)} & {\footnotesize{}(0.041)}\tabularnewline
    {\footnotesize{}Expert} & {\footnotesize{}-0.205} & {\footnotesize{}-0.126}\tabularnewline
     & {\footnotesize{}(0.032)} & {\footnotesize{}(0.037)}\tabularnewline
    \hline 
    \end{tabular}{\footnotesize\par}
    \medskip{}
    \begin{singlespace}
    \noindent \centering{}%
    \begin{minipage}[t]{3.5in}%
        \begin{singlespace}
        {\scriptsize{}Notes: This table presents results for the AR(1) regression of discontinuity classifications outlined in Equation (\ref{eq:AR1}). Standard errors are in parentheses. Standard errors are clustered at the participant level, and those in column (2) additionally correct the $O(1/S)$ bias by inverting equation (18) of \citet{Nickell1981}.}
        \end{singlespace}
    \end{minipage}
    \end{singlespace}
\end{table}


\begin{landscape}
    \begin{figure}[H]
        \caption{\label{fig:RDD-Power-Phase-By-DGP}Power Functions by DGP and Phase}
        \centering
        \includegraphics[width=3.5in]{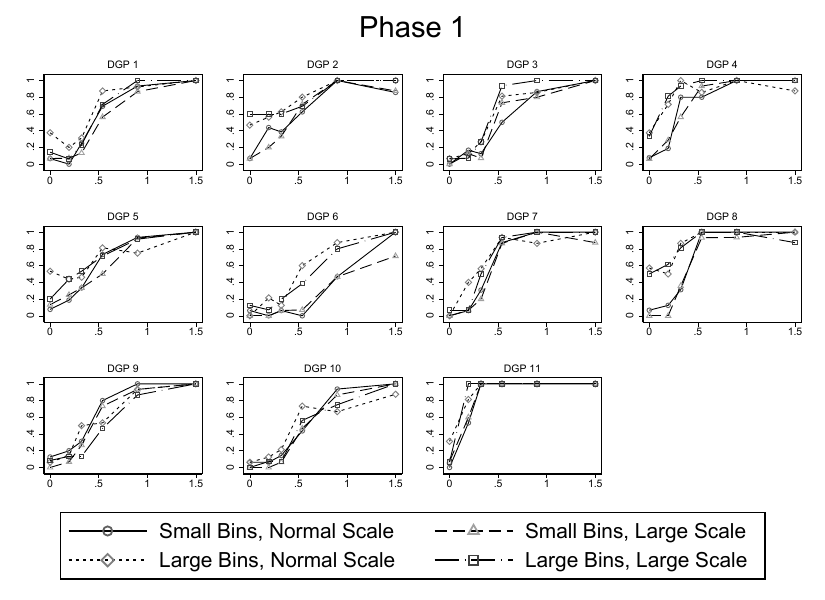}
        \includegraphics[width=3.5in]{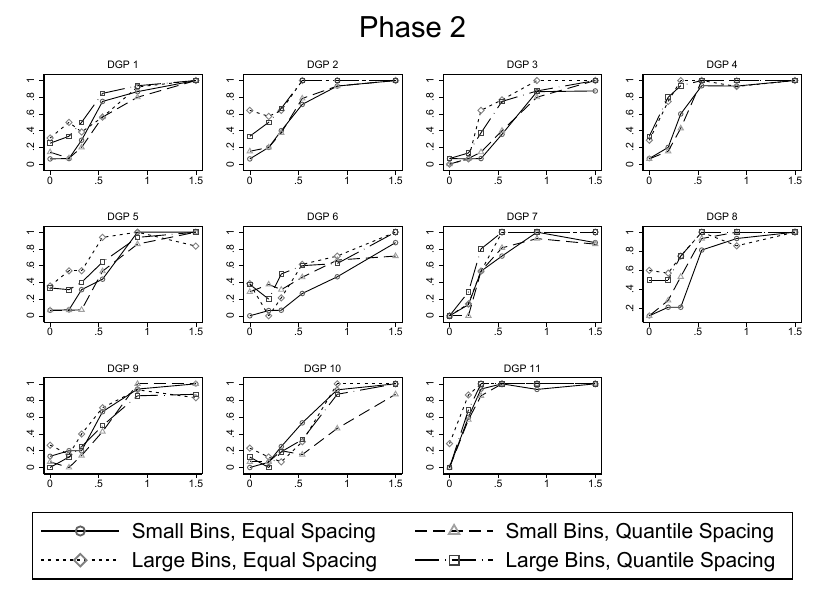}
        \includegraphics[width=3.5in]{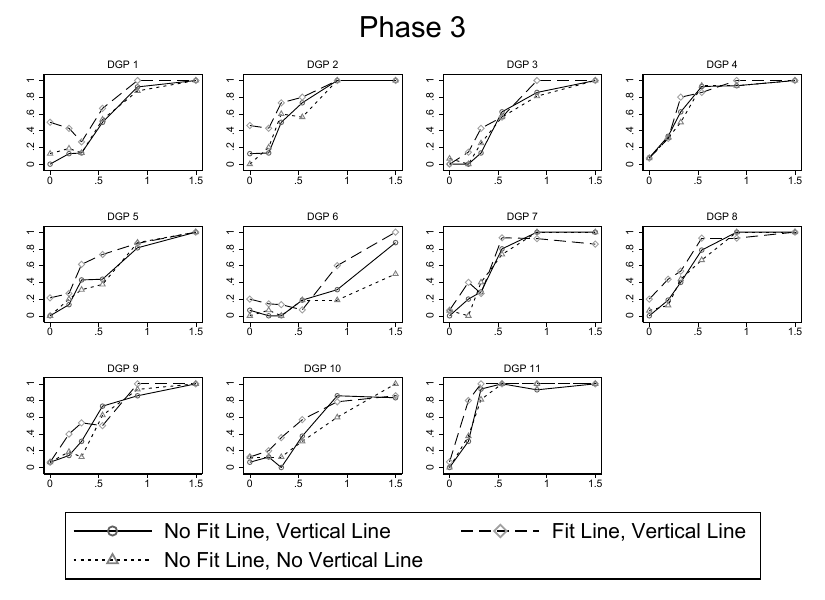}
        \includegraphics[width=3.5in]{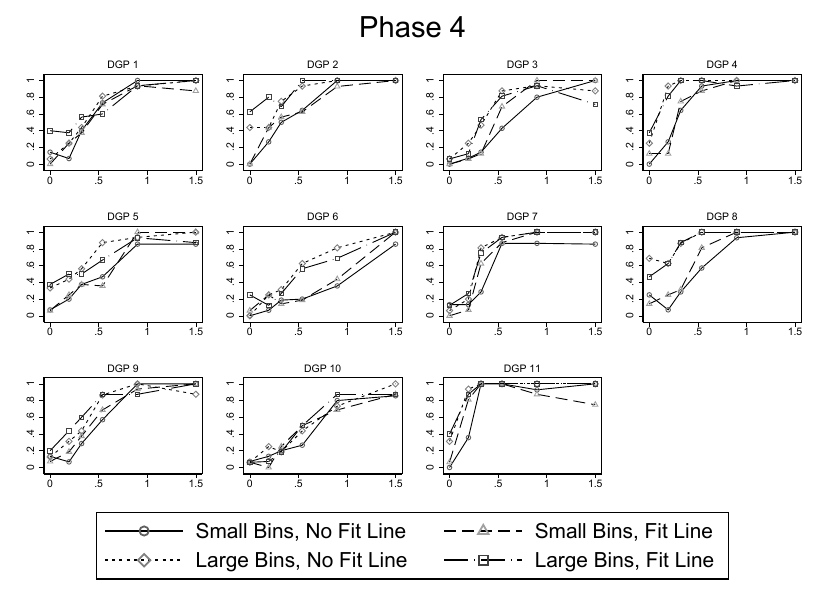}
        \begin{minipage}[t]{\columnwidth}%
            \noindent {\scriptsize{}Notes: These figures break up the power functions (as defined in Section \ref{sec:conceptual-framework}) in Figure \ref{fig:RDD-Power} by DGP. The discontinuity magnitude on the $x$-axis is specified as a multiple of the error standard deviation. The $y$-axis represents the share of respondents classifying a graph as having a discontinuity at the policy threshold. \textit{Large bins} corresponds to the \citet{Calonicoetal2015} bin width selector that minimizes the integrated mean squared error of the bin-average estimators of the conditional expectation function; \textit{small bins} corresponds to the \citet{Calonicoetal2015} bin width selector that aims to approximate the variability of the underlying data; \textit{quantile spacing} indicates that bins were spaced by quantiles rather than evenly spaced; \textit{fit line} indicates the presence of parametric fit lines; \textit{vertical line} indicates the presence of a vertical line at the policy threshold; \textit{normal scale} corresponds to the Stata 14 default $y$-axis scaling; \textit{large scale} doubles that default range.}%
        \end{minipage}{\scriptsize\par}
    \end{figure}

    \begin{figure}[H]
        \caption{\label{fig:Subj-Prob-CorrectBy-DGP}Average Subjective Probabilities of Correct Classification by Phase}
        \centering
        \includegraphics[width=3.5in]{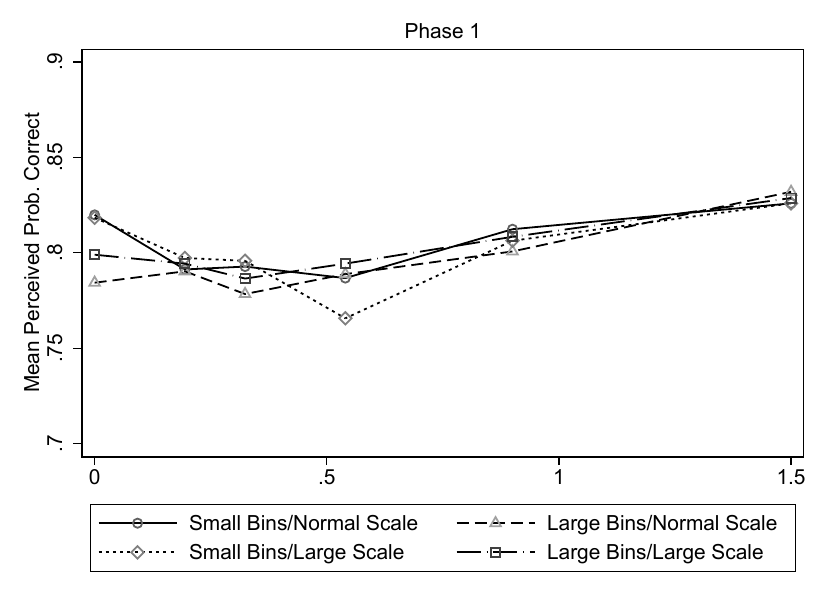}
        \includegraphics[width=3.5in]{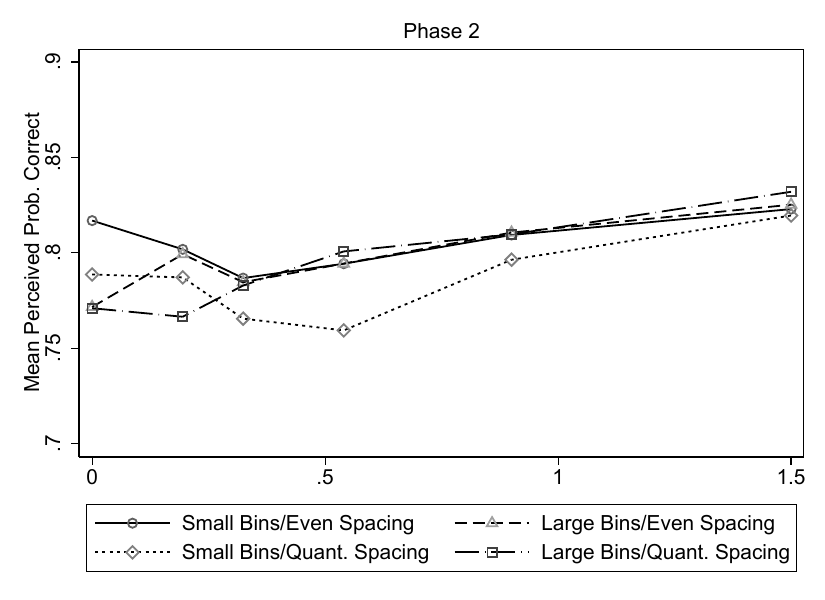}
        \includegraphics[width=3.5in]{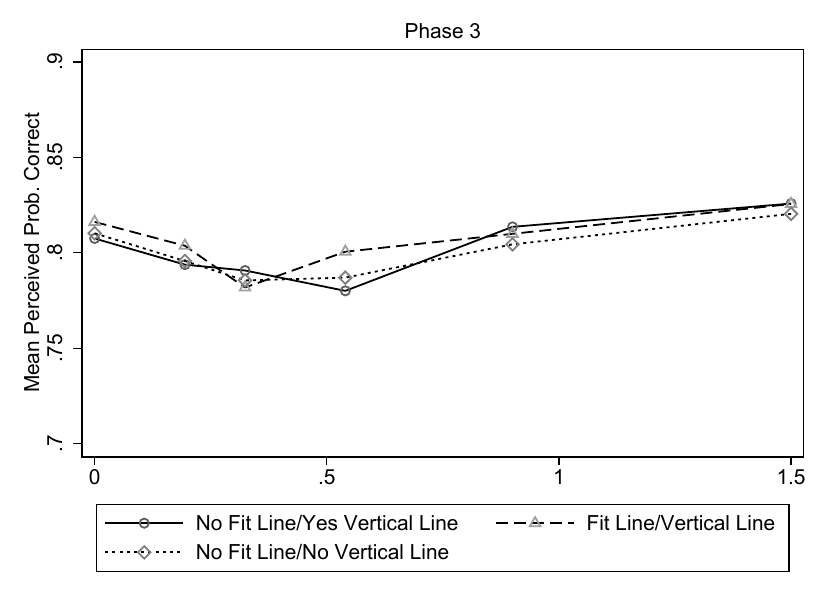}
        \includegraphics[width=3.5in]{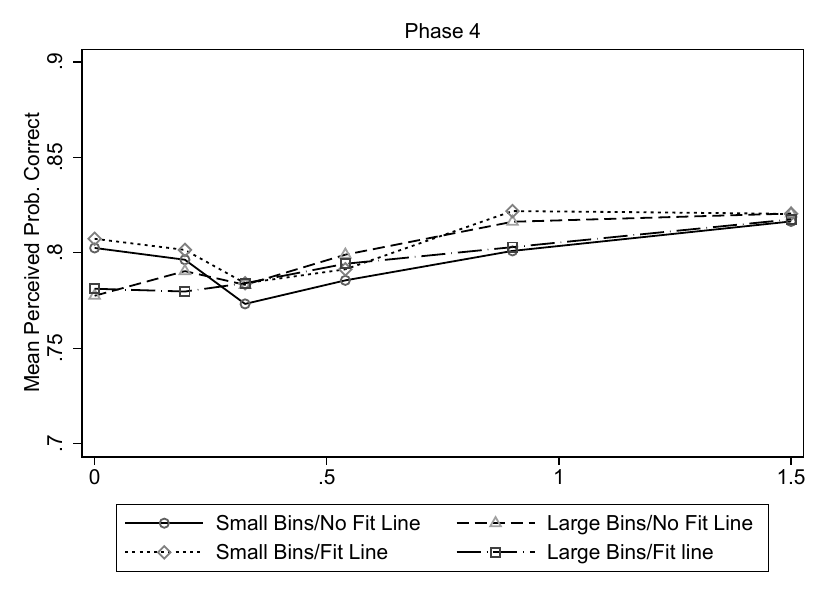}
        \begin{minipage}[t]{\textwidth}%
            \noindent {\scriptsize{Notes: This figure plots the average estimated subjective probabilities of a correct classification among non-expert participants (see Online Appendix \ref{sec:Risk} for details). \textit{Large bins} corresponds to the \citet{Calonicoetal2015} bin width selector that minimizes the integrated mean squared error of the bin-average estimators of the conditional expectation function; \textit{small bins} corresponds to the \citet{Calonicoetal2015} bin width selector that aims to approximate the variability of the underlying data; \textit{quantile spacing} indicates that bins were spaced by quantiles rather than evenly spaced; \textit{fit line} indicates the presence of parametric fit lines; \textit{vertical line} indicates the presence of a vertical line at the policy threshold; \textit{normal scale} corresponds to the default $y$-axis scaling using Stata 14; \textit{large scale} doubles that default range.}}%
        \end{minipage}{\scriptsize\par}
    \end{figure}
\end{landscape}

\begin{figure}[H]
    \caption{\label{fig:RDD-Power-Phase-VI-DGP9}Phase 4: Power Functions for DGP Shown to Experts (DGP 9)}
    \centering
    \includegraphics[width=5.5in]{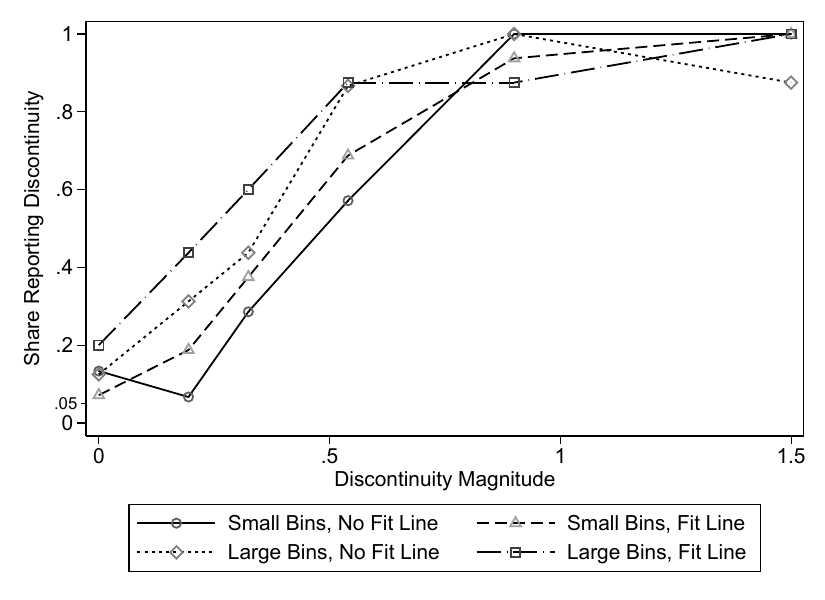}
    \begin{minipage}[t]{0.75\columnwidth}%
        {\scriptsize{}Notes: This figure shows the power functions for DGP 9, the DGP that was used to elicit expert preferences and beliefs across four considered treatments and three discontinuity levels in the second part of the expert study. The corresponding expert preferences and beliefs are shown in Figure \ref{fig:RDD-Experts-Pred-Pref}. \textit{Large bins} corresponds to the \citet{Calonicoetal2015} bin width selector that minimizes the integrated mean squared error of the bin-average estimators of the conditional expectation function; \textit{small bins} corresponds to the \citet{Calonicoetal2015} bin width selector that aims to approximate the variability of the underlying data; \textit{fit line} indicates the presence of parametric fit lines.}%
    \end{minipage}
\end{figure}

\begin{figure}[H]
    \caption{\label{fig:RDD-Experts-VS-CCT-CER}Expert Visual vs CCT Procedure with Inference-Optimal Bandwidth Power Functions}
    \includegraphics[width=3.25in]{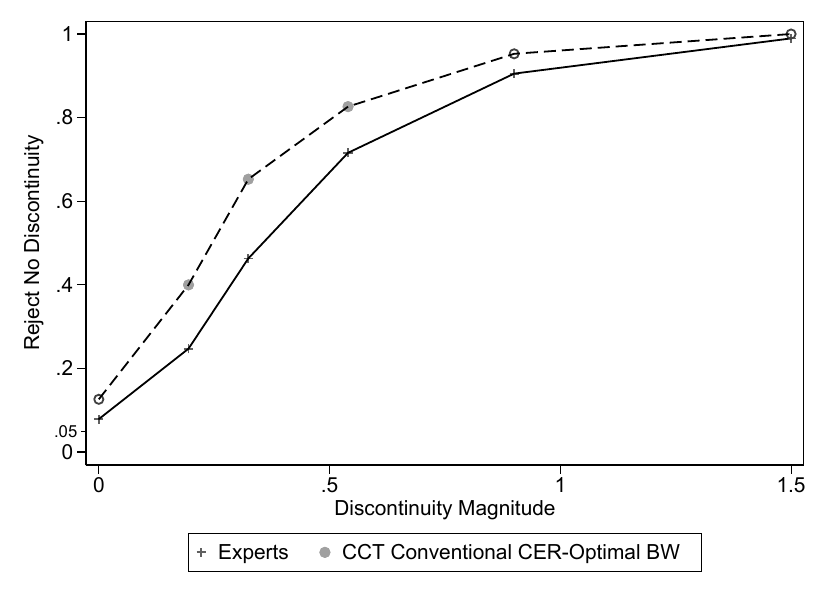}
    \includegraphics[width=3.25in]{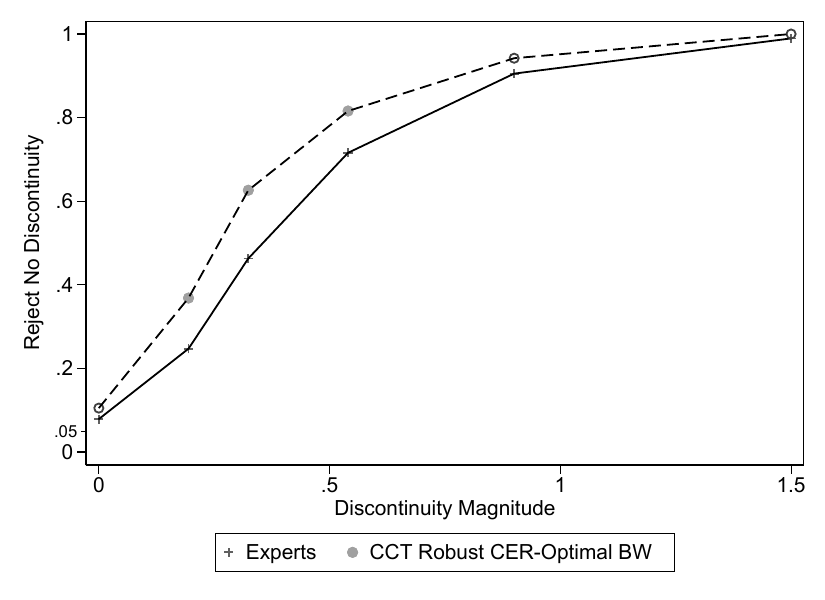}
    \begin{minipage}[t]{\textwidth}%
        \noindent {\scriptsize{}Notes: This figure plots the power functions of experts' visual inference and the CCT procedures with inference-optimal bandwidth (BW) per \citet{Calonicoetal2020optimal}. Markers in the figure are shown as solid, matching the legend, when the econometric inference procedure performs statistically significantly differently at the 5\% level from expert visual inference at the same discontinuity magnitude. Markers in the figure are shown as hollow instead whenever the econometric inference procedure does not perform statistically significantly differently at the 5\% level from expert visual inference at the same discontinuity magnitude.}%
    \end{minipage}{\scriptsize\par}
\end{figure}

\begin{figure}[H]
    \caption{\label{fig:RDD-Power-Experts-vs-Theoretical}Expert Visual vs Econometric Inference}
    \centering
    \includegraphics[width=5.5in]{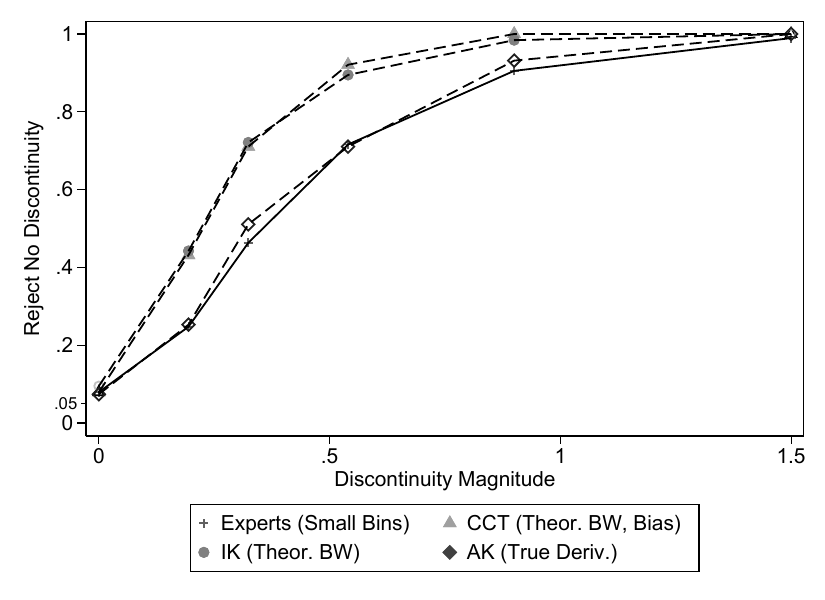}
    \centering{}%
    \begin{minipage}[t]{0.75\columnwidth}%
        {\scriptsize{}Notes: For this figure, we amend the econometric procedures by incorporating knowledge of the DGPs. We use the theoretical MSE-optimal bandwidths for IK and CCT, the theoretical asymptotic bias correction for CCT, and the true second derivative bounds for AK. Markers in the figure are shown as solid, matching the legend, when the econometric inference procedure performs statistically significantly differently at the 5\% level from expert visual inference at the same discontinuity magnitude. Markers in the figure are shown as hollow instead whenever the econometric inference procedure does not perform statistically significantly differently at the 5\% level from expert visual inference at the same discontinuity magnitude.}%
    \end{minipage}
\end{figure}

\newpage
\begin{landscape}
    \begin{figure}[H]
        \caption{\label{fig:RDD-Expert-Estimator-RMSE}\label{fig:RDD-Expert-Estimator-RMSE-Decomp}RMSE of Point Estimates and MSE Decomposition by DGP}
        \centering\includegraphics[width=4.25in]{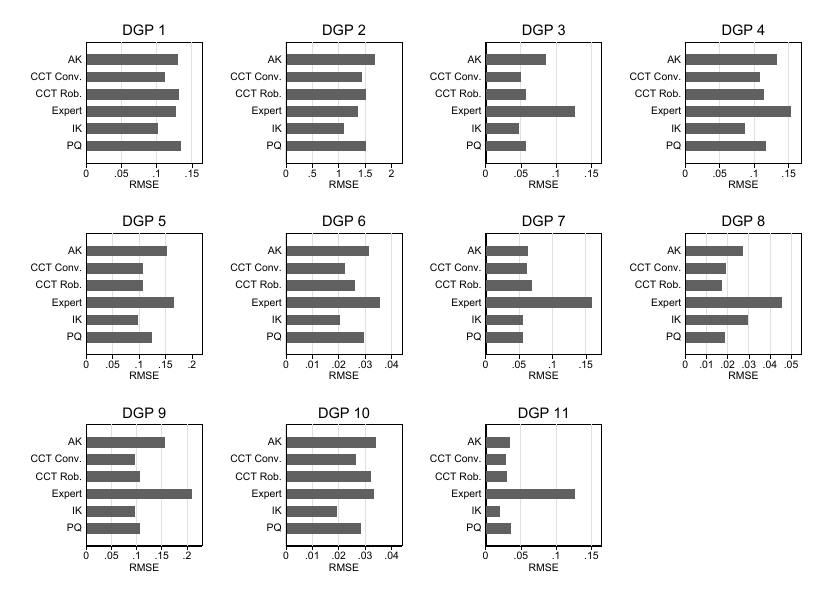}\includegraphics[width=4.25in]{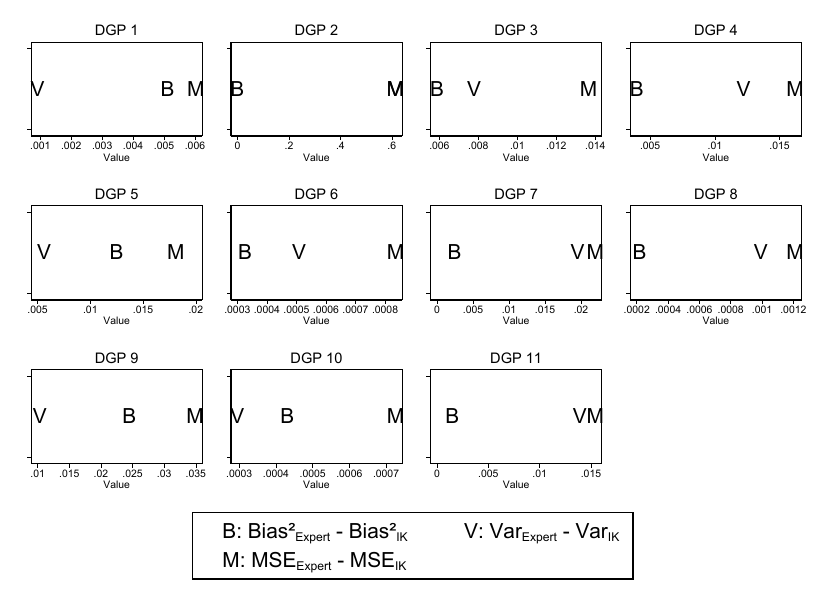}
        \begin{minipage}[t]{1.25\textwidth}%
            \noindent {\scriptsize{}Notes: The left panel compares for each DGP the root mean squared error (RMSE) in estimating the discontinuity magnitude in visual and econometric procedures. The right panel decomposes the difference in MSE between visual estimation and the IK procedure into bias and variance components. AK uses the }\texttt{\scriptsize{}RDHonest}{\scriptsize{} procedure with the rule-of-thumb bound on each DGP's second derivative \citep{ArmstrongandKolesar2017}. CCT Conventional is the conventional (not bias-corrected or robust) point estimates and standard errors with the MSE-optimal bandwidth as implemented in CCT's \texttt{rdrobust} \citep{Calonicoetal2014}. CCT Robust is the default RDD inference procedure from CCT's \texttt{rdrobust}. IK inference is based on a local linear estimator using the IK bandwidth \citep{ImbensKalyanaraman2012}. PQ uses a correctly specified regression model with global piecewise quintics above and below the treatment threshold and assuming homoskedasticity.}%
        \end{minipage}{\scriptsize\par}
    \end{figure}
\end{landscape}

\newpage{}

\begin{figure}[H]
    \caption{\label{fig:Lineup-Protocols-Fudged-Unfudged}Lineup Protocols for DGPs Specified Before and After Adding Noise to the Running Variable}
    \centering
    \includegraphics[width=3.25in]{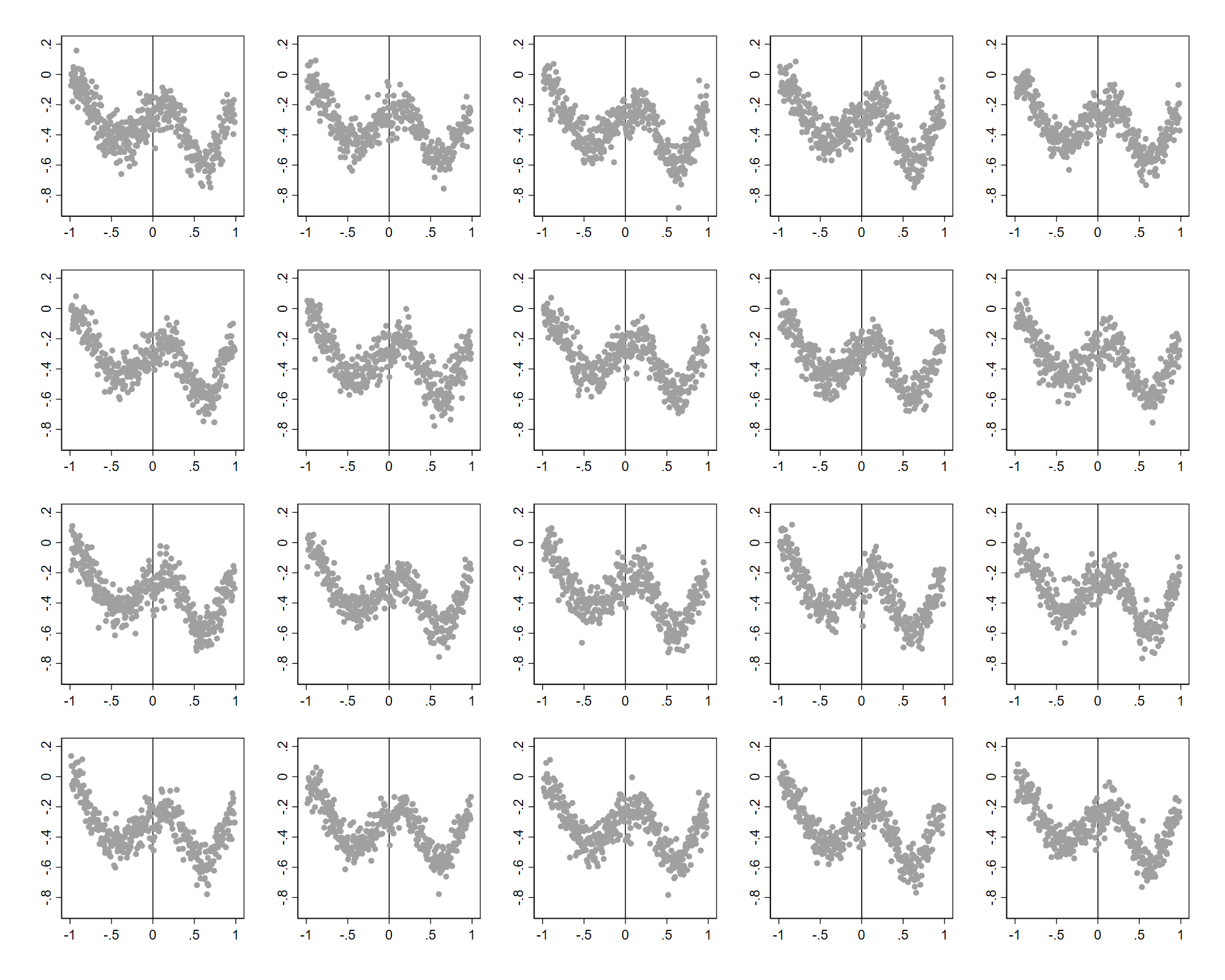}\includegraphics[width=3.25in]{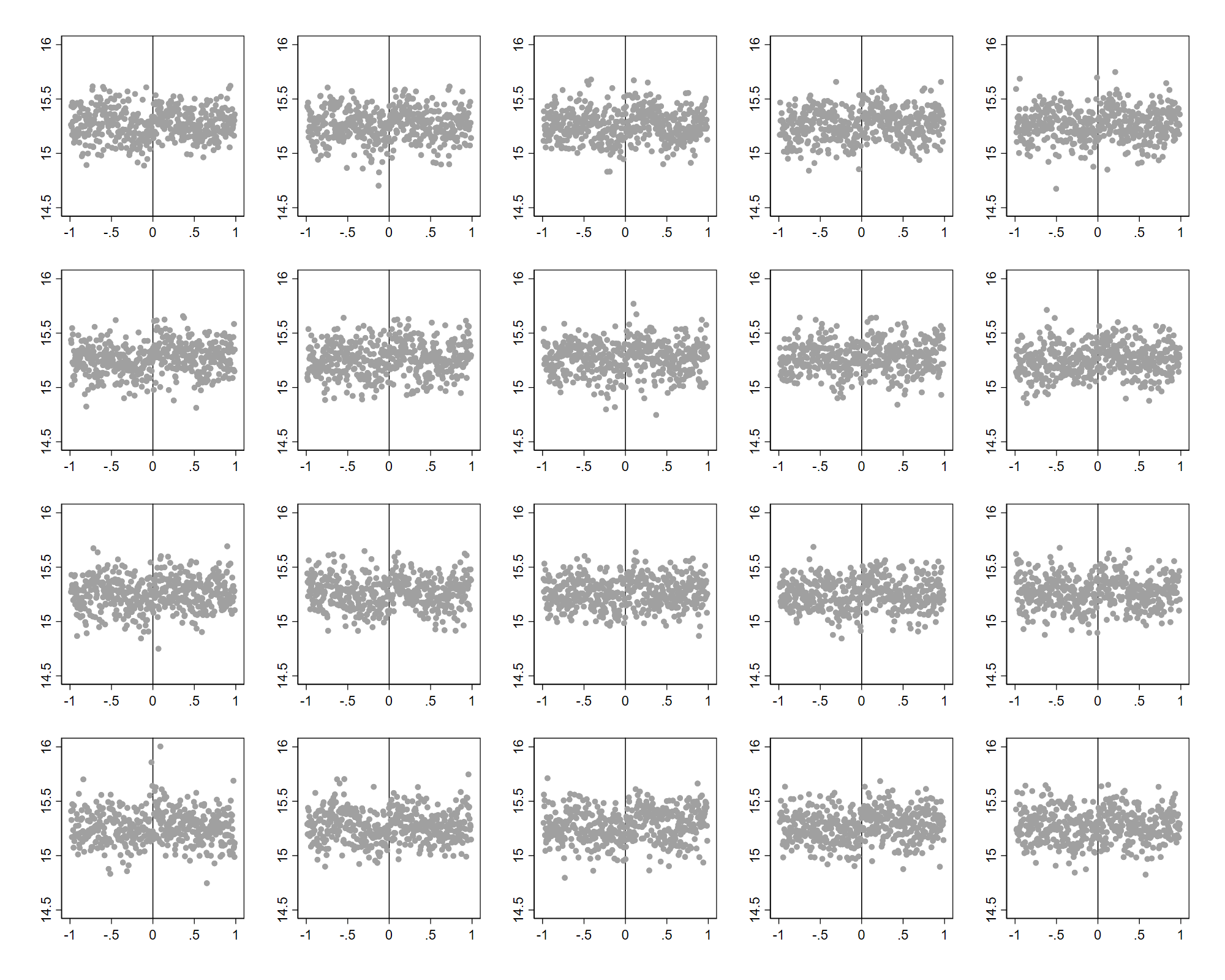}
    \centering{}%
    \begin{minipage}[t]{0.9\columnwidth}%
        {\scriptsize{}Notes: Two of our 11 DGPs feature semi-discrete running variables. We therefore add noise prior to fitting the piecewise quintic CEF to match the continuous running variable condition assumed in \citet{Calonicoetal2015}. To test the null hypothesis that the DGP resulting from fitting the piecewise quintic prior to adding noise is indistinguishable from the DGP where noise is added first, we randomly place one graph from the former distribution among 19 from the latter. The solution for the left lineup is row $5\cdot2-6$ and column $2^{2}-1$ (using conventional matrix index notation). The solution for the right lineup is row $-3\cdot2+10$ and column $9/3-2$.}%
    \end{minipage}
\end{figure}

\begin{figure}[H]
    \caption{\label{fig:CEFs-Fudged-Unfudged}Comparison between CEFs from Original Microdata and Adding Noise to Running Variable}
    \centering
    \includegraphics[width=3.25in]{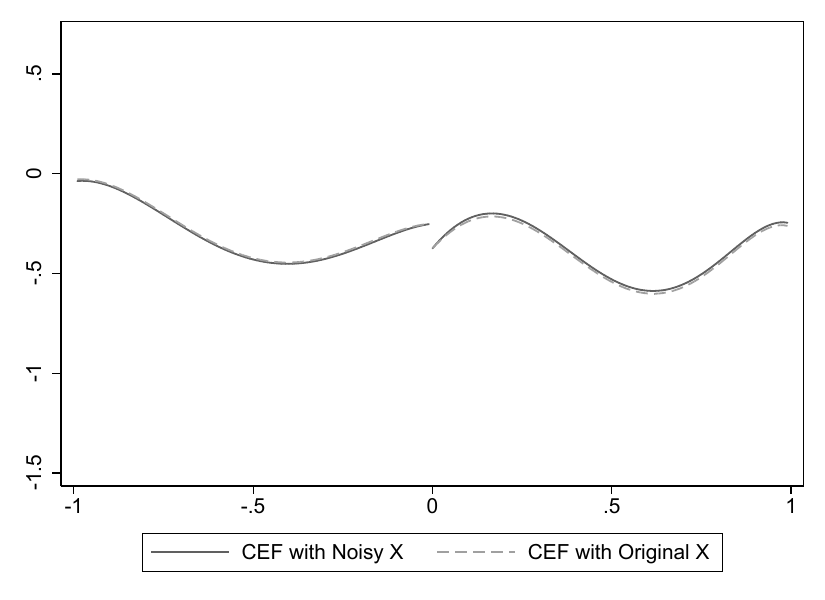}\includegraphics[width=3.25in]{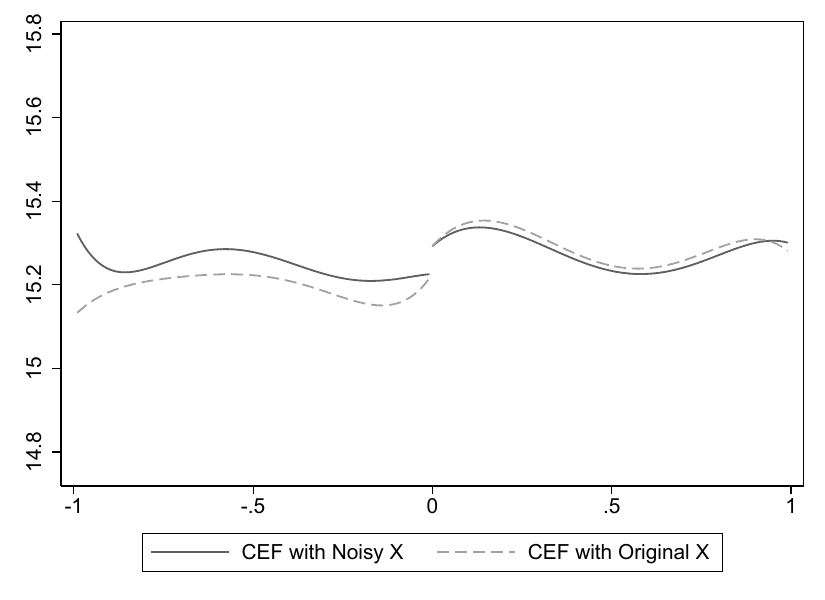}
    \centering{}%
    \begin{minipage}[t]{0.75\columnwidth}%
        {\scriptsize{}Notes: Two of our 11 DGPs feature semi-discrete running variables. We therefore add noise prior to fitting the piecewise quintic CEF to match the continuous running variable condition assumed in \citet{Calonicoetal2015}. The resulting CEFs based on the original data and based on adding noise to the original semi-discrete running variables are plotted.}%
    \end{minipage}
\end{figure}

\begin{figure}[H]
    \caption{\label{fig:RDD-CEF-local}Comparison between Local Linear (Top) and Cubic (Bottom) Conditional Expectation Functions}
    \centering
    \includegraphics[width=5.5in]{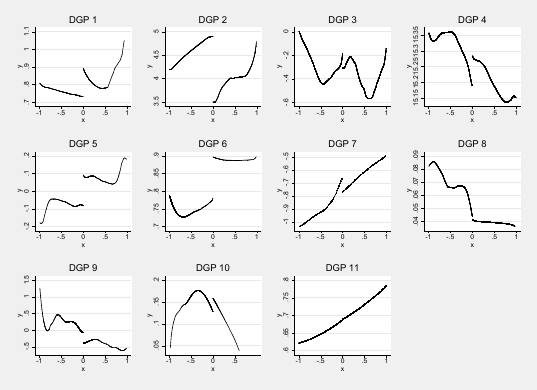}
    \centering
    \includegraphics[width=5.5in]{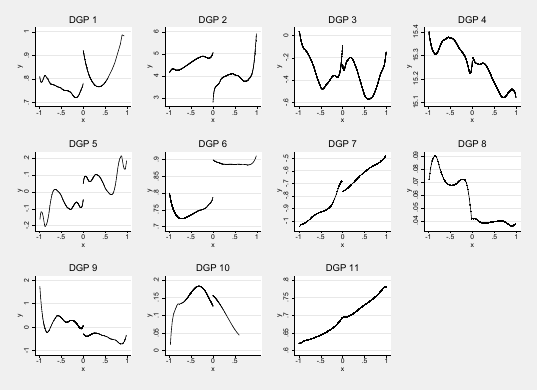}
    \centering{}%
    \begin{minipage}[t]{0.9\columnwidth}%
        {\scriptsize{}Notes: The figure above shows the 11 CEFs constructed using local linear and local cubic specifications. CEFs resulting from the global piecewise quintic specification are shown in Figure \ref{fig:RDD-CEF}.}%
    \end{minipage}
\end{figure}

\begin{landscape}
    \begin{figure}[H]
        \caption{\label{fig:RDD-local-examples}Example Graphs to Compare Different DGP Calibration Choices}
        \centering
        \includegraphics[width=3.25in]{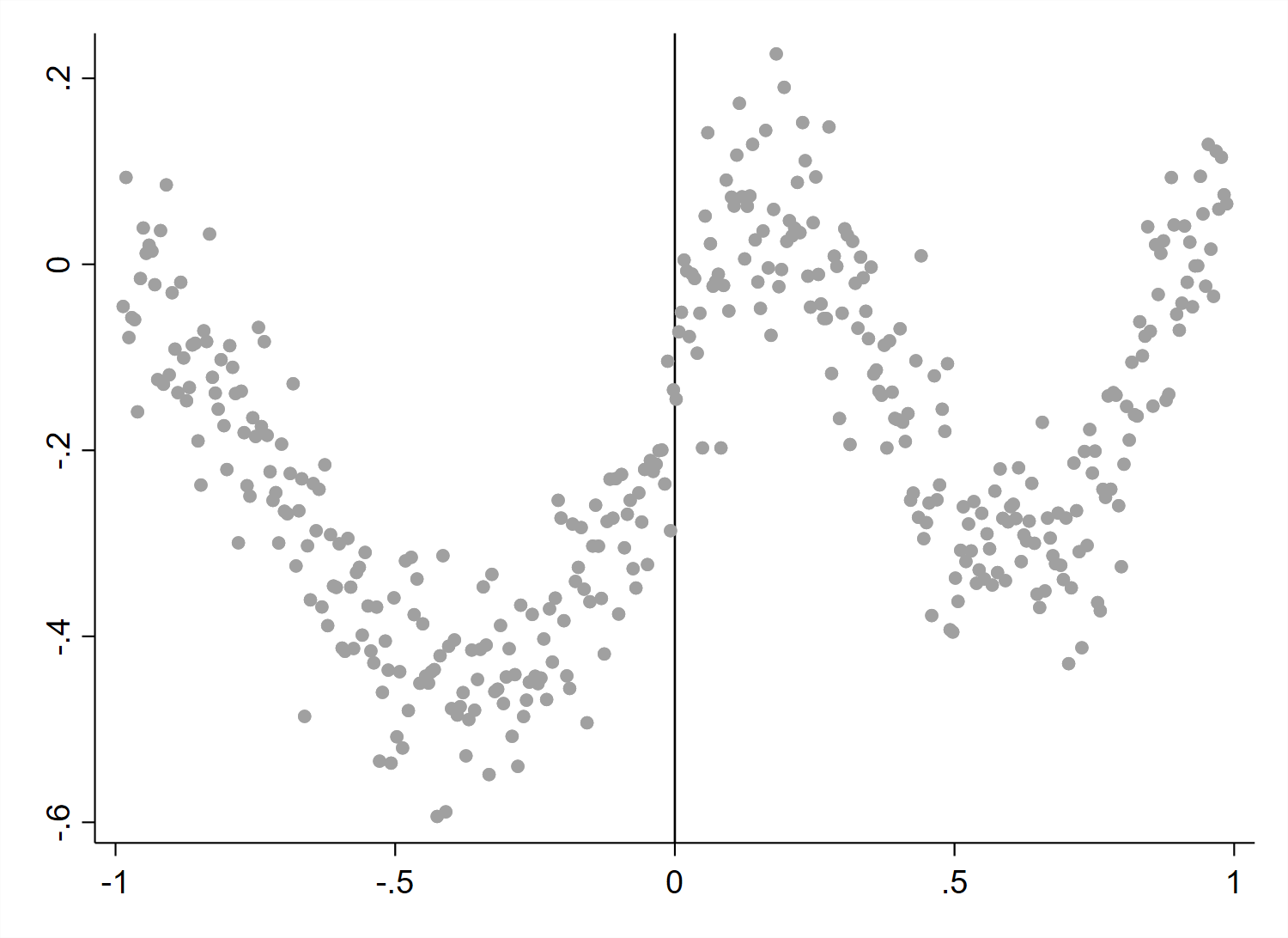}
        \includegraphics[width=3.25in]{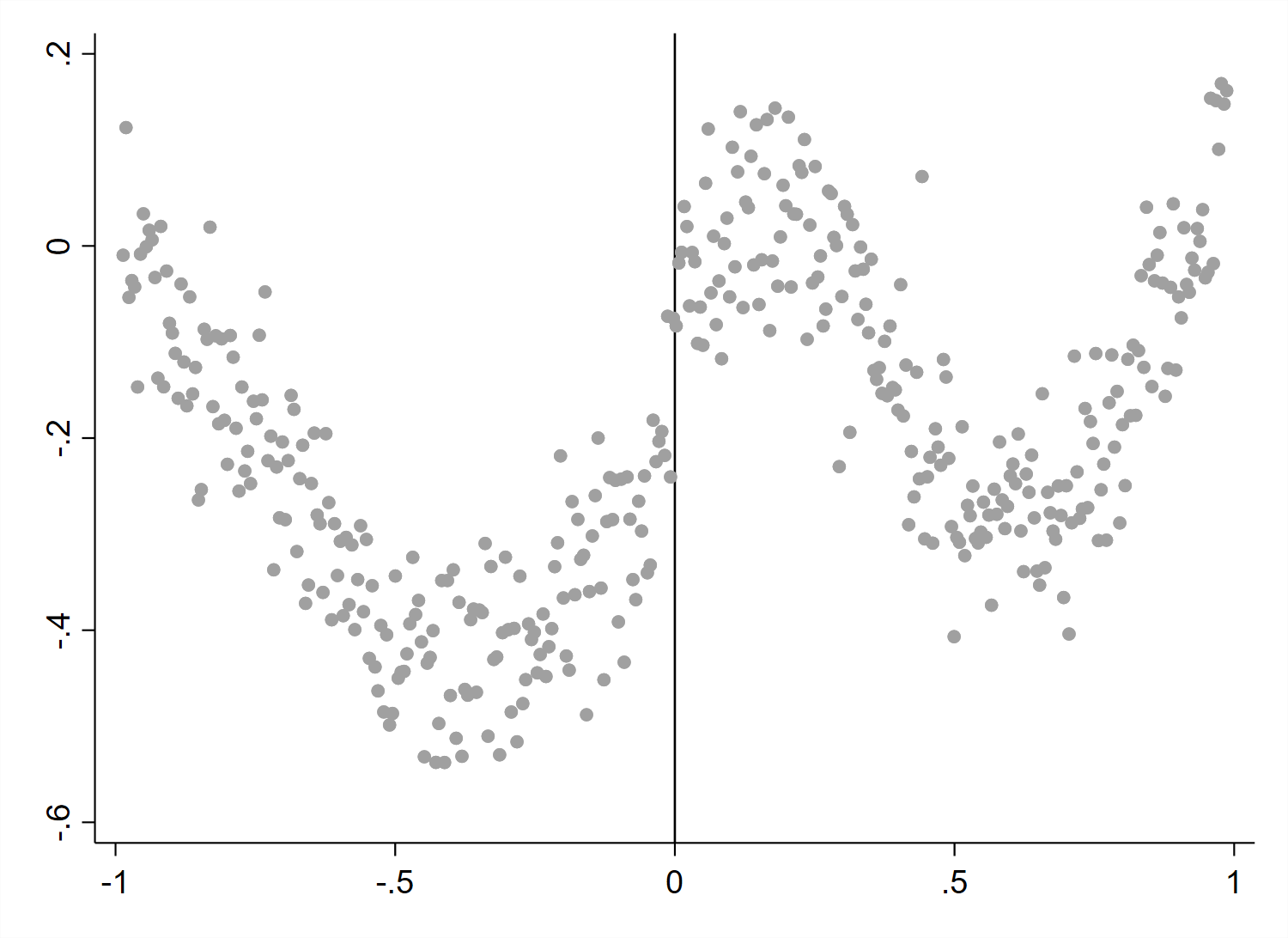}
        \centering
        \includegraphics[width=3.25in]{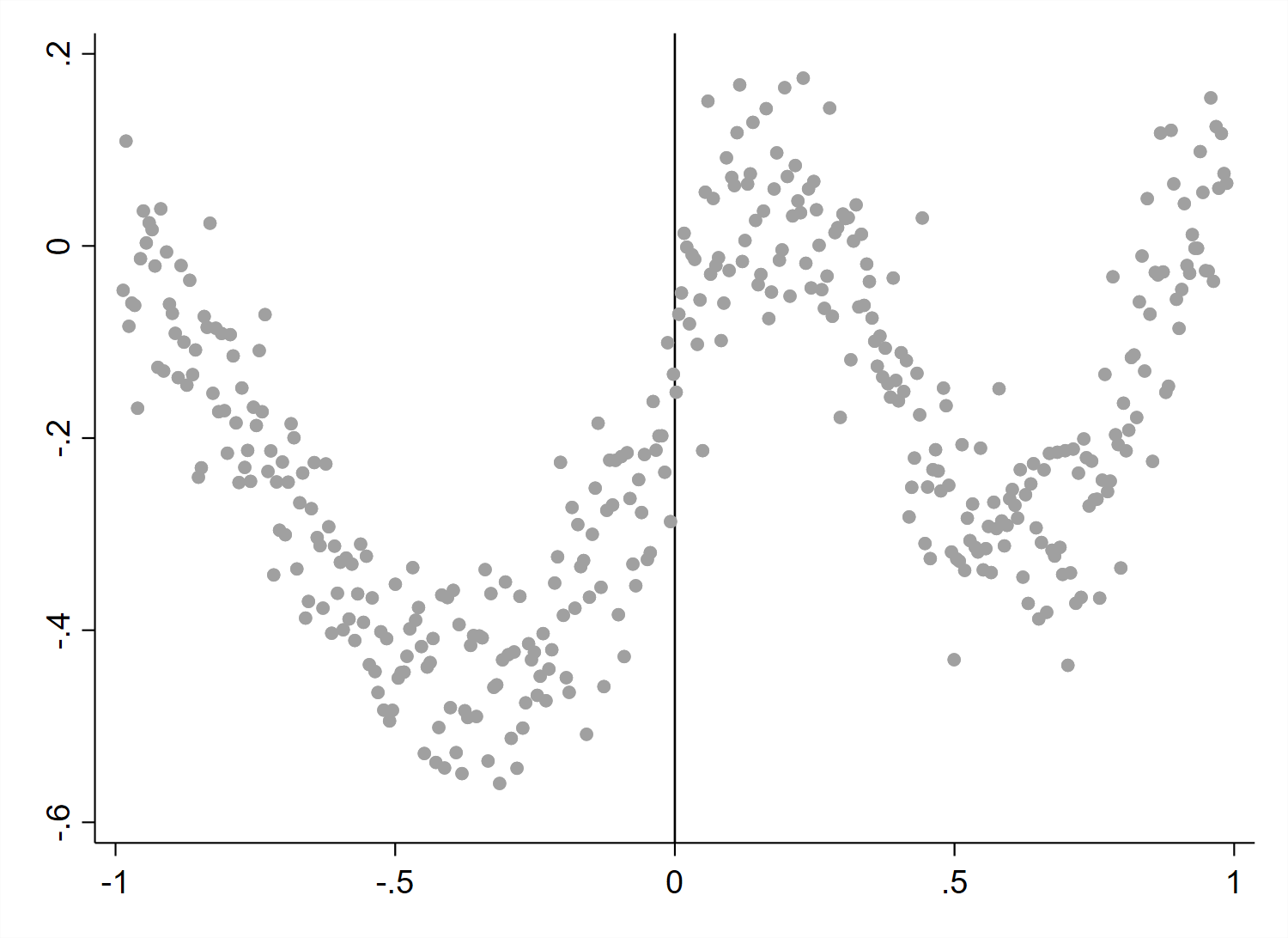}
        \includegraphics[width=3.25in]{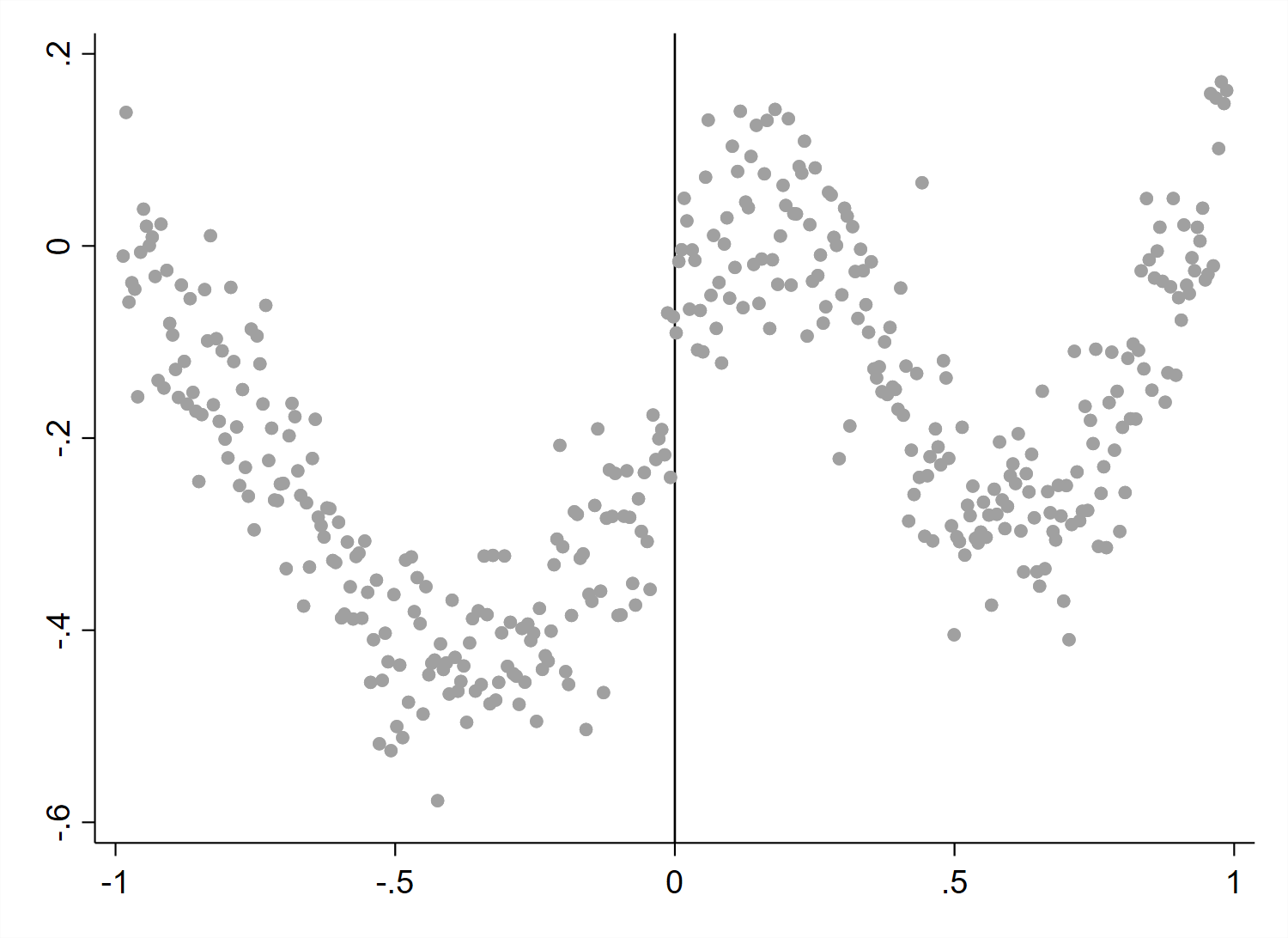}
        \begin{minipage}[t]{0.75\columnwidth}%
            \scriptsize{Notes: Plotted are experimental graphs from the supplemental exercise described in Appendix \ref{sec:dgp-robustness}. Graphs in the left column have CEFs fitted via piecewise global quintic regressions, while graphs in the right column have CEFs fitted via local linear regressions. Graphs in the top row have homoskedastic noise terms, while graphs in the bottom row allow for the variance of the noise term to vary with the running variable.}
        \end{minipage}
    \end{figure}
\end{landscape}

\begin{figure}[H]
    \caption{\label{fig:RDD-Power-RDD5}Supplemental Phase 5: Power Functions}
    \centering\includegraphics[width=5.5in]{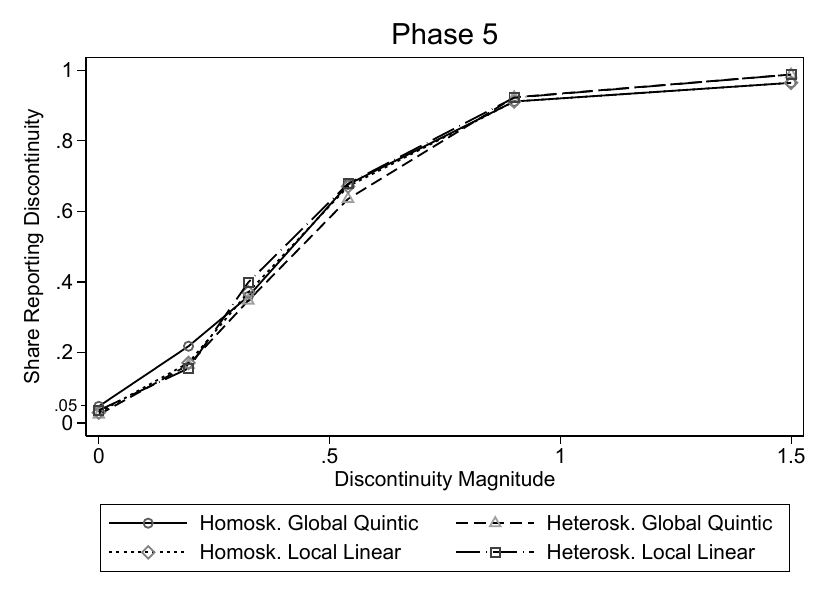}
    \begin{minipage}[t]{0.75\columnwidth}%
        {\scriptsize{}Notes: Plotted are power functions from the supplemental phase 5 to test the sensitivity of visual inference to alternative DGP specifications (while holding fixed the graphical method). We test all four combinations of global quintic/local linear CEFs and homoskedastic/heteroskedastic error terms. In all treatment arms, we use graphs with evenly spaced small bins with a vertical line at $x=0$, no fit lines, and Stata 14’s default spacing. The power functions are defined in Section \ref{sec:conceptual-framework}. The discontinuity magnitude on the $x$-axis is specified as a multiple of the error standard deviation. The $y$-axis represents the share of respondents classifying a graph as having a discontinuity at the policy threshold.}%
    \end{minipage}
\end{figure}

\begin{figure}[H]
    \caption{Video Tutorial Attention Check
    \label{fig:attention-check-bird}}
    \begin{minipage}[t][1\totalheight][c]{0.45\textwidth}%
        \includegraphics[scale=0.2]{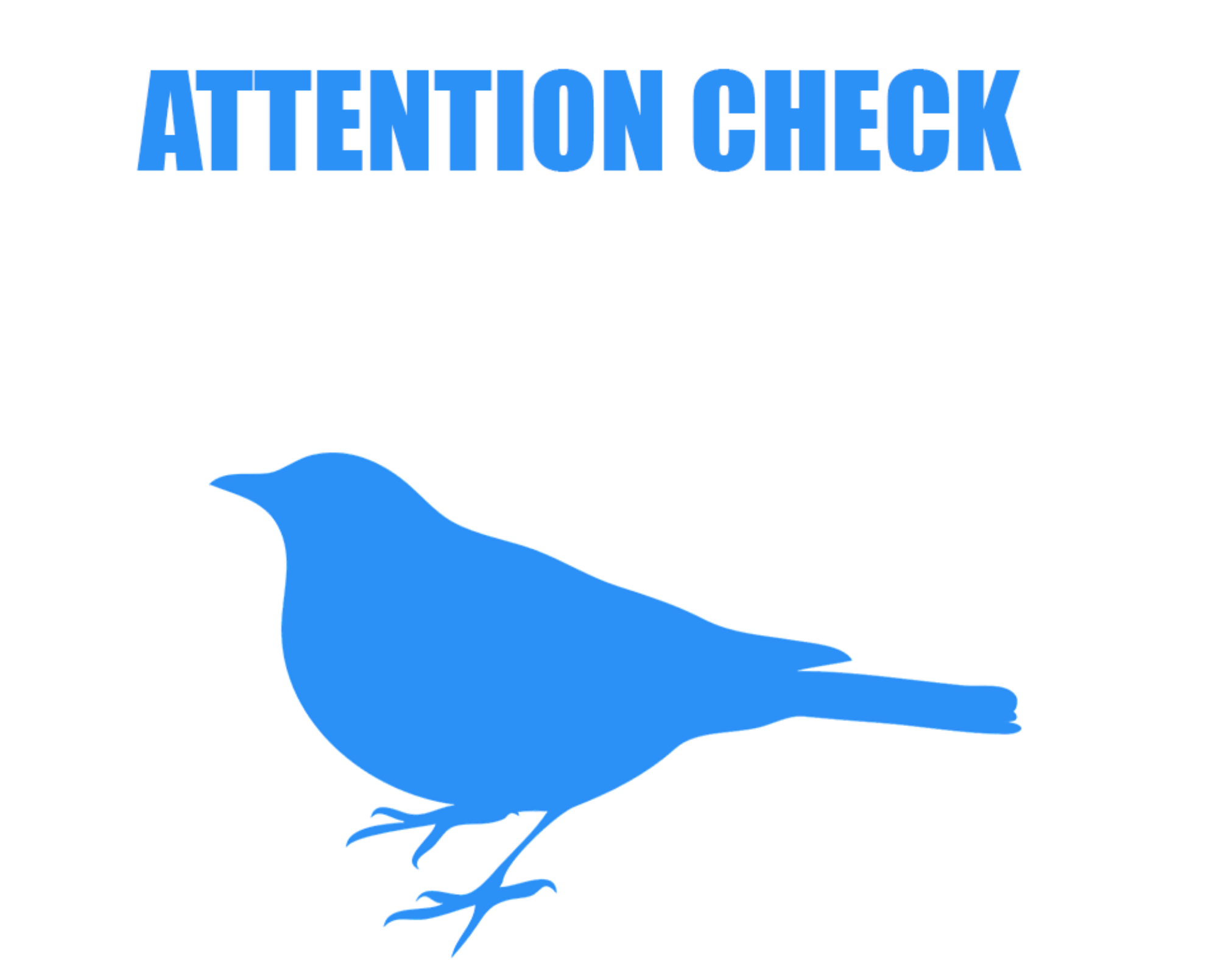}%
    \end{minipage}\hfill{}%
    \begin{minipage}[t][1\totalheight][c]{0.45\textwidth}%
        \includegraphics{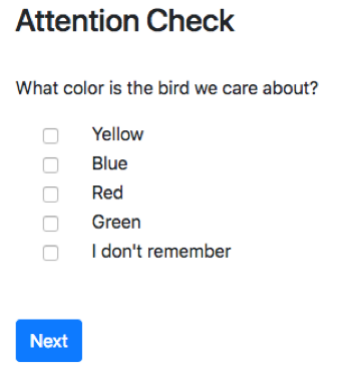}%
    \end{minipage}
\end{figure}

\begin{figure}[H]
    \caption{Example Tasks\label{fig:example-tasks}}
    \begin{minipage}[t]{0.49\textwidth}%
        \subfloat[Example 1]{
        \includegraphics[scale=0.95]{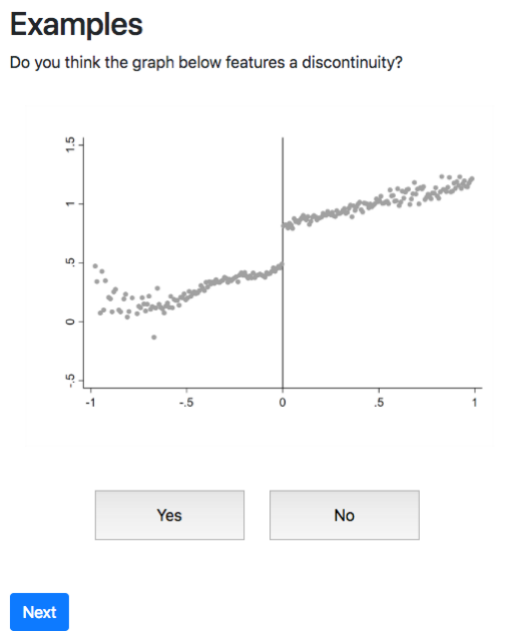}}%
    \end{minipage}\hfill{}%
    \begin{minipage}[t]{0.49\textwidth}%
        \subfloat[Example 1 - Feedback]{
        \includegraphics[scale=0.9]{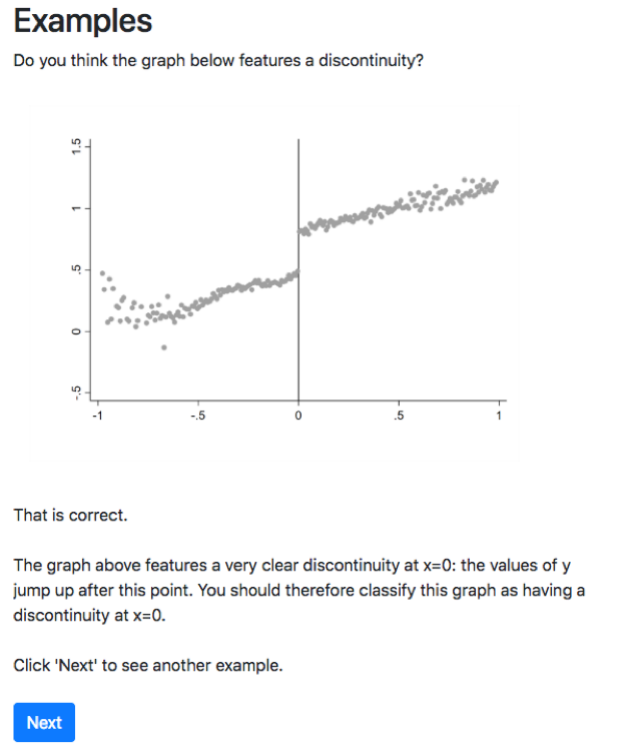}}%
    \end{minipage}
\end{figure}

\begin{figure}[H]
    \caption{Example 4 - Feedback and Navigation Buttons\label{fig:example-tasks-ctd}}
    \includegraphics{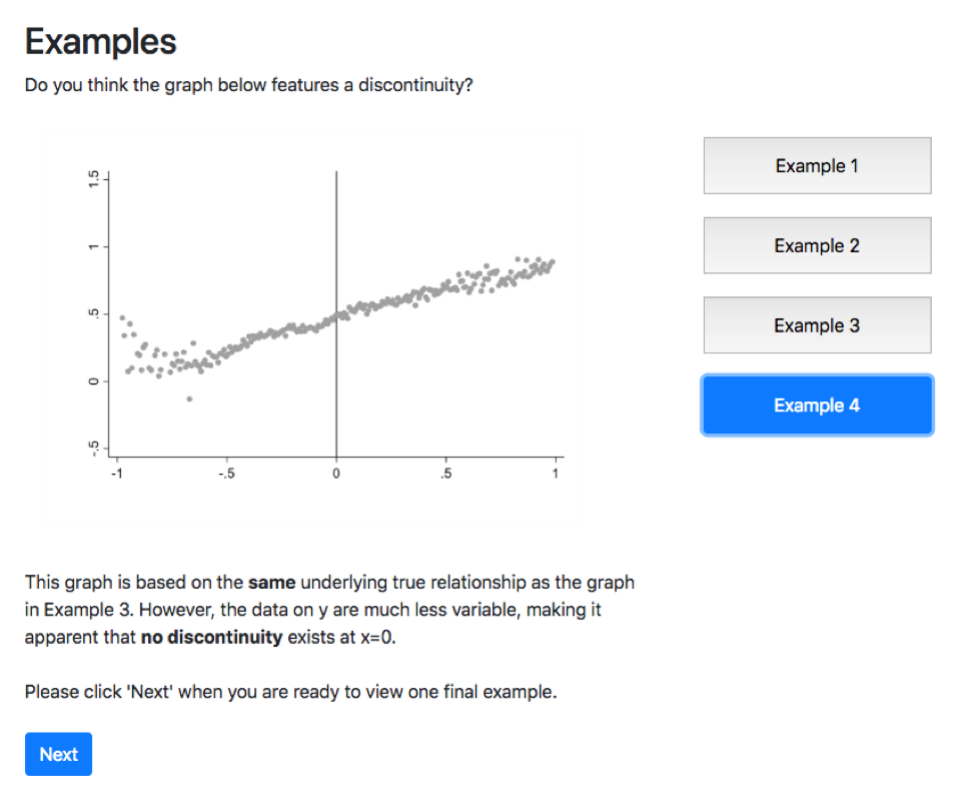}
\end{figure}

\begin{figure}[H]
    \caption{Classification Screen\label{fig:example-classification-screen}}
    \noindent\begin{minipage}[t]{1\columnwidth}%
        \subfloat[Incentives]{
        \includegraphics[scale=0.5]{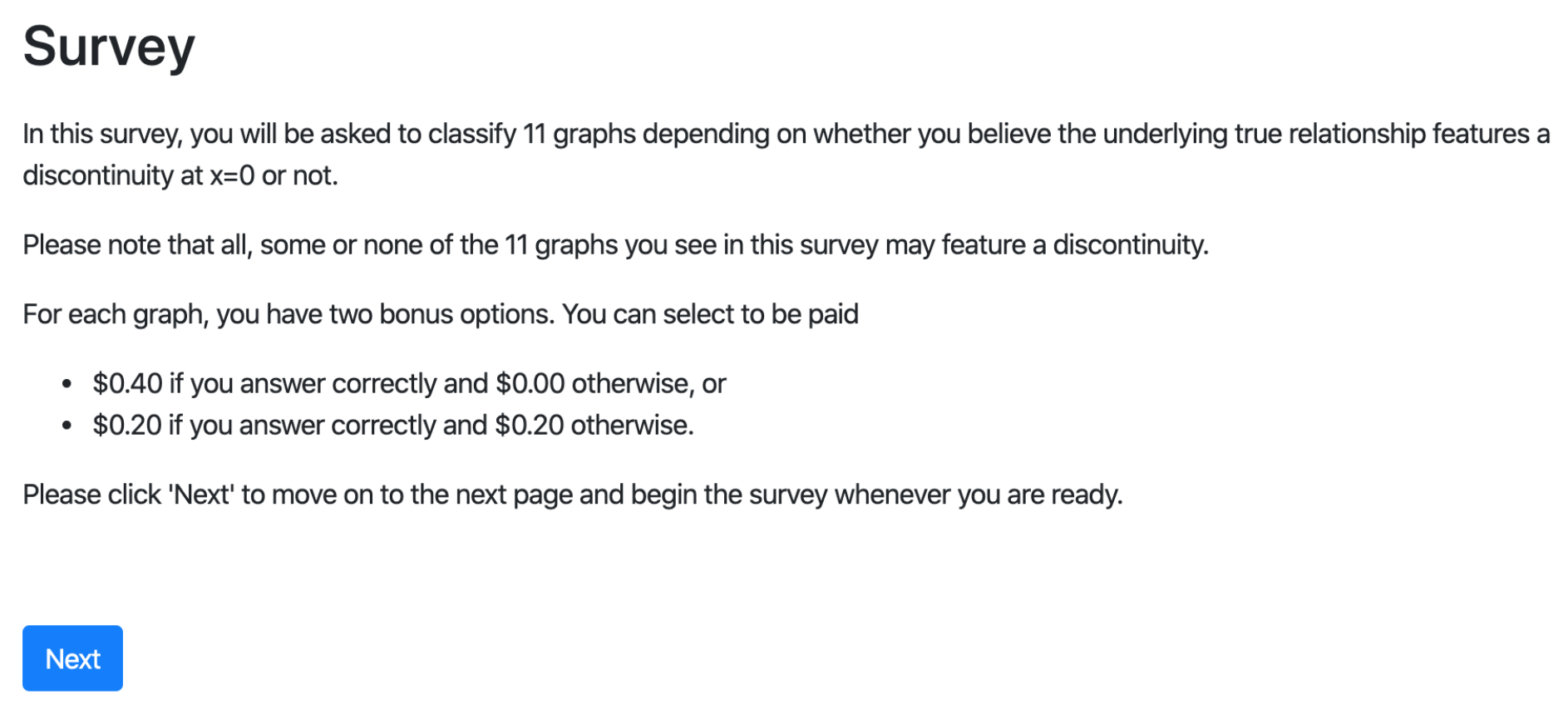}}%
    \end{minipage}
    \noindent\begin{minipage}[t]{1\columnwidth}%
        \subfloat[Classification Task]{
        \includegraphics[scale=0.45]{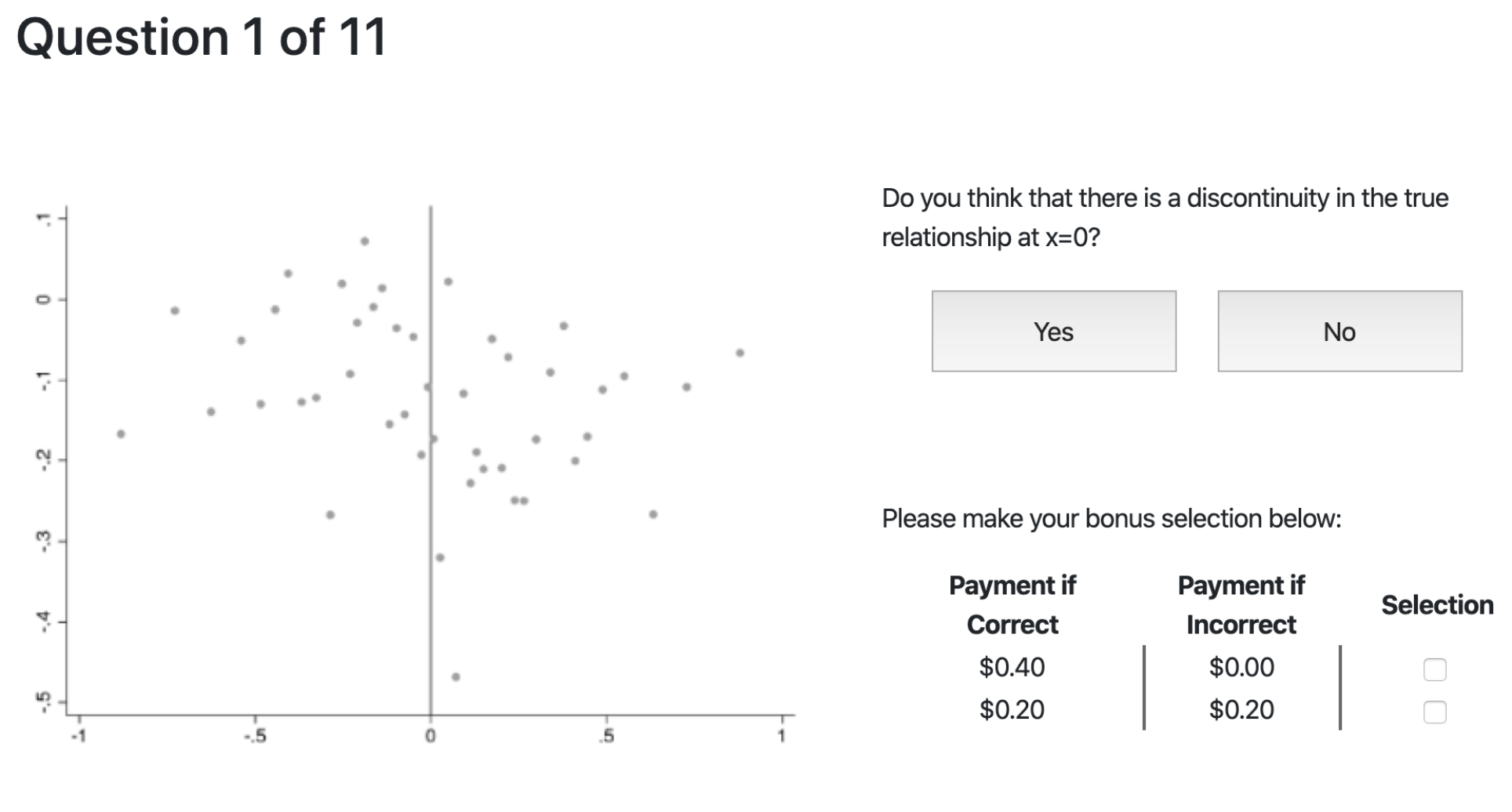}}%
    \end{minipage}
\end{figure}

\begin{figure}[H]
    \caption{Results Screen Example\label{fig:Result-Screen-Example}}
    \includegraphics[scale=0.45]{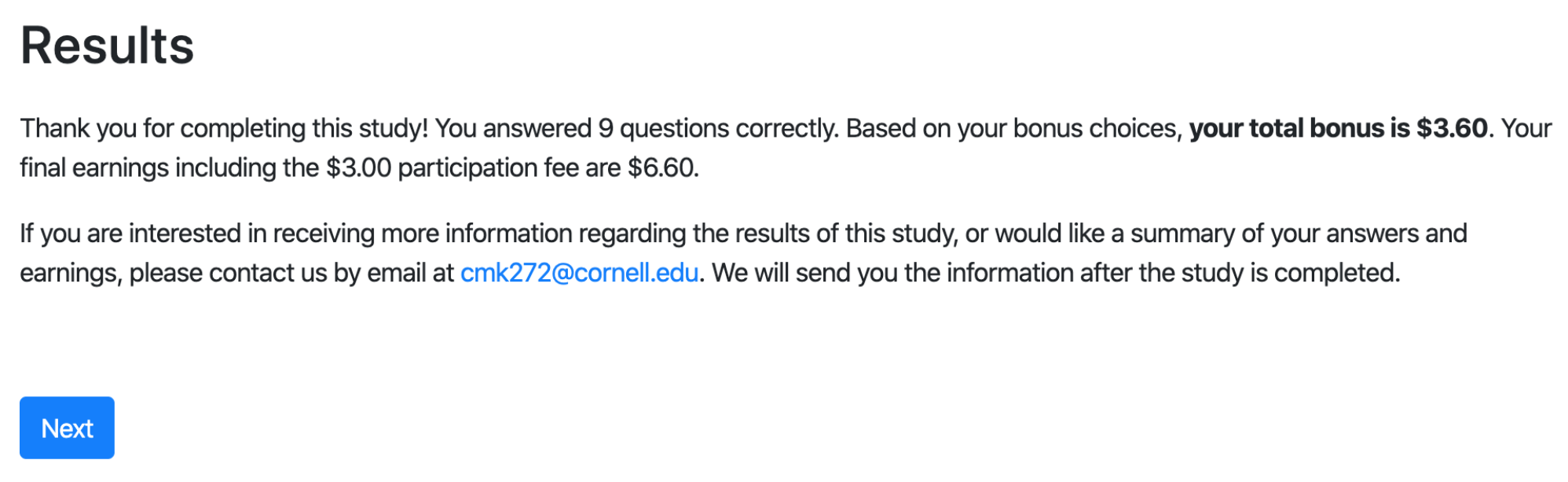}
\end{figure}

\begin{figure}[H]
    \caption{Sequence of Events Experiments\label{fig:Sequence-of-Events-RDD}}
    \centering
    \includegraphics[scale=0.9]{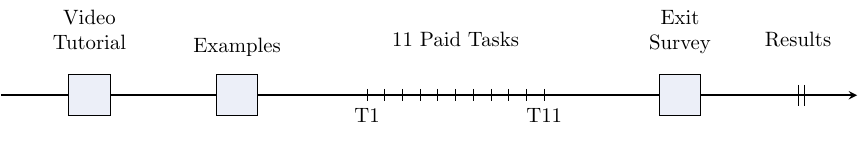}
\end{figure}

\begin{figure}[H]
    \caption{Sequence of Events - Expert Study\label{fig:Sequence-of-Events-Experts}}
    \includegraphics{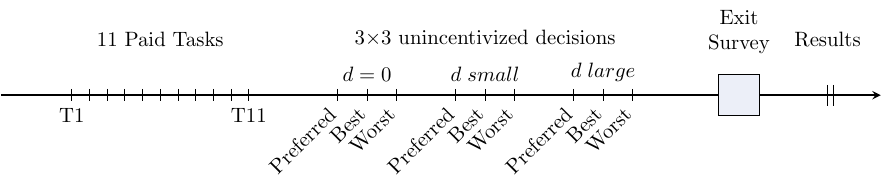}
\end{figure}

\begin{figure}[H]
    \caption{Expert Study - Part 1 Introduction\label{fig:Expert-Survey-Part1-Intro}}
    \includegraphics[scale=0.48]{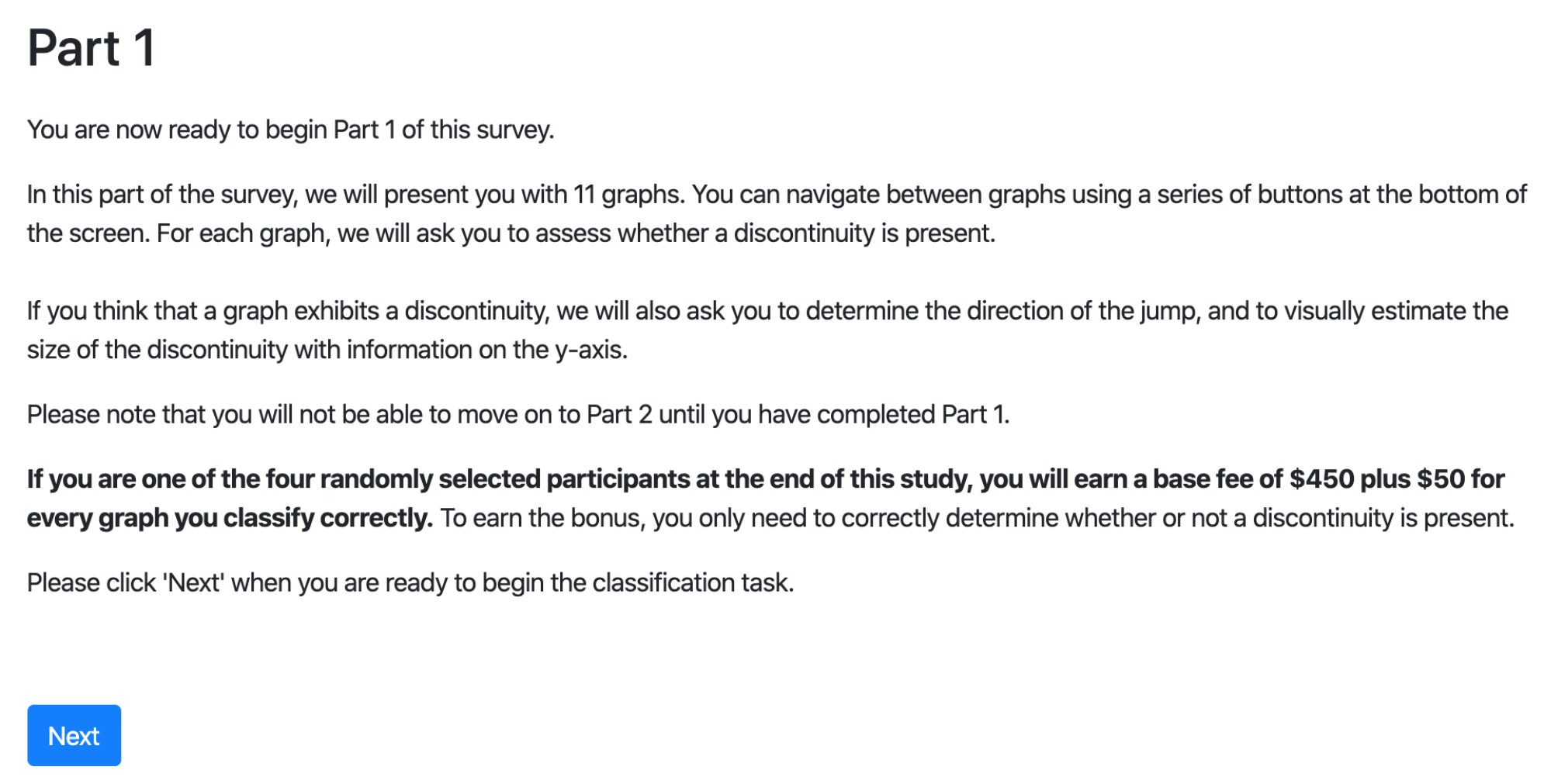}
\end{figure}

\begin{figure}[H]
    \caption{Expert Study - Part 1 Classification Task\label{fig:Expert-Survey-Part1}}
    \centering
    \noindent\begin{minipage}[t]{1\columnwidth}%
        {\includegraphics[scale=0.85]{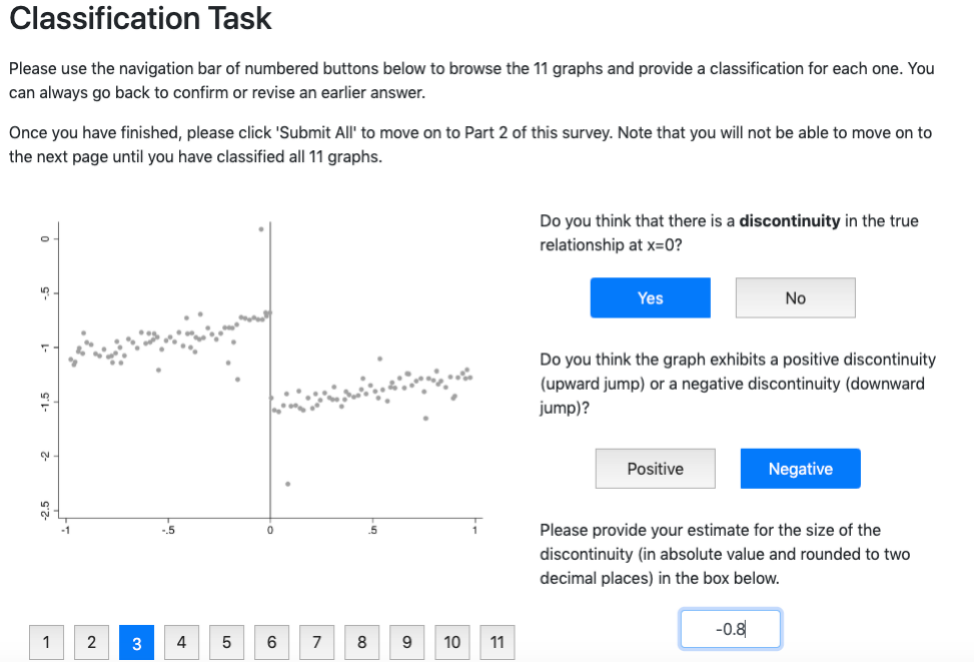}}%
    \end{minipage}
\end{figure}

\begin{figure}[H]
    \caption{Expert Study - Part 2 Instructions\label{fig:Expert-Survey-Part2-Intro}}
    \noindent\begin{minipage}[t]{1\columnwidth}%
        {\includegraphics[scale=0.95]{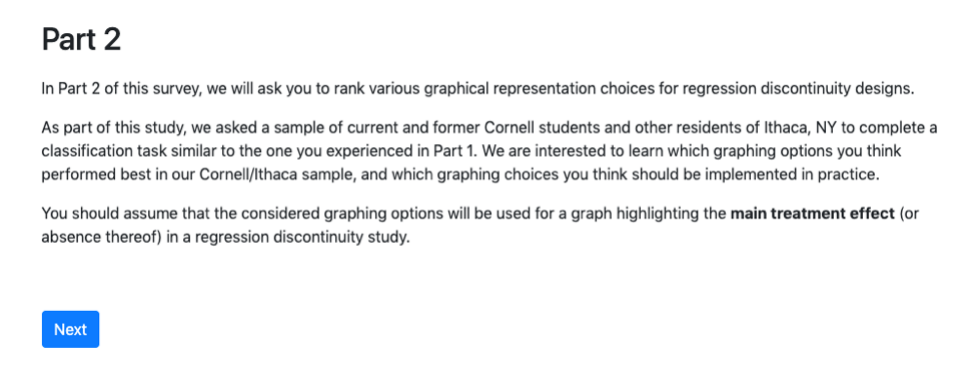}}%
    \end{minipage}
\end{figure}

\begin{figure}[H]
    \caption{Expert Study - Part 2 Decision Screen\label{fig:Expert-Survey-Part2-Decision-Screen}}
    \includegraphics[scale=0.95]{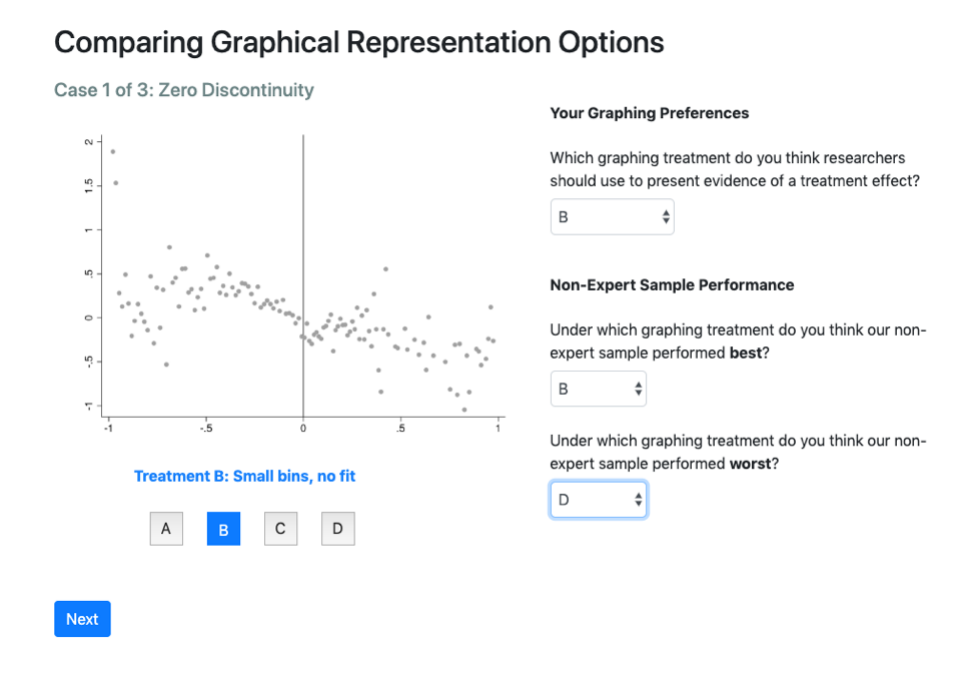}
\end{figure}

\begin{figure}[H]
    \caption{\label{fig:RDD-DGP-Histograms}DGP Histograms}
    \centering
    \includegraphics[width=5.5in]{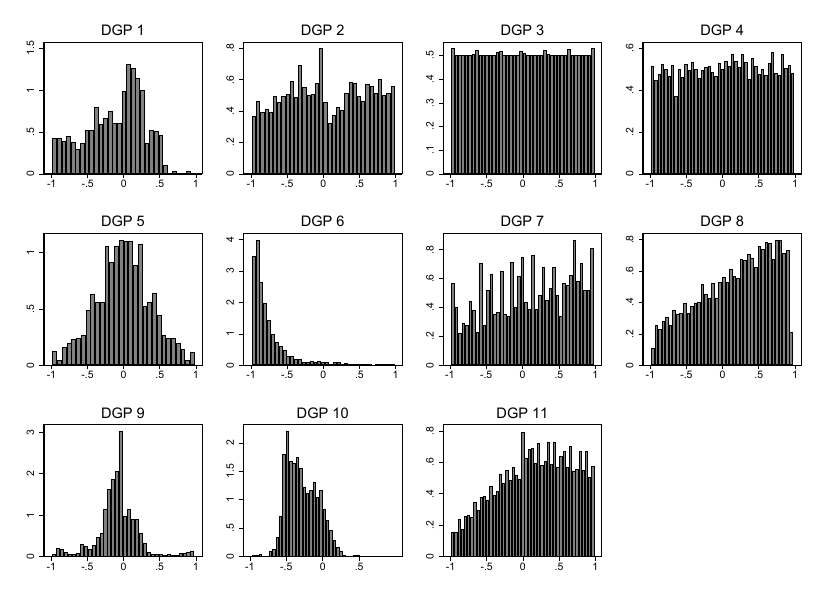}
    \noindent\begin{minipage}[t]{0.8\columnwidth}%
        \noindent {\scriptsize{}Notes: This figure plots the density of each DGP along its running variable. }%
    \end{minipage}{\scriptsize\par}
\end{figure}

\begin{figure}[H]
    \caption{\label{fig:ROC-Phase-6-Symmetric-Asymmetric}Comparison between Power Functions for Phase 4 by Symmetry of DGP Supports}
    \centering
    \includegraphics[width=3.25in]{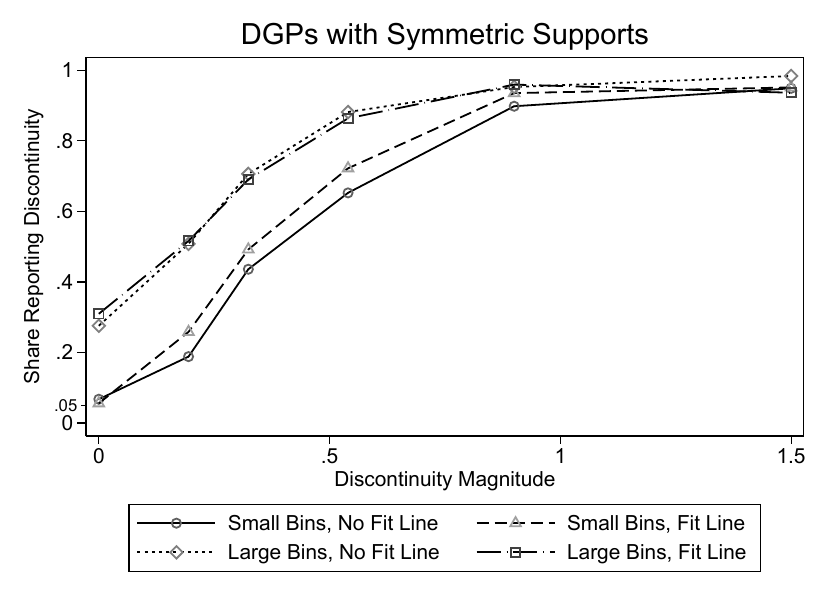}\includegraphics[width=3.25in]{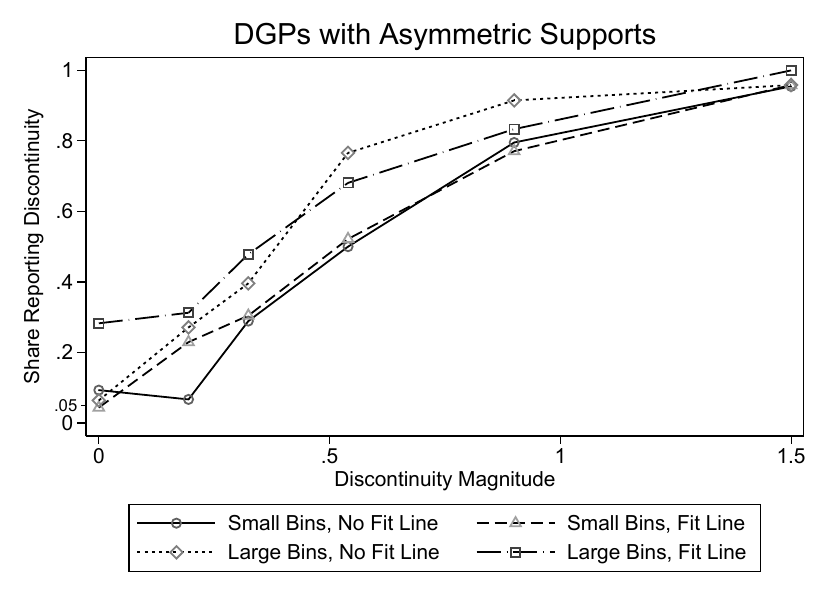}
    \begin{minipage}[t]{\textwidth}%
        \noindent {\scriptsize{}Notes: This figure plots the power functions from phase 4 broken down into DGPs with symmetric supports prior to normalization and those with asymmetric supports prior to normalization. \textit{Large bin} corresponds to the \citet{Calonicoetal2015} bin width selector that minimizes the integrated mean squared error of the bin-average estimators of the conditional expectation function; \textit{small bin} corresponds to the \citet{Calonicoetal2015} bin width selector that aims to approximate the variability of the underlying data; \textit{fit line} indicates the presence of parametric fit lines.}%
    \end{minipage}{\scriptsize\par}
\end{figure}

\begin{figure}[H]
    \caption{\label{fig:Gini-Spacing-Performance}DGP Gini Coefficients and Non-Expert Performance}
    \centering\includegraphics[width=3.25in]{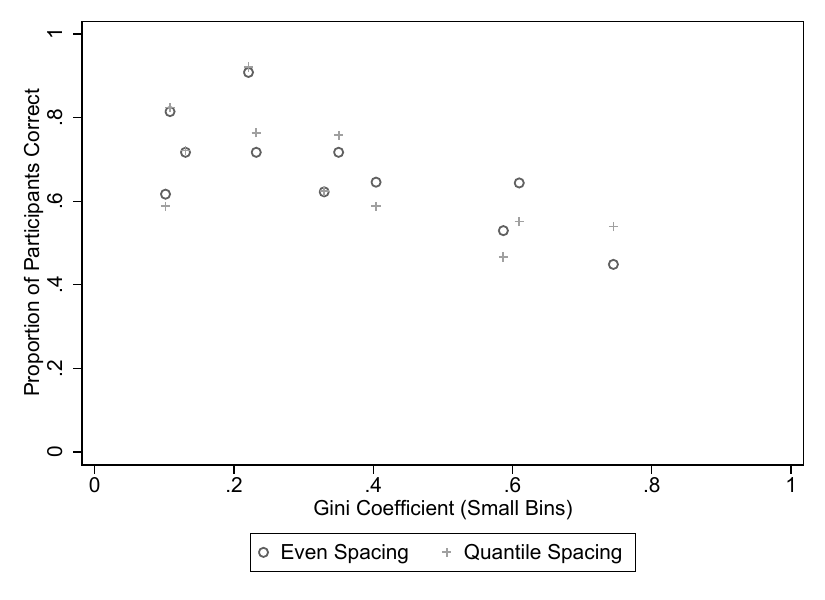}\includegraphics[width=3.25in]{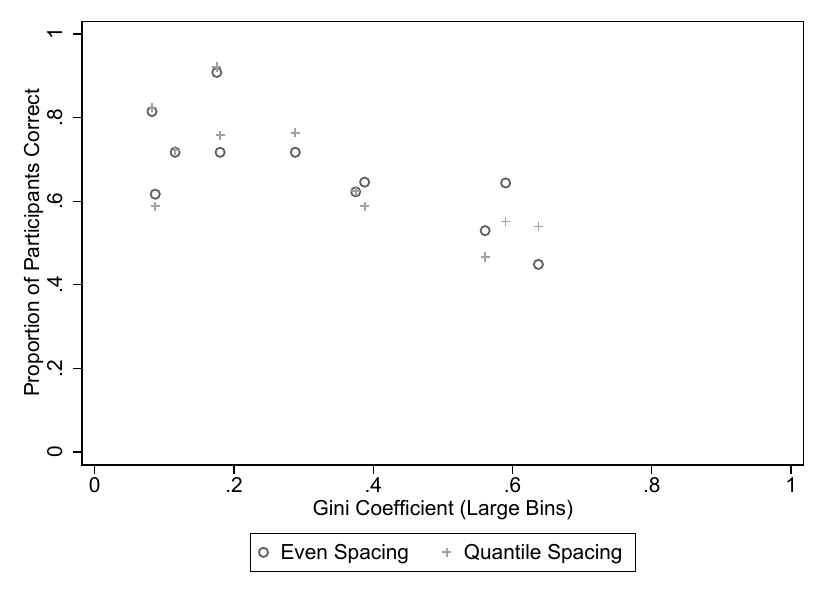}
    \begin{minipage}[t]{\textwidth}%
        \noindent {\scriptsize{}Notes: This figure compares non-expert accuracy in classifying RDD graphs by DGP based on the uniformity of the distribution of the running variable as measured by the Gini coefficient. The left panel contains results for graphs with small bins, and the right panel contains results for graphs with large bins.}%
    \end{minipage}{\scriptsize\par}
\end{figure}

\begin{figure}[H]
    \caption{\label{fig:Slope-Discontinuity-Direction} Discontinuity Easier to Detect When Both Slopes at Cutoff are in the Same Direction as the Discontinuity}
    \begin{minipage}[t]{0.49\textwidth}%
        \subfloat[Positive Slopes and Positive Discontinuity]{
        \includegraphics[width=3.25in]{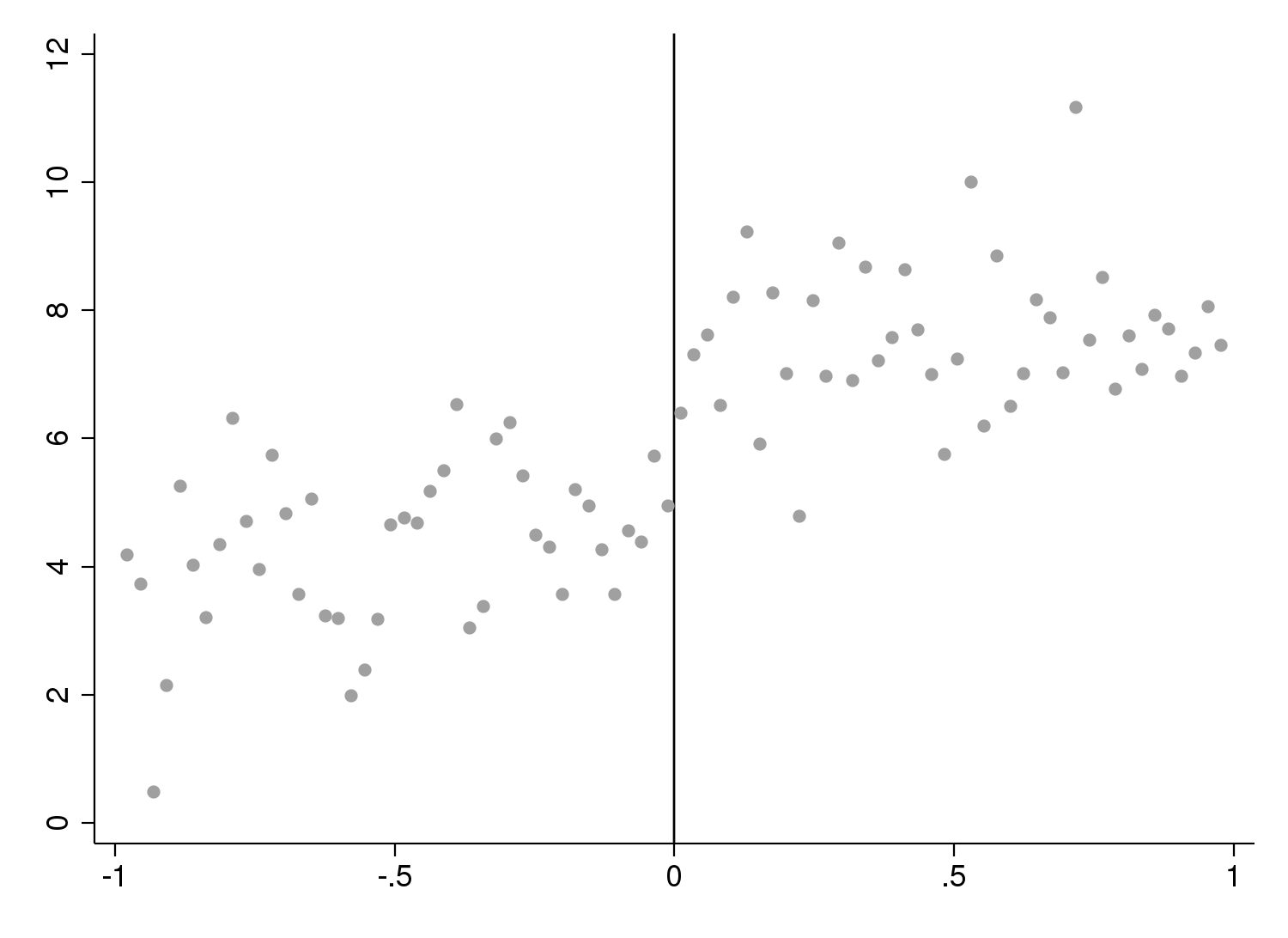}}%
    \end{minipage}\hfill{}%
    \begin{minipage}[t]{0.49\textwidth}%
        \subfloat[Positive Slopes and Negative Discontinuity]{
        \includegraphics[width=3.25in]{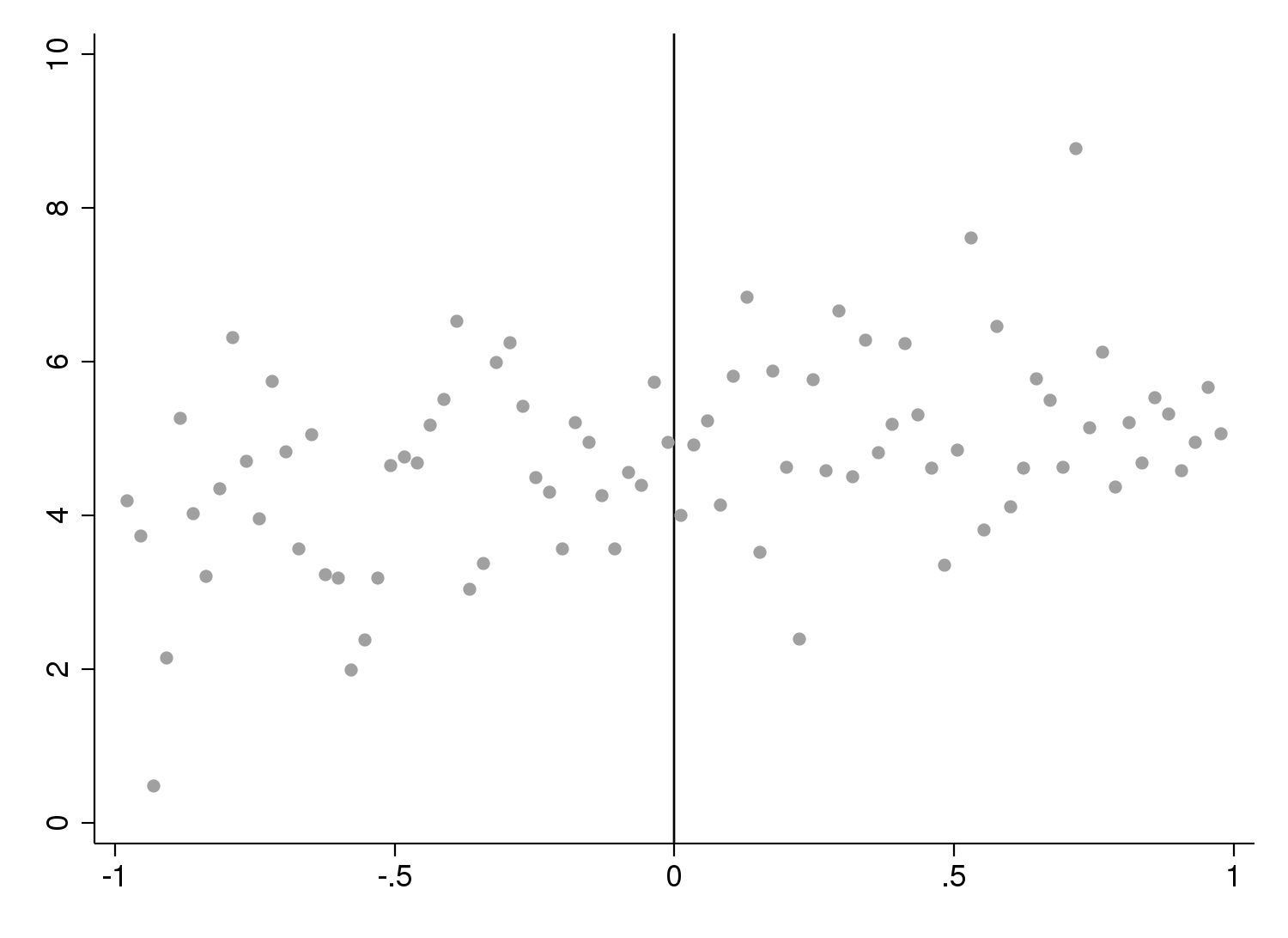}}%
    \end{minipage}
    \begin{minipage}[t]{\textwidth}%
        \noindent {\scriptsize{}Notes: This figure shows experimental graphs for a DGP for which the left and right first derivatives at the treatment threshold are positive when the discontinuity is also positive (left) and when the discontinuity is negative (right). Our regression results in Table \ref{tab:DGP-Direction-Prediction-Power} show that visual inference performs better when the discontinuity direction matches the signs of both the left and right derivatives.}%
    \end{minipage}{\scriptsize\par}
\end{figure}

\begin{landscape}
    \begin{figure}[H]
        \caption{Discontinuity Classifications over 11 Graphs}
        \label{fig:RDD-classifications-over-time}
        \includegraphics[width=3.15in]{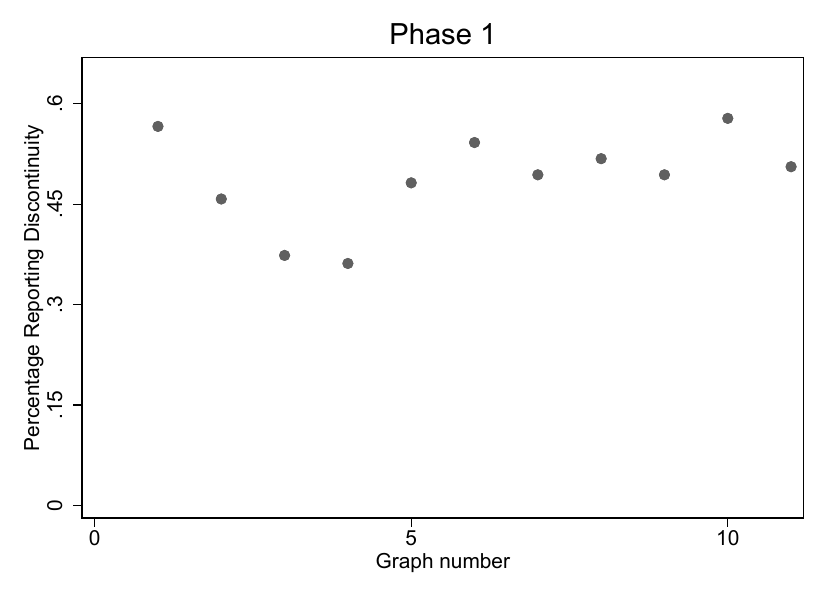}
        \includegraphics[width=3.15in]{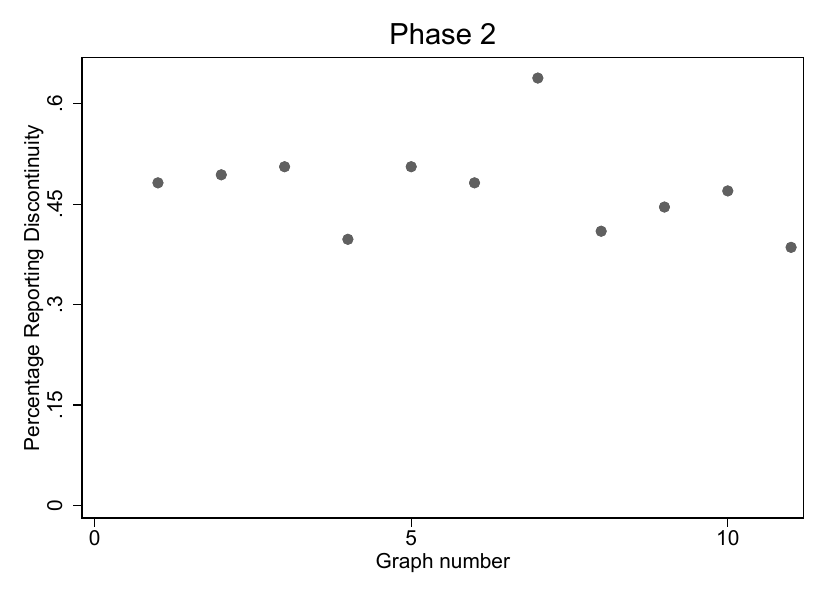}
        \includegraphics[width=3.15in]{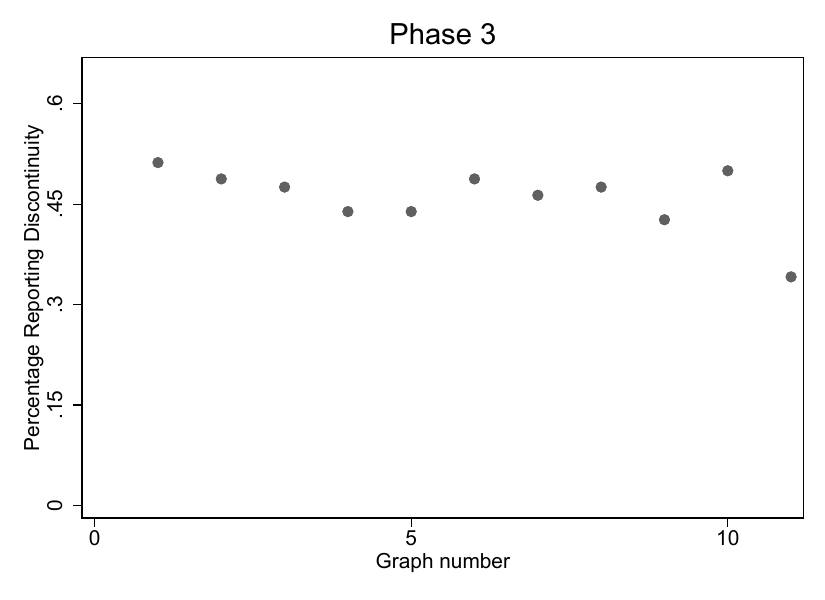}
        \includegraphics[width=3.15in]{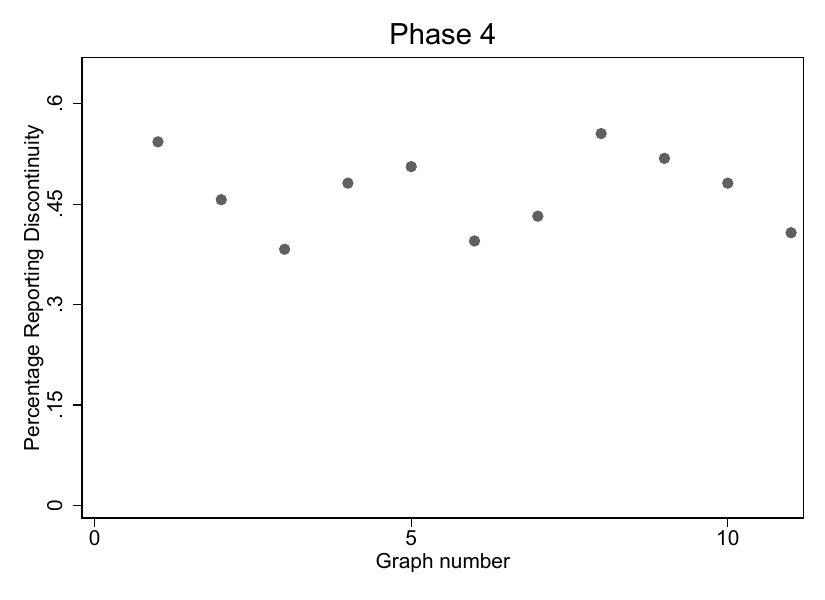}
        \includegraphics[width=3.15in]{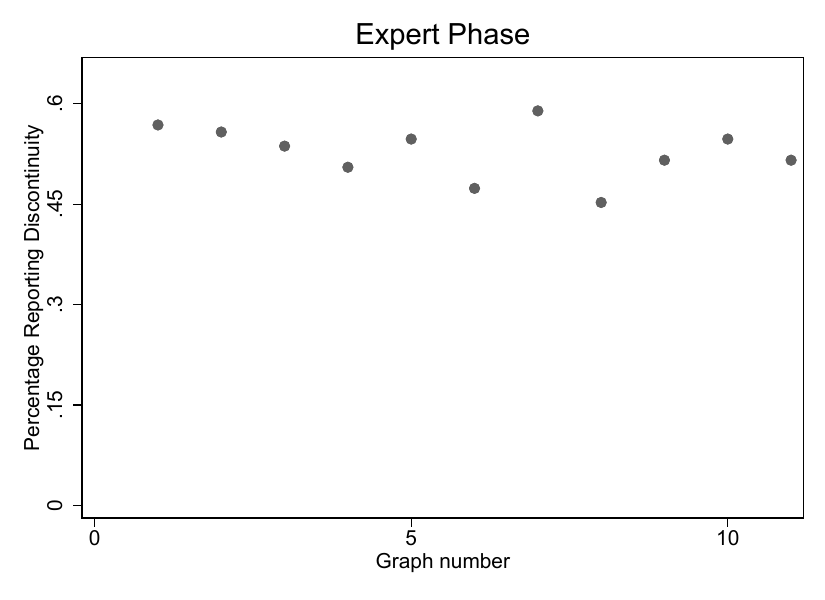}
    \end{figure}
    \noindent {\scriptsize{}Notes: This figure plots participants' likelihood of reporting a discontinuity over the course of the 11 graphs seen in the study by experimental phase.}
 
    \begin{figure}[H]
        \caption{\label{fig:RDD-Power-Rescaled}Power Functions against Asymptotic $t$-Statistics by Phase}
        \centering
        \includegraphics[width=3.5in]{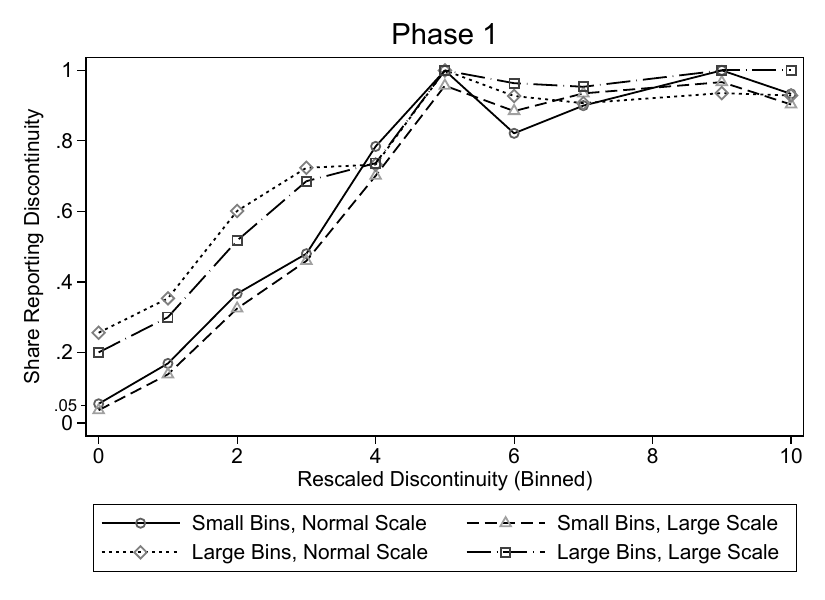}
        \includegraphics[width=3.5in]{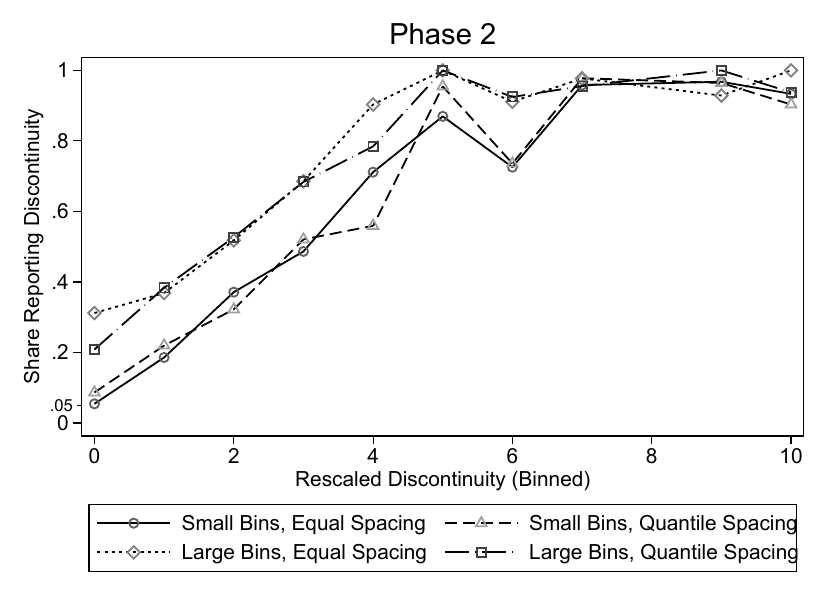}
        \includegraphics[width=3.5in]{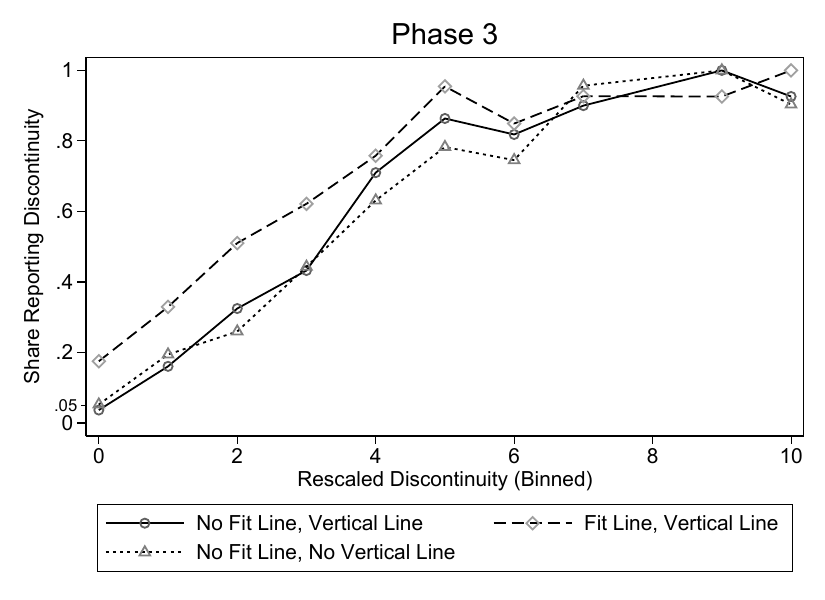}\includegraphics[width=3.5in]{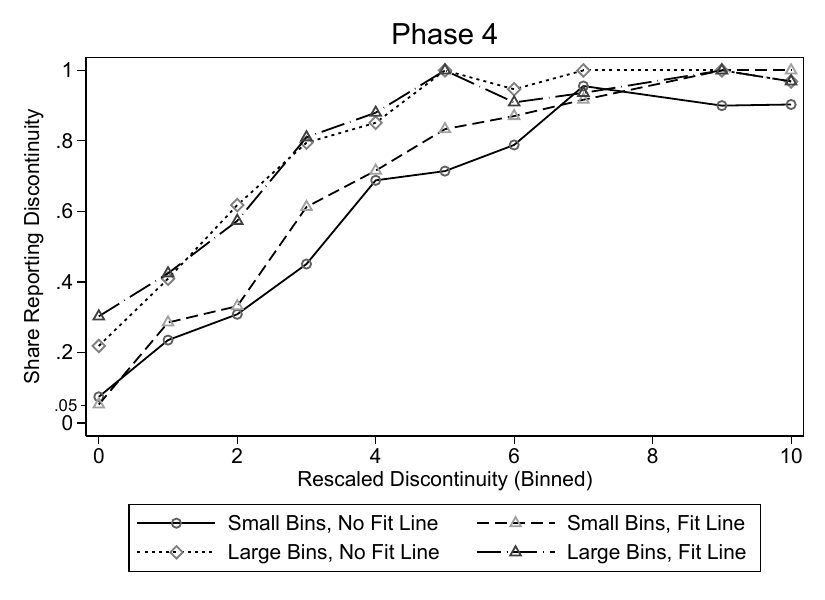}
        \begin{minipage}[t]{1\columnwidth}%
        \noindent {\scriptsize{}Notes: Plotted are power functions from the main four non-expert experiments broken down by phase. The power functions are defined in Section \ref{sec:conceptual-framework}. The $x$-axis represents binned values of the asymptotic $t$-statistic discussed in Appendix \ref{sec:t-stat}. \textit{Large bins} corresponds to the \citet{Calonicoetal2015} bin width selector that minimizes the integrated mean squared error of the bin-average estimators of the conditional expectation function; \textit{small bins} corresponds to the \citet{Calonicoetal2015} bin width selector that aims to approximate the variability of the underlying data; \textit{quantile spacing} indicates that bins were spaced by quantiles rather than evenly spaced; \textit{fit line} indicates the presence of parametric fit lines; \textit{vertical line} indicates the presence of a vertical line at the policy threshold; \textit{normal scale} corresponds to the default $y$-axis scaling using Stata 14; \textit{large scale} doubles that default range.}%
    \end{minipage}{\scriptsize\par}
    \end{figure}

    \begin{figure}[H]
        \caption{\label{fig:RDD-Power-By-DGP-Rescaled}Power Functions by DGP and Phase against Asymptotic $t$-Statistics}
        \centering
        \includegraphics[width=3.5in]{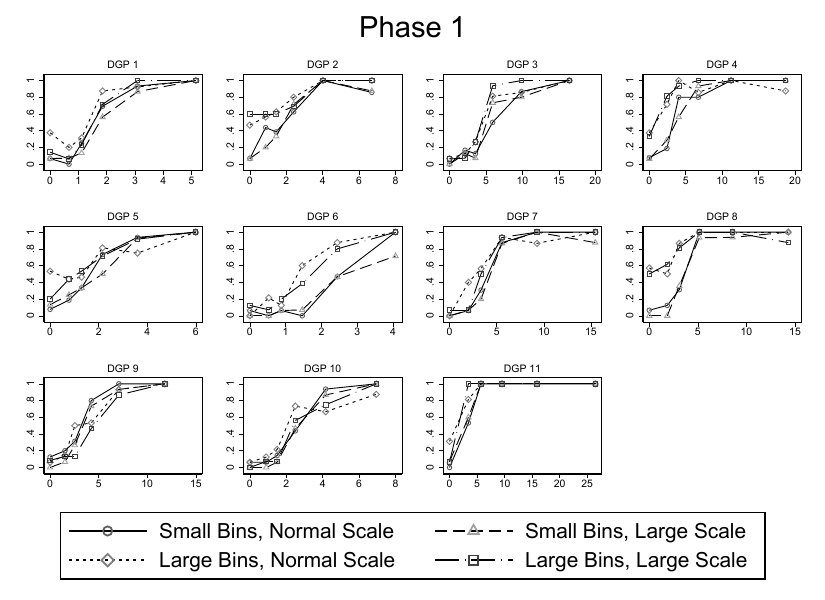}
        \includegraphics[width=3.5in]{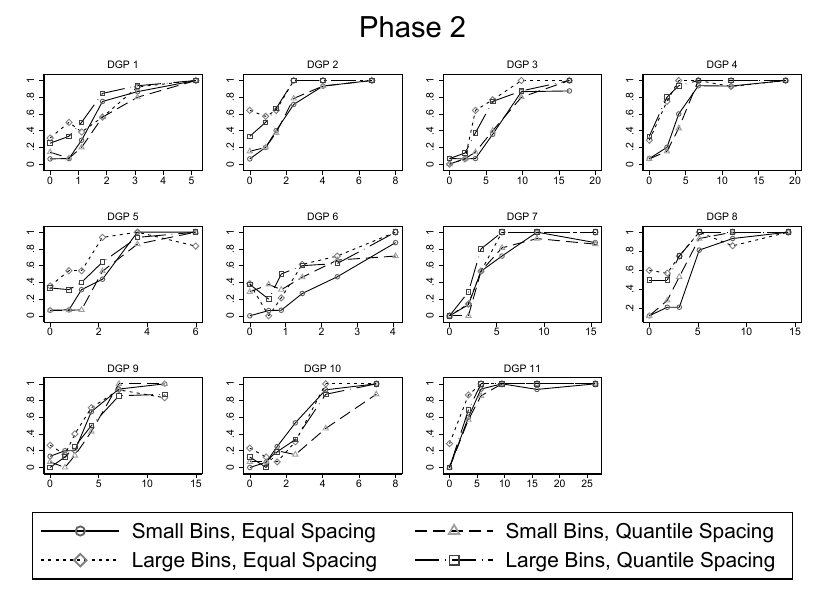}
        \includegraphics[width=3.5in]{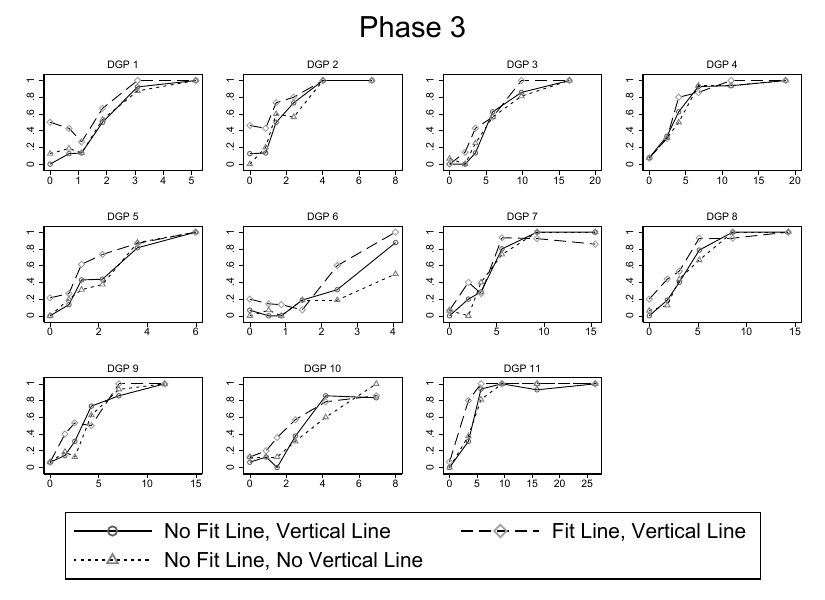}
        \includegraphics[width=3.5in]{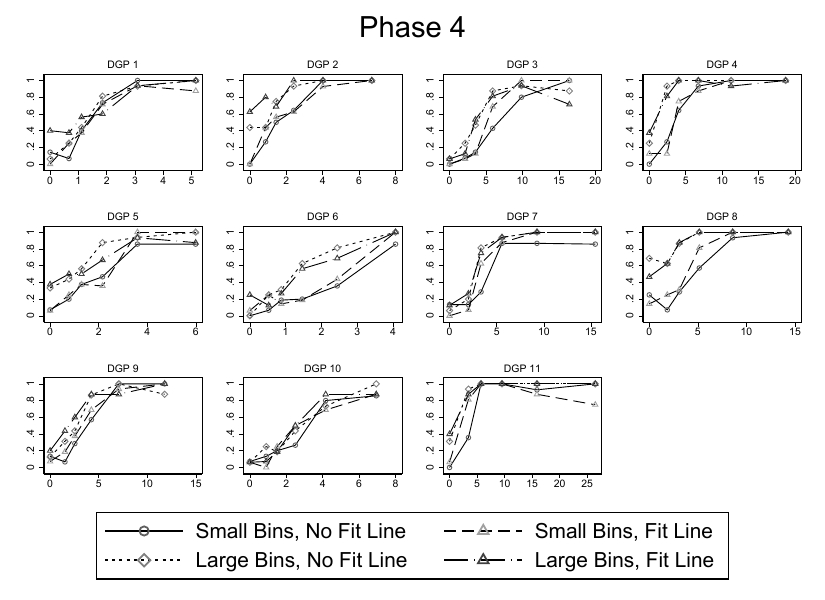}
        \noindent\begin{minipage}[t]{1\columnwidth}%
            \noindent {\scriptsize{}Notes: Plotted are power functions from the main four non-expert experiments broken down by DGP per phase. The power functions are defined in Section \ref{sec:conceptual-framework}. The $x$-axis represents binned values of the asymptotic $t$-statistic discussed in Appendix \ref{sec:t-stat}. The $y$-axis represents the share of respondents classifying a graph as having a discontinuity at the policy threshold. \textit{Large bins} corresponds to the \citet{Calonicoetal2015} bin width selector that minimizes the integrated mean squared error of the bin-average estimators of the conditional expectation function; \textit{small bins} corresponds to the \citet{Calonicoetal2015} bin width selector that aims to approximate the variability of the underlying data; \textit{quantile spacing} indicates that bins were spaced by quantiles rather than evenly spaced; \textit{fit line} indicates the presence of parametric fit lines; \textit{vertical line} indicates the presence of a vertical line at the policy threshold; \textit{normal scale} corresponds to the default $y$-axis scaling using Stata 14; \textit{large scale} doubles that default range.}
        \end{minipage}{\scriptsize\par}
    \end{figure}
\end{landscape}




\begin{figure}[H]
    \caption{\label{fig:RDD-Experts-Nonexperts-Diff}Expert vs Non-Expert Performance Differences}
    \centering
    \includegraphics[width=3.25in]{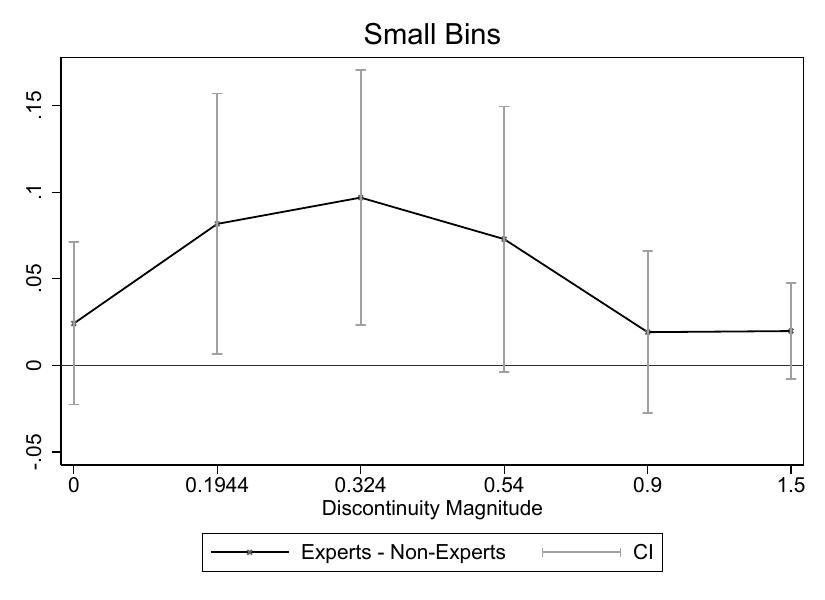}\includegraphics[width=3.25in]{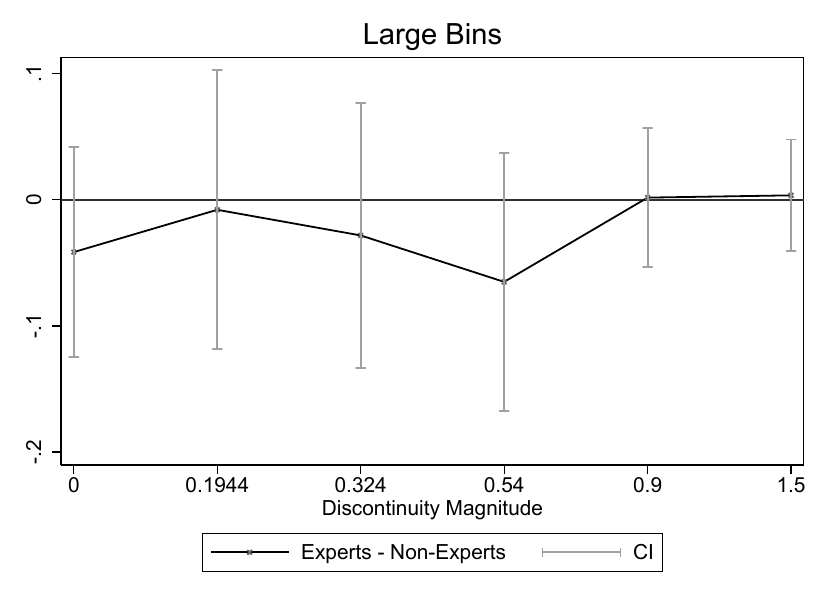}
    \begin{minipage}[t]{\textwidth}%
        \noindent {\scriptsize{}Notes: These figures plot the difference in the power functions between experts and non-experts from Figure \ref{fig:Experts-Nonexperts-ROC}. We compute 95\% confidence intervals using the large sample approximation described at the end of Appendix \ref{sec:Proofs} and by assuming independence between the experts and non-experts. \textit{Large bins} corresponds to the \citet{Calonicoetal2015} bin width selector that minimizes the integrated mean squared error of the bin-average estimators of the conditional expectation function; \textit{small bins} corresponds to the \citet{Calonicoetal2015} bin width selector that aims to approximate the variability of the underlying data.}%
    \end{minipage}{\scriptsize\par}
\end{figure}

\begin{landscape}
    \begin{figure}[H]
        \caption{\label{fig:RDD-Experts-Vs-All-Diff}Expert Visual vs Econometric Inference Power Function Differences}
        \centering
        \includegraphics[width=3.25in]{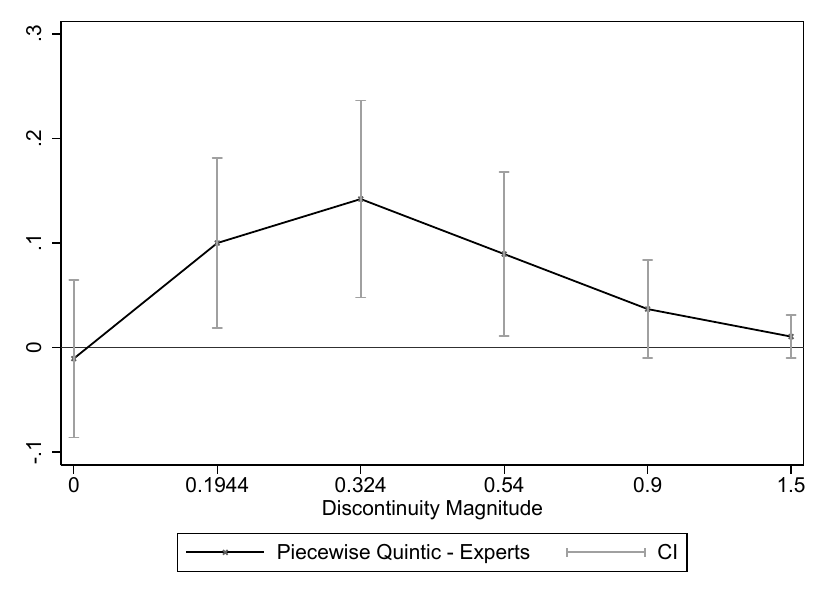}
        \includegraphics[width=3.25in]{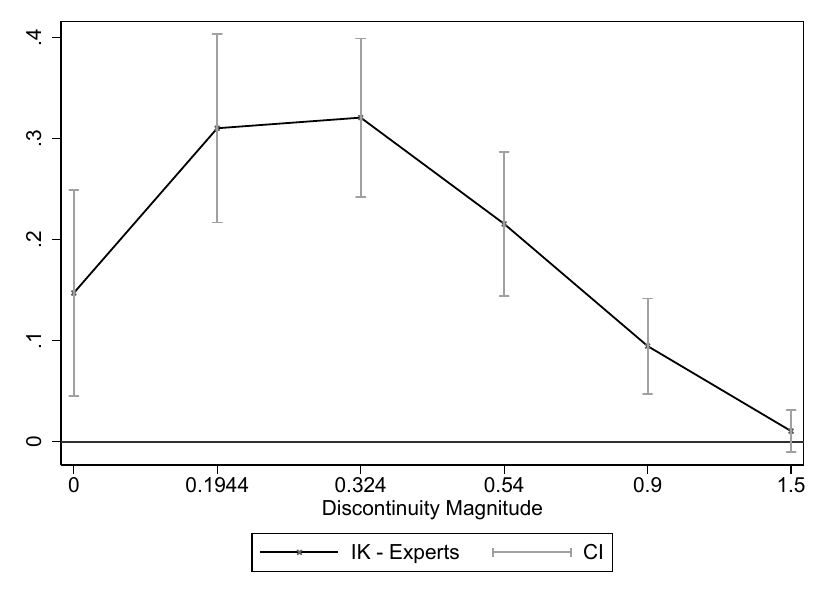}
        \includegraphics[width=3.25in]{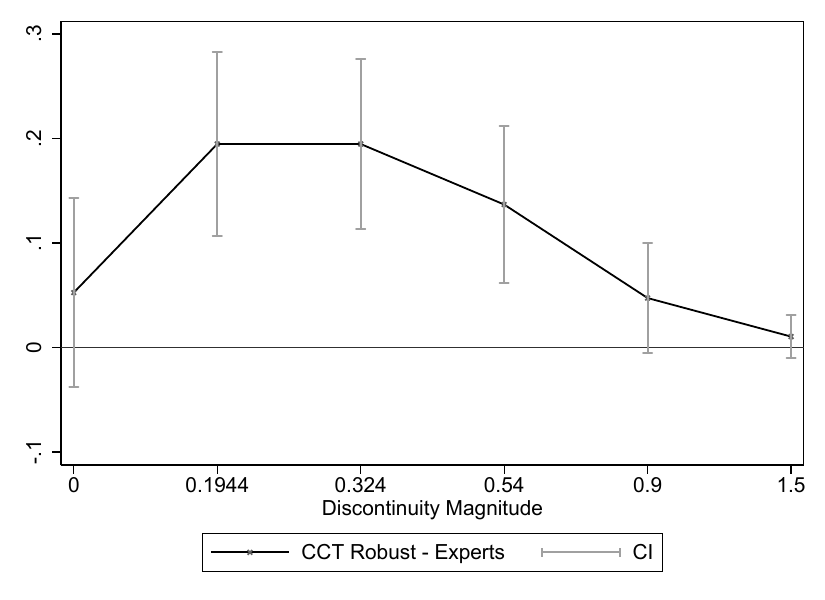}
        \includegraphics[width=3.25in]{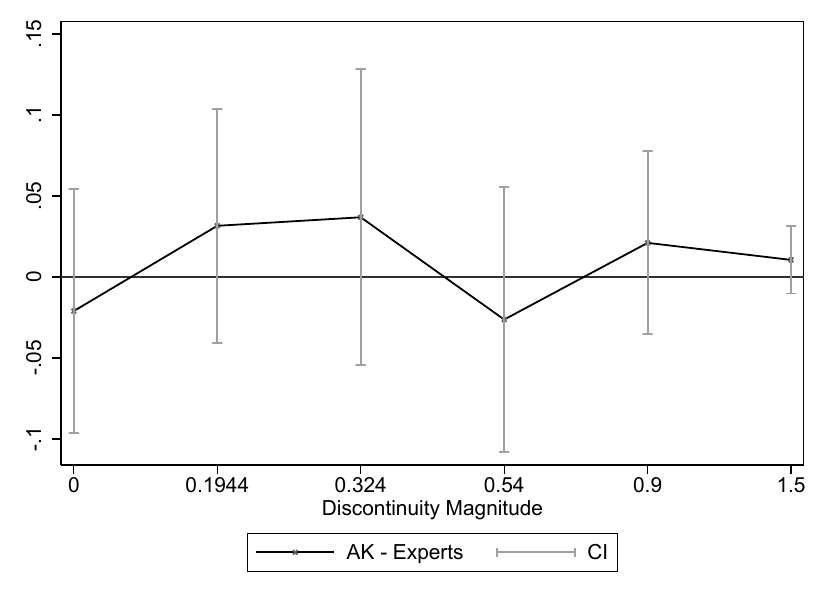}
        \begin{minipage}[t]{\textwidth}%
            \noindent {\scriptsize{}Notes: This figure plots the difference in the power functions between experts' visual inference and four econometric inference procedures from the left panel of Figure \ref{fig:RDD-Power-Experts-vs-All}. We compute two-way cluster-robust confidence intervals for these differences via a stacked regression where we account for the potential correlation between visual and econometric inferences at the dataset level (there are 88 datasets in total) and in visual inferences for the same individual across graphs---see Appendix \ref{subsec:two-way-clustering} for details. Piecewise Quintic uses a correctly specified regression model with global piecewise quintics above and below the treatment threshold and assuming homoskedasticity. IK is based on a local linear estimator using the IK bandwidth. CCT Robust is the default RDD inference procedure from CCT's }\texttt{\scriptsize{}rdrobust}{\scriptsize{}. AK uses the }\texttt{\scriptsize{}RDHonest}{\scriptsize{} procedure with the rule-of-thumb bound on each DGP's second derivative..}%
        \end{minipage}{\scriptsize\par}
    \end{figure}

    \begin{figure}[H]
        \caption{\label{fig:RDD-Experts-Vs-All-Size-Adjusted-Diff}Expert Visual vs Type I Error-Adjusted Econometric Inference Power Function Differences}
        \centering
        \includegraphics[width=3.25in]{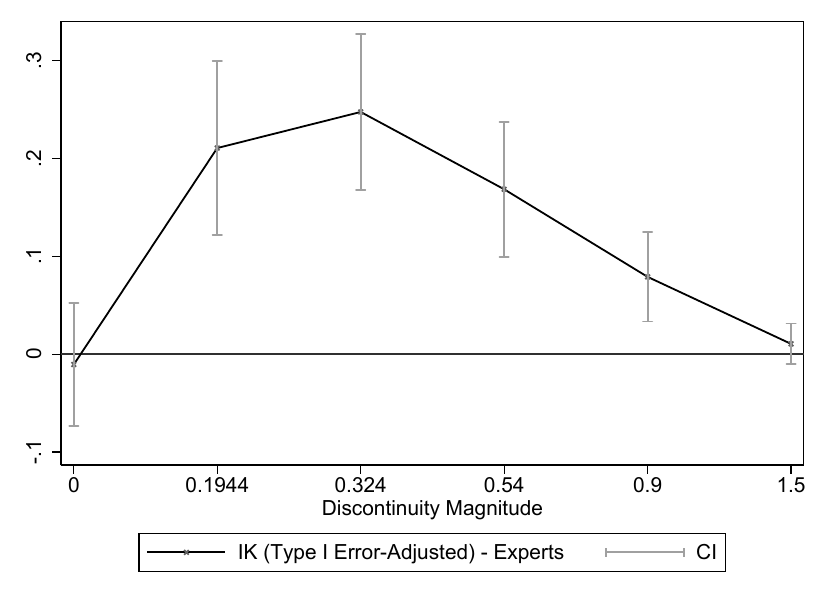}
        \includegraphics[width=3.25in]{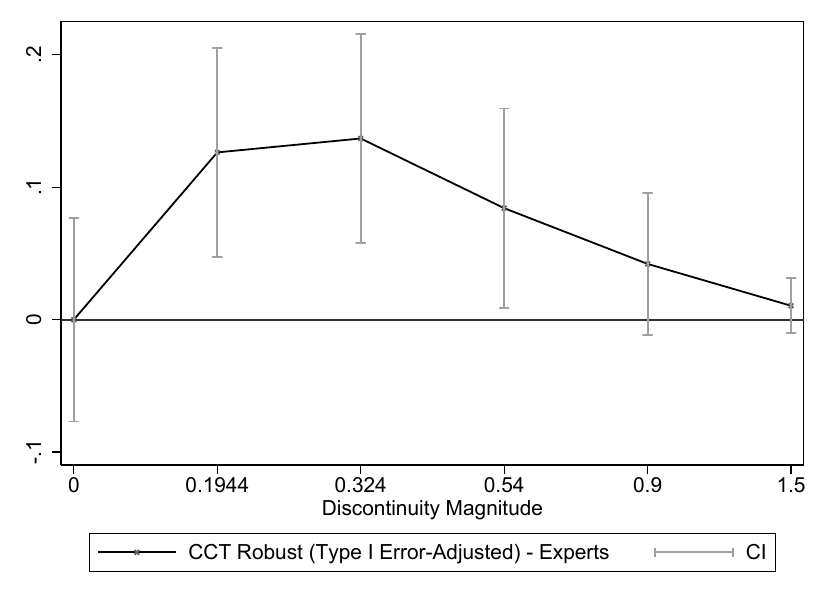}
        \begin{minipage}[t]{\textwidth}%
            \noindent {\scriptsize{}Notes: This figure plots the difference in the power functions between experts' visual inference and the type I error rate-adjusted IK and CCT procedures from the right panel of Figure \ref{fig:RDD-Power-Experts-vs-All}. IK is based on a local linear estimator using the IK bandwidth. CCT Robust is the default RDD inference procedure from CCT's }\texttt{\scriptsize{}rdrobust}{\scriptsize{}. We compute
            two-way cluster-robust confidence intervals for these differences via a stacked regression where we account for the potential correlation between visual and econometric inferences at the dataset level (there are 88 datasets in total) and in visual inferences for the same individual across graphs---see Appendix \ref{subsec:two-way-clustering} for details.}%
        \end{minipage}{\scriptsize\par}
    \end{figure}
\end{landscape}

\end{document}